\documentclass[10pt]{article}
\usepackage{euscript,authblk,bbm,color,slashed,enumerate,bm}

\usepackage[margin=1.0in]{geometry}
\usepackage{graphicx}
\usepackage{amsmath}
\usepackage{amsfonts}
\usepackage{amssymb}
\usepackage{mathrsfs}
\usepackage[normalem]{ulem}
\usepackage{caption}
\usepackage{subcaption}
\usepackage{amsthm}
\usepackage{textcomp}
\usepackage{hyperref}
\usepackage{slashed}
\usepackage{mathtools}

\newtheorem*{proposition*}{Proposition}
\newtheorem*{theorem*}{Theorem}
\newtheorem*{conjecture*}{Conjecture}
\newtheorem*{claim*}{Claim}
\newtheorem*{lemma*}{Lemma}
\newtheorem*{corollary*}{Corollary}
\newtheorem{theorem}{Theorem}[section]
\newtheorem{proposition}[theorem]{Proposition}

\newtheorem{lemma}[theorem]{Lemma}
\newtheorem{corollary}[theorem]{Corollary}

\newtheorem*{definition*}{Definition}
\newtheorem{definition}{Definition}[section]
\newtheorem*{assumption*}{Assumption}

\newtheorem*{remark*}{Remark}
\newtheorem{remark}{Remark}[section]

\newcommand{\R}{\mathbb{R}}
\newcommand{\s}{\mathbb{S}}

\newcommand{\ep}{\epsilon}
\newcommand{\ls}{\lesssim}
\newcommand{\f}{\frac}
\newcommand{\rd}{\partial}
\newcommand{\mfe}{\mathfrak e}
\newcommand{\mfm}{\mathfrak m}
\newcommand{\alp}{\alpha}
\newcommand{\bt}{\beta}
\newcommand{\nab}{\nabla}
\newcommand{\Omg}{\Omega}
\newcommand{\Om}{\Omega}
\newcommand{\CH}{\mathcal C\mathcal H^+}
\newcommand{\de}{\delta}
\numberwithin{equation}{section}
\numberwithin{theorem}{section}
\allowdisplaybreaks

\begin{document}

\title{The interior of dynamical extremal black holes in \\spherical symmetry}

\author[1]{Dejan Gajic\thanks{dejan.gajic@imperial.ac.uk}}
\author[2]{Jonathan Luk\thanks{jluk@stanford.edu}}
\affil[1]{\small  Department of Mathematics, South Kensington campus, Imperial College London, London SW7 2AZ, UK \vskip.1pc \ }
\affil[2]{\small  Department of Mathematics, Stanford University, Stanford CA 94305-2125, USA\vskip.1pc \ }

\maketitle

\begin{abstract}
We study the nonlinear stability of the Cauchy horizon in the interior of \textbf{extremal} Reissner--Nordstr\"om black holes under \textbf{spherical symmetry}. We consider the Einstein--Maxwell--Klein--Gordon system such that the charge of the scalar field is appropriately small in terms of the mass of the background extremal Reissner--Nordstr\"om black hole. Given spherically symmetric characteristic initial data which approach the event horizon of extremal Reissner--Nordstr\"om sufficiently fast, we prove that the solution extends beyond the Cauchy horizon in $C^{0,\f 12}\cap W^{1,2}_{loc}$, in contrast to the subextremal case (where generically the solution is $C^0\setminus (C^{0,\f 12}\cap W^{1,2}_{loc})$). In particular, there exist non-unique spherically symmetric extensions which are moreover solutions to the Einstein--Maxwell--Klein--Gordon system. Finally, in the case that the scalar field is chargeless and massless, we additionally show that the extension can be chosen so that the scalar field remains Lipschitz.
\end{abstract}

\section{Introduction}

In this paper, we initiate the study of the interior of dynamical extremal black holes. The Penrose diagram corresponding to maximal analytic extremal Reissner--Nordstr\"om and Kerr spacetimes is depicted in Figure~\ref{fig:fullspacetime}. In particular,  if one restricts to a globally hyperbolic subset with an (incomplete) asymptotically flat Cauchy hypersurface (cf.~the region $D^+(\Sigma)$ in Figure~\ref{fig:fullspacetime}), then these spacetimes possess smooth \emph{Cauchy horizons} $\CH$, whose stability property is the main object of study of this paper.

\begin{figure}
\centering{
\includegraphics[width=2in]{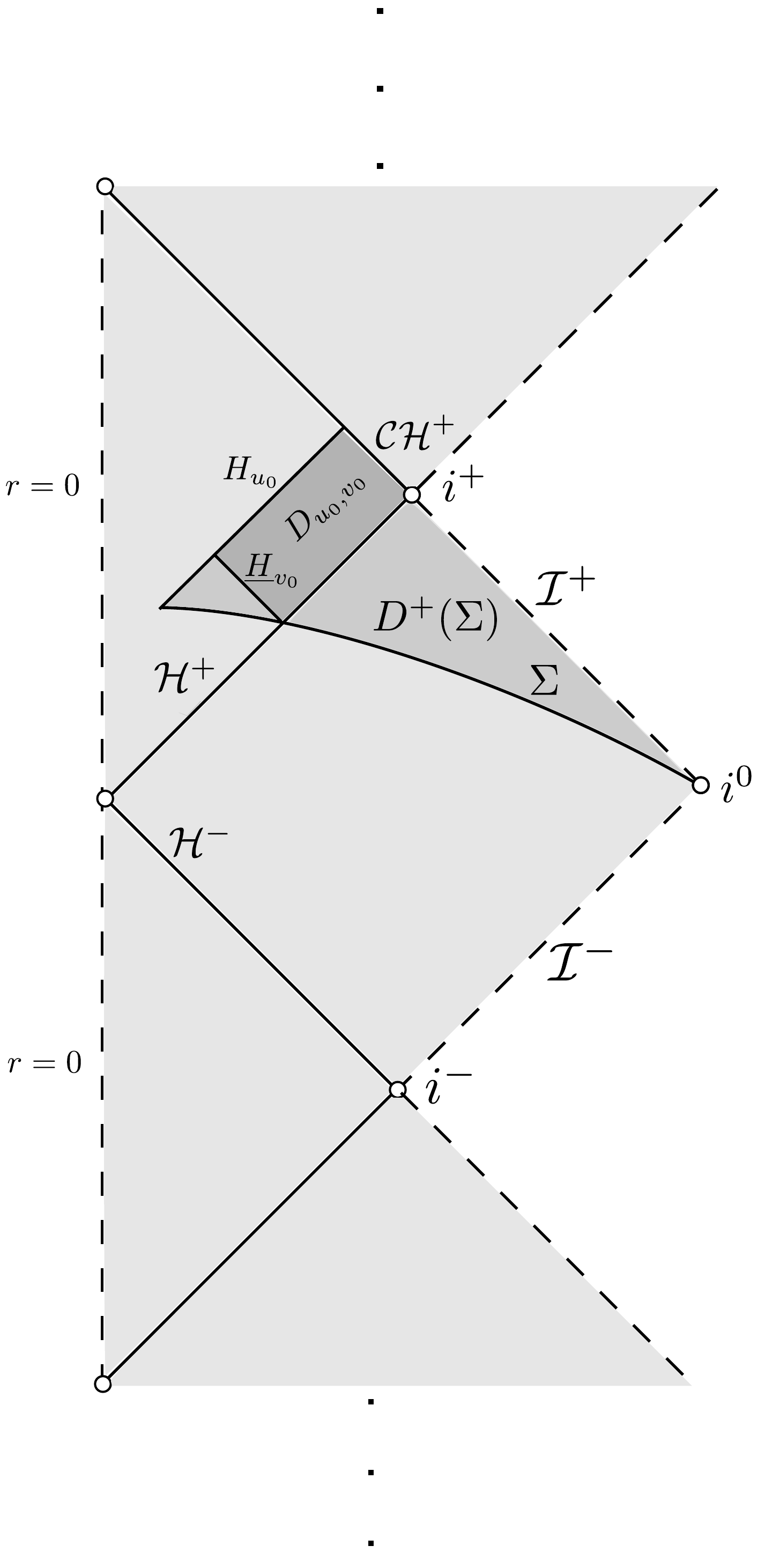}}
\caption{\label{fig:fullspacetime} Maximal analytically extended extremal Reissner--Nordstr\"om.}
\end{figure}

Since the pioneering work of Poisson--Israel \cite{PI1} and the seminal work of Dafermos \cite{D1, D2} in the spherically symmetric setting, we now have a rather complete understanding of the interior of dynamical black holes which approach \emph{subextremal} limits along the event horizon, at least regarding the \emph{stability} of the Cauchy horizons. The works \cite{CGNS2,D1,D2,D3,Fra,Hintz,Kommemi,LWeaknull} culminated in the recent work of Dafermos--Luk \cite{DL}, which proves the $C^0$ stability of the Kerr Cauchy horizon without any symmetry assumptions, i.e.~they show that whenever the exterior region of a black hole approaches a subextremal, strictly rotating Kerr exterior, then maximal Cauchy evolution can be extended across a non-trivial piece of Cauchy horizon as a Lorentzian manifold with continuous metric. Moreover, it is expected that for a generic subclass of initial data, the Cauchy horizon is an \emph{essential weak null singularity}, so that there is no extension beyond the Cauchy horizon as a weak solution to the Einstein equations (see \cite{D2,Gleeson,LO.instab,LO.interior,LS,VDM} for recent progress and discussions). 

On the other hand, much less is known about dynamical black holes which become extremal along the event horizon. Mathematically, the only partial progress was made for a related linear problem, namely the study of the linear scalar wave equation on extremal black hole backgrounds. For the linear scalar wave equation, the first author established \cite{Gajic, Gajic2} that in the extremal case, the Cauchy horizon is \emph{more stable} than its subextremal counterpart. In particular, the solutions to linear wave equations are not only bounded, as in the subextremal case, but they in fact obey higher regularity bounds which \emph{fail} in the subextremal case (see Section~\ref{sec:previous} for a more detailed discussion). Extrapolating from the linear result, it may be conjectured that in the interior of a black hole which approaches an extremal black hole along the event horizon, not only does the solution remain continuous up to the Cauchy horizon as in the subextremal case, but in fact there are non-unique extensions beyond the Cauchy horizon \emph{as weak solutions}. This picture, if true, would also be consistent with the numerical study of this problem by Murata--Reall--Tanahashi \cite{MRT}.

In this paper, we prove that this picture holds in a simple nonlinear setting. More precisely, we study the Einstein--Maxwell--Klein--Gordon system of equations with spherically symmetric initial data (see Section~\ref{sec:geometry} for further discussions on the system). We solve for a quintuple $(\mathcal M, g, \phi,A,F)$, where $(\mathcal M,g)$ is a Lorentzian metric, $\phi$ is a complex valued function on $\mathcal M$, $A$ and $F$ are real $1$- and $2$-forms on $\mathcal M$ respectively. The system of equations is as follows:
\begin{equation}\label{EMCSFS}
\left\{
\begin{aligned}
&Ric_{\mu\nu}-\f12 g_{\mu\nu} R=8\pi(\mathbb T^{(sf)}_{\mu\nu}+\mathbb  T^{(em)}_{\mu\nu}),\\
&\mathbb  T^{(sf)}_{\mu\nu}=\f 12 D_\mu\phi \overline{D_\nu\phi} + \f 12\overline{D_\mu\phi} D_\nu\phi-\f 12 g_{\mu\nu} ((g^{-1})^{\alp\beta}D_\alp\phi\overline{D_{\beta}\phi}+\mfm^2|\phi|^2),\\
&\mathbb  T^{(em)}_{\mu\nu}=(g^{-1})^{\alp\bt}F_{\mu\alp}F_{\nu\bt}-\f 14 g_{\mu\nu}(g^{-1})^{\alp\bt}(g^{-1})^{\gamma\sigma}F_{\alp\gamma}F_{\bt\sigma},\\
&(g^{-1})^{\alp\bt}D_\alp D_\bt\phi=\mfm^2\phi,\\
&F=dA,\\
&(g^{-1})^{\alpha\mu}\nab_\alpha F_{\mu\nu}=2\pi i \mfe (\phi \overline{D_\nu\phi} - \overline{\phi} D_\nu \phi).
\end{aligned}
\right.
\end{equation}
Here, $\nab$ denotes the Levi--Civita connection associated to the metric $g$, and $Ric$ and $R$ denote the Ricci tensor and the Ricci scalar, respectively. We also use the notation $D_{\alp} = \nab_\alp + i\mfe A_\alp$, and $\mfm\geq 0$, $\mfe\in \mathbb R$ are fixed constants. The extremal Reissner--Nordstr\"om solution (cf.~Section~\ref{sec:geometry}) is a special solution to \eqref{EMCSFS} with a vanishing scalar field $\phi$.

In the following we will restrict the parameters so that $|\mfe|$ is sufficiently small in terms of $M$. More precisely, we assume
\begin{equation}\label{par.con}
1-\left(10+5\sqrt{6}-3\sqrt{9+4\sqrt{6}}\right)|\mathfrak{e}|M> 0.
\end{equation}

Under the assumption \eqref{par.con}, our main result can be stated informally as follows (we refer the reader to Theorem~\ref{thm:main} for a precise statement):
\begin{theorem}\label{thm:main.intro}
Consider the characteristic initial value problem to \eqref{EMCSFS} with spherically symmetric smooth characteristic initial data on two null hypersurfaces transversely intersecting at a $2$-sphere. Assume that one of the null hypersurfaces is affine complete and that the data approach the event horizon of extremal Reissner--Nordstr\"om at a sufficiently fast rate. 

Then, the solution to \eqref{EMCSFS} arising from such data, when restricted to a sufficiently small neighborhood of timelike infinity (i.e.~a neighborhood of $i^+$ in Figure~\ref{fig:fullspacetime}), satisfies the following properties:
\begin{itemize}
\item It possesses a non-trivial Cauchy horizon.
\item The scalar field, the metric, the electromagnetic potential (in an appropriate gauge) and the charge can be extended in (spacetime) $C^{0,\f 12}\cap W^{1,2}_{loc}$ up to the Cauchy horizon. Moreover, the Hawking mass \eqref{Hawking.mass} can be extended continuously up to the Cauchy horizon.
\item The metric converges to that of extremal Reissner--Nordstr\"om towards timelike infinity and the scalar field approaches $0$ towards timelike infinity in an appropriate sense.
\end{itemize}
Moreover, the maximal globally hyperbolic solution is \textbf{future extendible} (non-uniquely) as a spherically symmetric solution to \eqref{EMCSFS}.
\end{theorem}

\begin{remark}[Solutions with regularity below $C^{2}$]\label{rmk:low.regularity}
The extensions of the solution we construct has regularity below (spacetime) $C^{2}$ and as such do not make sense as classical solutions. As is well-known, however, the Einstein equations admit a \textbf{weak formulation} which makes sense already if the metric is in (spacetime) $C^0 \cap W^{1,2}_{loc}$ and the stress-energy-momentum tensor is in spacetime $L^1_{loc}$ \cite{GerochTraschen}. The weak formulation can be recast geometrically as follows: given a smooth $(3+1)$-dimensional manifold $\mathcal M$, a $C^0_{loc} \cap W^{1,2}_{loc}$ Lorentzian metric $g$ and an $L^1_{loc}$ symmetric $2$-tensor $T$, we say that the Einstein equations $Ric(g) - \f 12 g R(g) = 8\pi T$ is satisfied weakly if for all smooth and compactly supported vector fields $X,Y$, 
$$\int_{\mathcal M} ((\nab_\mu X)^\mu (\nab_\nu Y)^\nu - (\nab_\mu X)^\nu (\nab_\nu Y)^\mu) = 8\pi \int_{\mathcal M} (T(X,Y) - \f12 g(X,Y) \mbox{tr}_g T).$$
It is easy to check that any classical solution is indeed a weak solution in the sense above. Moreover, the extensions that we construct in Theorem~\ref{thm:main.intro} have more than sufficient regularity to be interpreted in the sense above. 

However, in our setting we do not need to use the notion in \cite{GerochTraschen}. Instead, we introduce a stronger notion of solutions, defined on a quotient manifold for which we quotiented out the spherical symmetry; see Definition~\ref{def:strong.solutions} in Section~\ref{sec:extension}. This class of solutions --- even though they are not classical solutions --- should be interpreted as \textbf{strong solutions} (instead of just weak solutions) since a well-posedness theory can be developed for them\footnote{\label{fn:LR}In fact, in order to develop a well-posedness theory for strong solutions, one can even drop the assumption of spherical symmetry, but instead require additional regularity along the ``spherical directions'' with respect to an appropriately defined double null foliation gauge; see \cite{LR2} for details.}; see Section~\ref{sec:extension}.
\end{remark}
 
\begin{remark}[Contrast with the subextremal case]\label{rmk:compare}
Like in the subextremal case, the solution extends in $C^0$ to the Cauchy horizon. However, the $C^{0,\f 12}\cap W^{1,2}_{loc}$ extendibility, the finiteness of the Hawking mass, as well as the extendibility as spherically symmetric solution stand in \textbf{contrast} to the subextremal case. In particular, according to the results of \cite{LO.interior, LO.massinflation} (see also \cite{D2}), there are solutions which asymptote to subextremal Reissner--Nordstr\"om black holes in the exterior region such that the Hawking mass blows up at the Cauchy horizon, and the solution cannot be extended as a spherically symmetric solution to the Einstein--Maxwell--scalar field system\footnote{Though the estimates in \cite{LO.interior} strongly suggest that the scalar field ceases to be in $W^{1,2}_{loc}$ for any $C^0$ extension of the spacetime, this remains an open problem unless spherical symmetry is imposed. In particular, it is not known whether the solutions constructed in \cite{LO.interior, LO.massinflation} can be extended as weak solutions to the Einstein--Maxwell--scalar field system if no spherical symmetry assumption is imposed.}.
\end{remark}

\begin{remark}[Regularity of the metric and extensions as solutions to \eqref{EMCSFS}] The fact that we can extend the solutions beyond the Cauchy horizon is intimately connected to the regularity of the solutions up to the Cauchy horizon. In particular this relies on the fact the metric, the scalar field and the electromagnetic potential remain in (spacetime) $C^0\cap W^{1,2}_{loc}$. In fact, the solutions are at a level of regularity for which the Einstein equations are still locally \underline{well-posed}\footnote{Note that in general the Einstein equations are not locally well-posed with initial data only in $C^0\cap W^{1,2}$. Nevertheless, when there is spherical symmetry (away from the axis of symmetry), or at least when there is additional regularity in the spherical directions (cf.~Footnote~\ref{fn:LR}), one can indeed develop a local well-posedness theory with such low regularity.}. One can therefore construct extensions beyond the Cauchy horizon by solving appropriate characteristic initial value problems; see Section~\ref{sec:extension}.

In this connection, note that we emphasised in the statement of the theorem that the solution can be extended beyond the Cauchy horizon as a spherically symmetric solution to \eqref{EMCSFS}. The emphasis on the spherical symmetry of the extension is made mostly to contrast with the situation in the subextremal case (cf.~Remark~\ref{rmk:compare}). This should not be understood as implying that the extensions necessarily are spherically symmetric: In fact, with the bounds that we establish in this paper, one can in principle construct using the techniques in \cite{LR2} extensions (still as solutions to \eqref{EMCSFS}) without any symmetry assumptions (cf.~Footnote~\ref{fn:LR}).
\end{remark}

\begin{remark}[Assumptions on the event horizon]
The assumptions we impose on the event horizon are consistent with the expected late-time behavior of the solutions in the exterior region of the black hole, at least in the $\mfe=\mfm=0$ case if one extrapolates from numerical results \cite{MRT}. In particular, the transversal derivative of the scalar field is not required to decay along the event horizon, and is therefore consistent with the Aretakis instability \cite{Aretakis2}. Of course, in order to completely understand the structure of the interior, one needs to prove that the decay estimates along the event horizon indeed hold for general dynamical solutions approaching these extremal black holes. This remains an open problem.
\end{remark}

\begin{remark}[Range of parameters of $\mfe$ and $M$]
Our result only covers a limited range of parameters of the model; see \eqref{par.con}. This restriction comes from a Hardy-type estimate used to control the renormalised energy (cf.~Sections~\ref{sec:ideas} and \ref{sec:coercive.phi}) and we have not made an attempt to obtain the sharp range of parameters.
\end{remark}

\begin{remark}[The $\mfm=\mfe=0$ case and and higher regularity for $\phi$]
In the special case $\mfm=\mfe=0$, the analysis is simpler and we obtain a stronger result; namely, we show that the scalar field in fact is Lipschitz up to the Cauchy horizon (cf.~Theorem~\ref{thm:Lipschitz}). 
\end{remark}

\begin{remark}[The $\mfm\neq 0$, $\mfe\neq 0$ case]\label{rmk:cases}
While the result we obtain in the $\mfm\neq 0$, $\mfe\neq 0$ case is weaker, the general model allows for the charge of the Maxwell field to be non-constant, and serves as a better model problem for the stability of the extremal Cauchy horizon without symmetry assumptions. Another reason that we do not restrict ourselves to the simpler $\mfm=\mfe=0$ case is that in the $\mfm=\mfe=0$ case, extremal black holes do not naturally arise dynamically:
\begin{itemize}
\item There are no one-ended black holes with non-trivial Maxwell field with regular data on $\mathbb R^3$ since in that setting the Maxwell field necessarily blows up at the axis of symmetry.
\item In the two-ended case, given future-admissible\footnote{The future-admissibility condition can be thought of as an analogue of the physical ``no anti-trapped surface'' assumption in the one-ended case.} (in the sense of \cite{D3}) initial data, the solution always approaches \underline{subextremal} black holes in each connected component of the exterior region \cite{Kommemi, LO.interior}.
\end{itemize}
On the other hand, if $\mfe\neq 0$, then in principle there are no such obstructions\footnote{Nevertheless, it is an open problem to construct a dynamical black hole with regular data that settles down to an extremal black hole.}.
\end{remark}

\begin{remark}[Geometry of the black hole interior]
One feature of the black hole interior of extremal Reissner--Nordstr\"om is that it is free of radial trapped surfaces---a stark contrast to the subextremal case (where every sphere of symmetry is the black hole interior is trapped!). Let us note that this feature has sometimes been taken as the defining feature of spherically symmetric extremal black holes; see for instance \cite{Israel}. We will not use this definition in this paper, and when talking about ``extremal black holes'', we will only be referring to black holes which converge to a stationary extremal black hole along the event horizon as in Theorem~\ref{thm:main.intro}. Indeed, while our estimates imply that for the solutions in Theorem~\ref{thm:main.intro}, the geometry of the black hole interior is close to that of extremal Reissner--Nordstr\"om, it remains an open problem in the general case whether the black hole interior contains any radial trapped surface.\footnote{Note however that in the $\mfe=\mfm=0$ case, if we assume in addition that $\rd_U r<0$ everywhere along the event horizon, we can in principle modify the monotonicity argument of Kommemi \cite{Kommemi} in establishing the subextremality of two-ended black holes (see Remark~\ref{rmk:cases}) to show that the interior in the extremal case is free of radial trapped surfaces. Indeed, the argument of \cite{Kommemi} exactly proceeds by (1) showing that there are no interior trapped surfaces in the interior of extremal black holes and (2) establishing a contradiction with the future-admissibility condition. See also the appendix of \cite{LO.interior}.}
\end{remark}

The fact that the extremal Cauchy horizons are ``more stable'' than their subextremal counterparts can be thought of as related to the vanishing of the surface gravity in the extremal case. Recall that in both the extremal and subextremal charged Reissner--Nordstr\"om spacetime, there is a global infinite blue shift effect such that the frequencies of signals sent from the exterior region into the black hole are shifted infinitely to the blue~\cite{penrose1968battelle}. As a result, this gives rise to an instability mechanism. Indeed, Sbierski~\cite{Sbi.1} showed\footnote{This statement is technically not explicitly proven in \cite{Sbi.1}, but it follows from the result there together with routine functional analytic arguments.} that for the linear scalar wave equation on \emph{both} extremal and subextremal Reissner--Nordstr\"om spacetime, there exist \emph{finite energy} Cauchy data which give rise to solutions that are not $W^{1,2}_{loc}$ at the Cauchy horizon. On the other hand, as emphasised in \cite{Sbi.1}, this type of considerations do not take into account the \emph{strength} of the blue shift effect and do not give information on the behaviour of the solutions arising from more localised data. Heuristically, for localised data, one needs to quantify the amplification of the fields by a ``local'' blue shift effect at the Cauchy horizon, whose strength can be measured by the surface gravity. In this language, what we see in Theorem~\ref{thm:main.intro} is a manifestation of the vanishing of the local blue shift effect at the extremal limit.

This additional stability of the extremal Cauchy horizon due to the vanishing of the surface gravity may at the first sight seem to make the problem simpler than its subextremal counterpart. Ironically, from the point of view of the analysis, the local blue shift effect in the subextremal case in fact provides a way to prove stable estimates! This method fails at the extremal limit.

To further illustrate this, first note that in the presence of the blue shift effect, one necessarily proves \emph{degenerate} estimates. By exploiting the geometric features of the interior of subextremal black holes, the following can be shown: when proving degenerate energy-type estimates, by choosing the weights in the estimates appropriately, one can prove that the bulk spacetime integral terms (up to a small error) have a good sign, and can be used to control the error terms (cf.~discussions in the introduction of \cite{DL}). As a consequence, one can in fact obtain a stability result for a large class of nonlinear wave equations with a null condition, irrespective of the precise structure of the linearized system. This observation is also at the heart of the work \cite{DL}.

In the extremal case, however, it is not known how to obtain a sufficiently strong coercive bulk spacetime integral term when proving energy estimates. Moreover, if one naively attempts to control the spacetime integral error terms by using the boundary flux terms of the energy estimate and Gr\"onwall's lemma, one encounters a logarithmic divergence. To handle the spacetime integral error terms, we need to use more precise structures of the equations, and we will show that there is a cancellation in the weights appearing in the bulk spacetime error terms. This improvement of the weights then allows the bulk spacetime error terms to be estimated using the boundary flux terms and a suitable adaptation\footnote{In fact, using the smallness parameters in the problem, this will be implemented without explicitly resorting to Gr\"onwall's lemma.} of Gr\"onwall's lemma. In particular, we need to use the fact that (1) a renormalised energy can be constructed to control the scalar field and the Maxwell field \emph{simultaneously}, and that (2) the equations for the matter fields and the equations for the geometry are ``sufficiently decoupled'' (cf.~Section~\ref{sec:ideas}). These structures seem to be specific to the spherically symmetric problem: to what extent this is relevant to the general problem of stability of extremal Cauchy horizons without symmetry assumptions remains to be seen.

The study of the stability properties of subextremal Cauchy horizons is often motivated by the \emph{strong cosmic censorship conjecture}. The conjecture states that solutions arising from generic asymptotically flat initial data are \emph{inextendible} as suitably regular Lorentzian manifolds. In particular, the conjecture, if true, would imply that \emph{smooth} Cauchy horizons, which are present in both extremal and subextremal Reissner--Nordstr\"om spacetimes, are not generic in black hole interiors. As we have briefly discussed above, there are various results establishing this in the subextremal case; see for example \cite{D2,LO.interior,LO.massinflation,VDM}. In fact, one expects that generically, if a solution approaches a subextremal black hole at the event horizon, then the spacetime metric does not admit $W^{1,2}_{loc}$ extensions beyond the Cauchy horizon; see discussions in \cite{DL}. On the other hand, our result shows that at least in our setting, this does not occur for extremal Cauchy horizons. Nevertheless, since one expects that generic dynamical black hole solutions are non-extremal, our results, which only concern black holes that become extremal in the limit, are in fact irrelevant to the strong cosmic censorship conjecture. In particular, provided that extremal black holes are indeed non-generic as is expected, the rather strong stability that we prove in this paper does not pose a threat to cosmic censorship.

Finally, even though our result establishes the $C^{0,\f 12}\cap W^{1,2}_{loc}$ stability of extremal Cauchy horizons in spherical symmetry, it still leaves open the possibility of some higher derivatives of the scalar field or the metric blowing up (say, the $C^k$ norm blows up for some $k\in \mathbb N$). Whether this occurs or not for generic data remains an open problem.

\subsection{Previous results on the linear wave equation}\label{sec:previous}

In this section, we review the results established in \cite{Gajic, Gajic2} concerning the behaviour of solutions to the linear wave equation $\Box_g\phi=0$ in the interior of extremal black holes. The results concern the following cases:
\begin{itemize}
\item general solutions on extremal Reissner--Nordstr\"om,
\item general solutions on extremal Kerr--Newman \emph{with sufficiently small specific angular momentum},
\item \emph{axisymmetric} solutions on extremal Kerr.
\end{itemize}
In each of these cases, the following results are proven (in a region sufficiently close to timelike infinity):
\begin{enumerate}[(A)]
\item $\phi$ is bounded and continuously extendible up to the Cauchy horizon.
\item $\phi$ is $C^{0,\alp}$ up to the Cauchy horizon for all $\alp \in (0,1)$.
\item $\phi$ has finite energy and is $W^{1,2}_{loc}$  up to the Cauchy horizon.
\end{enumerate}

As we mentioned earlier, these results are in contrast with the subextremal case. (A) holds also for subextremal Reissner--Nordstr\"om and Kerr \cite{Fra, Hintz}, (B) is \textbf{false}\footnote{This result is not explicitly stated in the literature, but can be easily inferred given the sharp asymptotics for generic solutions in \cite{AAG} and the blowup result in \cite{D2} appropriately adapted to the linear setting.} on subextremal Reissner--Nordstr\"om (\cite{D2, AAG}) and (C) is \textbf{false} on both subextremal Reissner--Nordstr\"om and Kerr \cite{LO.instab, LS} (In fact, in subextremal Reissner--Nordstr\"om, generic solutions fail to be in $W^{1,p}_{loc}$ for all $p>1$; see \cite{D2, AAG, Gleeson}).

At this point, it is not clear whether the estimates in \cite{Gajic, Gajic2} are sharp. In the special case of spherically symmetric solutions on extremal Reissner--Nordstr\"om, \cite{Gajic} proves that the solution is in fact $C^1$ up to the Cauchy horizon. Moreover, if one assumes more precise asymptotics along the event horizon (motivated by numerics), then it is shown that spherically symmetric solutions are $C^2$.

Our results in the present paper can be viewed as an extension of those in \cite{Gajic} to a nonlinear setting. In particular, we show that even in the nonlinear (although only spherically symmetric) setting, $\phi$ still obeys (A) and (C), and satisfies (B) in the subrange $\alp\in (0,\f 12]$. Moreover, the metric components, the electromagnetic potential, and the charge, in appropriate coordinate systems and gauges, verify similar bounds.

\subsection{Ideas of the proof}\label{sec:ideas}

\textbf{Model linear problems.} The starting point of the analysis is to study \emph{linear systems of wave equations} on fixed extremal Reissner--Nordstr\"om background. A simple model of such a system is the following (where $a,b,c,d\in \mathbb R$):
\begin{equation}\label{eq:model}
\Box_{g_{eRN}}\phi = a \phi+ b\psi,\quad \Box_{g_{eRN}}\psi = c\psi+d\phi.
\end{equation}
It turns out that in the extremal setting, we still lack an understanding of solutions to such a model system in general. (This is in contrast to the subextremal case, where the techniques of \cite{Fra, DL} show that solutions to the analogue of \eqref{eq:model} are globally bounded for any fixed $a,b,c,d\in \mathbb R$.) 

Instead, we can only handle some subcases of \eqref{eq:model}. Namely, we need $a\geq 0$, $c\geq 0$ and $b=d=0$. Put differently, this means that we can only treat \emph{decoupled} Klein--Gordon equations with \emph{non-negative mass}. Remarkably, as we will discuss later, although the linearized equations of \eqref{EMCSFS} around extremal Reissner--Nordstr\"om are more complicated than decoupled Klein--Gordon equations with non-negative masses, one can find a structure in the equations so that the ideas used to handle special subcases of \eqref{eq:model} can also apply to the nonlinear problem at hand.

\textbf{Estimates for linear fields using ideas in \cite{Gajic}.} The most simplified case of \eqref{eq:model} the linear wave equation with zero potential $\Box_{g_{eRN}}\phi=0$. In the interior of extremal Reissner--Nordstr\"om spacetime, this has been treated by the first author in \cite{Gajic}. The work \cite{Gajic} is based on the vector field multiplier method, which obtains $L^2$-based energy estimates for the derivatives of $\phi$. The vector field multiplier method can be summarized as follows: Consider the stress-energy-momentum tensor
$$\mathbb T_{\mu\nu} = \rd_\mu\phi\rd_\nu\phi -\f 12 g_{\mu\nu}(g^{-1})^{\alp\bt}\rd_\alp\phi\rd_\bt\phi.$$
For a well-chosen vector field $V$, one can then integrate the following identity for the current $\mathbb T_{\mu\nu}V^\nu$
$$\nab^{\mu} (\mathbb T_{\mu\nu}V^\nu) = \f 12 \mathbb T_{\mu\nu}(\nab^{\mu} V^\nu+\nab^\nu V^\mu)$$
to obtain an identity relating a spacetime integral and a boundary integral.

When $V$ is casual, future-directed and Killing, the above identity yields a coercive conservation law. In the interior of extremal Reissner--Nordstr\"om $\rd_t = \f 12(\rd_v+\rd_u)$ (cf.~definition of $(u,v)$ coordinates in Section~\ref{sec:geometry}) is one such vector field. This vector field, however, is too degenerate near the event horizon and the Cauchy horizon, and one expects the corresponding estimates to be of limited use in a nonlinear setting. A crucial observation in \cite{Gajic} is that the vector field $V=|u|^2\rd_u+v^2\rd_v$ (in Eddington--Finkelstein double null coordinates, cf.~Section~\ref{sec:geometry}) can give a useful, stronger, estimate. More precisely, $V=|u|^2\rd_u+v^2\rd_v$ has the following properties:
\begin{enumerate}
\item $V$ is a non-degenerate vector field at both the event horizon and the Cauchy horizon. 
\item $V$ is causal and future-directed. Hence, together with 1., this shows that the current associated to $V$, when integrated over null hypersurfaces, corresponds to non-degenerate energy.
\item Moreover, although $V$ is not Killing, $\nab^{\mu} V^\nu+\nab^\nu V^\mu$ has a crucial cancellation\footnote{More precisely, 
$$\nab^{v} V^u+\nab^u V^v=-\f 12\Om^{-2}(\rd_u u^2 + \rd_v v^2 + \Om^{-2}(u^2\rd_u+v^2\rd_v)\Om^2)\ls \Om^{-2},$$
whereas each single term, e.g., $\Om^{-2}\rd_v v^2 \sim v\Om^{-2}$, behaves worse as $|u|,v\to \infty$. Without this cancellation, the estimate exhibits a logarithmic divergence.\label{footnote.cancellation}} so that the spacetime error terms can be controlled by the boundary integrals. 
\end{enumerate}
These observations allow us to close the estimate and to obtain non-degenerate $L^2$ control for the derivatives of $\phi$. Furthermore, the boundedness of such energy implies that
\begin{equation}\label{eq:phi.intro.ptwise}
|\phi|(u,v) \ls \mbox{Data}(v) + |u|^{-\f 12},
\end{equation}
where, provided the data term decays in $v$, $\phi \to 0$ as $|u|,v\to \infty$. Moreover, using the embedding $W^{1,2}\hookrightarrow C^{0,\f 12}$ in $1$-D, we also conclude that $\phi \in C^{0,\f 12}$.

In spherical symmetry, it is in fact possible to control the solution to the linear wave equation up to the Cauchy horizon using only the method of characteristics\footnote{In fact, as in shown in \cite{Gajic}, the method of characteristics, when combined with the energy estimates, yields more precise estimates when the initial data are assumed to be spherically symmetric. While \cite{Gajic} does not give a proof of the estimates in the spherically symmetric case purely based on the method characteristics, such a proof can be inferred from the proof of Theorem~\ref{thm:Lipschitz} in Section~\ref{sec:imp.est}.}. Here, however, there is an additional twist to the problem. We will need to control solutions to the \emph{Klein--Gordon} equation with \emph{non-zero mass}\footnote{As we will discuss below, the need to consider the Klein--Gordon equation with non-zero mass stems not only from our desire to include mass in the matter field in \eqref{EMCSFS}, but when attempting to control the metric components, one naturally encounters a Klein--Gordon equation with positive mass.}:
$$\Box_{g_{eRN}}\phi = \mfm^2\phi.$$
For this scalar equation, however, whenever the mass is non-vanishing, using the method of characteristics and na\"ive estimates leads to potential logarithmic divergences. Nonetheless, if $\mfm^2>0$, the argument above which makes use of the vector field multiplier method can still be applied. In this case, one defines instead the stress-energy-momentum tensor as follows:
$$\mathbb T_{\mu\nu} = \rd_\mu\phi\rd_\nu\phi -\f 12 g_{\mu\nu}((g^{-1})^{\alp\bt}\rd_\alp\phi\rd_\bt\phi+\mfm^2\phi^2).$$
As it turns out, the observations 1., 2., and 3. still hold in the $\mathfrak{m}^2>0$ case for the vector field $V=|u|^2\rd_u+v^2\rd_v$. In particular, there is a crucial cancellation in the bulk term as above, which removes the logarithmically non-integrable term and allows one to close the argument. Again, this consequently yields also decay estimates and $C^{0,\f 12}$ bounds for $\phi$.

Let us recap what we have achieved for the model problem \eqref{eq:model}. The discussions above can be used to deal fully with the case $a\geq 0$, $c\geq 0$ and $b=d=0$. If $a<0$ or $c<0$, one still has a cancellation in the bulk spacetime term, but the boundary terms are not non-negative. If, on the other hand, $b\neq 0$ or $d\neq 0$, then in general one sees a bulk term which is exactly borderline and leads to logarithmic divergence.

\textbf{Renormalised energy estimates for the matter fields.} In order to attack our problem at hand, the first step is to understand the propagation of the matter field even without coupling to gravity. In other words, we need to control the solution to the Maxwell--charged Klein--Gordon system in the interior of \emph{fixed} extremal Reissner--Nordstr\"om. (A special case of this, when $\mfm=\mfe=0$ is exactly what has been studied in \cite{Gajic}.) 

One difficulty that arises in controlling the matter fields is that when $\mfe\neq 0$, the energy estimates for the scalar field couple with estimates for the Maxwell field. If one na\"ively estimates each field separately, while treating the coupling as error terms, one encounters logarithmically divergent terms similar to those appearing when controlling \eqref{eq:model} for $b\neq 0$ or $d\neq 0$. Instead, we prove coupled estimates for the scalar field and the Maxwell field simultaneously. 

In order to prove coupled energy estimates, a natural first attempt would be to use the full stress-energy-momentum tensor (i.e.~the sum $\mathbb T=\mathbb T^{(sf)}+\mathbb T^{(em)}$) and consider the current $\mathbb T_{\mu\nu}V^\nu$, where $V=|u|^2\rd_u+v^2\rd_v$ as in \cite{Gajic}. However, since the charge is expected to asymptote to a \emph{non-zero} value (as it does initially along the event horizon), this energy is \emph{infinite}! 

Instead, we \emph{renormalise} the energy to take out the infinite contribution from the background charge. We are then faced with two new issues:
\begin{itemize}
\item the renormalised energy is not manifestly non-negative
\item additional error terms are introduced.
\end{itemize}
Here, it turns out that one can use a Hardy-type inequality to show that the renormalised energy is coercive. (This is the step for which we need a restriction on the parameters of the problem.) Moreover, the additional spacetime error terms that are introduced in the energy estimates also exhibit the cancellation described in Footnote~\ref{footnote.cancellation} on page~\pageref{footnote.cancellation}.

\textbf{Estimates for the metric components.} Having understood the uncoupled Maxwell--Klein--Gordon system, we now discuss the problem where the Maxwell--Klein--Gordon system is coupled with the Einstein equations. First, we write the metric in \emph{double null} coordinates:
$$g = -\Om^2(u,v) \, dudv + r^2(u,v) \sigma_{\mathbb S^2},$$
where $\sigma_{\mathbb S^2}$ is the standard round metric on $\mathbb S^2$ (with radius $1$). In such a gauge, the metric components $r$ and $\Om$ satisfy nonlinear wave equations with $\phi$ and $\rd\phi$ as sources. In addition, $r$ satisfies the Raychaudhuri equations, which can be interpreted as constraint equations.

We control $r$ directly using the method of characteristics. As noted before, using the method of characteristics for wave equations with non-trivial zeroth order terms\footnote{The wave equation for $r$ indeed has such zeroth order terms, cf.~\eqref{eq:transpeqr}.} leads to potentially logarithmically divergent terms. To circumvent this, we use both the wave equation and the Raychaudhuri equations satisfied by $r$: using different equations in different regions of spacetime, one can show that using the method of characteristics that
$$|r-M|\ls v^{-1}+|u|^{-1}, \quad |\rd_v r|\ls v^{-2},\quad |\rd_u r|\ls |u|^{-2}.$$

For $\Omega$, instead of controlling it directly, we bound the difference $\log\Omega - \log \Omega_0$, where $\Omega_0$ corresponds to the metric component of the background extremal Reissner--Nordstr\"om spacetime. We will control it using the wave equation satisfied by $\Omega$ (cf.~\eqref{eq:waveqOmega}). Again, as is already apparent in the discussion of \eqref{eq:model}, to obtain wave equation estimates, we need to use the structure of the equation. Using the estimates for $\phi$ and $r$, the equation for $\log\Omega - \log \Omega_0$ can be thought of as follows (modulo terms that are easier to deal with and are represented by $\dots$):
\begin{equation}\label{Om.Klein.Gordon}
\rd_u \rd_v \log \f{\Om}{\Om_0} = -\f 14 M^2 (e^{\log \Om^2}-e^{\log \Om_0^2})+\dots 
\end{equation}
Thus, when $\Om$ is close to $\Om_0$, \eqref{Om.Klein.Gordon} can be viewed as a nonlinear perturbation of the Klein--Gordon equation with positive mass, which is moreover essentially decoupled from the other equations. Hence, as long as we can control the error terms and justify the approximation \eqref{Om.Klein.Gordon}, we can handle this equation using suitable modifications of the ideas discussed before. In particular, an appropriate modification of \eqref{eq:phi.intro.ptwise} implies that $\Om \to \Om_0$ in a suitable sense as $|u|,v\to \infty$.

Finally, revisiting the argument for the energy estimates for Maxwell--Klein--Gordon, one notes that it can in fact be used to control solutions to the Maxwell--Klein--Gordon system on a \emph{dynamical} background such that $r$ and $\log \Om$ approach their Reissner--Nordstr\"om values with a sufficiently fast polynomial rate as $|u|,v\to \infty$. In particular, the estimates we described above for the metric components is sufficient for us to set up a bootstrap argument to simultaneously control the scalar field and the geometric quantities.

Note that in terms of regularity, we have closed the problem at the level of the (non-degenerate) $L^2$ norm of first derivatives of the metric components and scalar field. As long as $r>0$, this is the level of regularity for which well-posedness holds in spherical symmetry. It follows that we can also construct an extension which is a solution to \eqref{EMCSFS}.

\subsection{Structure of the paper}

The remainder of the paper is structured as follows. In \textbf{Section~\ref{sec:geom}}, we will introduce the geometric setup and discuss \eqref{EMCSFS} in spherical symmetry. In \textbf{Section~\ref{sec:geometry}}, we discuss the geometry of the interior of the extremal Reissner--Nordstr\"om black hole. In \textbf{Section~\ref{sec:data}}, we introduce the assumptions on the characteristic initial data. In \textbf{Section~\ref{sec:main.thm}}, we give the statement of the main theorem (Theorem~\ref{thm:main}, see also Theorem~\ref{thm:Lipschitz}). In \textbf{Section~\ref{sec:bootstrap}}, we begin the proof of Theorem~\ref{thm:main} and set up the bootstrap argument. In \textbf{Section~\ref{sec:pointwise}}, we prove the pointwise estimates. In \textbf{Section~\ref{sec:energy}}, we prove the energy estimates. In \textbf{Section~\ref{sec:stability}}, we close the bootstrap argument and show that the solution extends up to the Cauchy horizon. In \textbf{Section~\ref{sec:extension}}, we complete the proof of Theorem~\ref{thm:main} by constructing a spherically symmetric solution which extends beyond the Cauchy horizon. In \textbf{Section~\ref{sec:imp.est}}, we prove additional estimates in the case $\mfm=\mfe=0$.

\subsection*{Acknowledgements} Part of this work was carried out when both authors were at the University of Cambridge. J. Luk is supported by a Sloan fellowship, a Terman fellowship, and the NSF grant DMS-1709458.

\section{Geometric preliminaries}\label{sec:geom}

\subsection{Class of spacetimes}\label{doublenull}
In this paper, we consider \emph{spherically symmetric spacetimes} $(\mathcal M,g)$ with $\mathcal M = \mathcal Q\times \mathbb S^2$ such that the metric $g$ takes the form
$$g=g_{\mathcal Q}+r^2(d\theta^2+\sin^2\theta\,d\varphi^2),$$
 where $(\mathcal Q,g_{\mathcal Q})$ is a smooth $(1+1)$-dimensional Lorentzian spacetime and $r:\mathcal Q\to \mathbb R_{>0}$ is smooth and can geometrically be interpreted as the area radius of the orbits of spherical symmetry. We assume that $(\mathcal Q,g_{\mathcal Q})$ admits a global double null foliation\footnote{Note that for sufficiently regular $g_{\mathcal Q}$, the metric can always be put into double null coordinates \emph{locally}. Hence the assumption is only relevant for global considerations. We remark that the interior of extremal Reissner--Nordstr\"om spacetimes can be written (globally) in such a system of coordinates (see Section~\ref{sec:geometry}) and so can spacetimes that arise from spherically symmetric perturbations of the interior of extremal Reissner--Nordstr\"om, which we consider in this paper.}, so that we write the metric $g$ in double null coordinates as follows:
\begin{equation*}
g=-\Omega^2(u,v)dudv+r^2(u,v)(d\theta^2+\sin^2\theta\,d\varphi^2),
\end{equation*}
for some smooth and strictly positive function $\Om^2$ on $\mathcal Q$.

\subsection{The Maxwell field and the scalar field}

We will assume that both the Maxwell field $F$ and the scalar field $\phi$ in \eqref{EMCSFS} are spherically symmetric. For $\phi$, this means that $\phi$ is constant on each spherical orbit, and can be thought of as a function on $\mathcal{Q}$.

For the Maxwell field, spherical symmetry means that there exists a function $Q$ on $\mathcal{Q}$, so that the Maxwell field $F$ takes the following form
$$F = \f{Q}{2({\bm\pi}^* r)^2}{\bm\pi}^*(\Om^2 du\wedge dv),$$
where ${\bm \pi}$ denotes the projection map ${\bm\pi}:\mathcal M\to \mathcal Q$. We will call $Q$ the \emph{charge} of the Maxwell field.

\subsection{The system of equations}

In this subsection, we write down the symmetry-reduced equations in a double null coordinate system as in Section~\ref{doublenull} (see \cite{Kommemi} for details). %It is useful to keep in mind that in geometrised units ($G=c=1$):
%\begin{align*}
%[\phi]=&\:[A_{\mu}]=[L^0],\\
%[M]=&\:[Q]=[r]=[u]=[v]=[L^1],\\
%[\mathfrak{e}]=&\:[\mathfrak{m}]=[L^{-1}],
%\end{align*}
%where $L$ is a length.
Before we write down the equations, we introduce the following notation for the covariant derivative operator with respect to the 1-form $A$:
\begin{equation*}
D_{\mu}\phi=\partial_{\mu}\phi+i\mathfrak{e} A_{\mu}\phi.
\end{equation*}

\subsubsection{Propagation equations for the metric components}

\begin{align}
\label{eq:transpeqr}
r\partial_u\partial_v r=&-\frac{1}{4}\Omega^2-\partial_ur \partial_vr+\mathfrak{m}^2\pi r^2 \Omega^2 |\phi|^2+\frac{1}{4}\Omega^2 r^{-2} Q^2,\\
\label{eq:waveqOmega}
r^2\partial_u\partial_v\log \Omega=&-2\pi r^2(D_u\phi \overline{D_v\phi}+\overline{D_u\phi}D_v\phi)-\frac{1}{2}\Omega^2 r^{-2} Q^2+\frac{1}{4}\Omega^2+\partial_ur\partial_vr.
\end{align}

\subsubsection{Propagation equations for the scalar field and electromagnetic tensor}
\begin{align}
\label{eq:phi1}
D_uD_v\phi+D_v D_u\phi=&-\frac{1}{2}\mathfrak{m}^2\Omega^2\phi-2r^{-1}(\partial_ur D_v\phi+\partial_v r D_u \phi),\\
\label{eq:phi2}
D_uD_v\phi-D_vD_u \phi=&\:\frac{1}{2} r^{-2}\Omega^2i \mathfrak{e}Q\cdot \phi,\\
\label{eq:Q1}
\partial_uQ=&\:2\pi i r^2\mathfrak{e} (\phi \overline{D_u\phi}-\overline{\phi}D_u\phi),\\
\label{eq:Q2}
\partial_vQ=&\:-2\pi i r^2\mathfrak{e} (\phi \overline{D_v\phi}-\overline{\phi}D_v\phi).
\end{align}
Furthermore, we can express
\begin{equation}\label{QintermsofA}
Q=2r^2\Omega^{-2}(\partial_uA_v-\partial_vA_u).
\end{equation}

\subsubsection{Raychaudhuri's equations}
\begin{align}
\label{eq:raychu}
\partial_u(\Omega^{-2}\partial_u r)=&-4\pi r\Omega^{-2}|D_u\phi|^2,\\
\label{eq:raychv}
\partial_v(\Omega^{-2}\partial_v r)=&-4\pi r\Omega^{-2}|D_v\phi|^2.
\end{align}

\subsection{Hawking mass}

Define the Hawking mass $m$ by
\begin{equation}\label{Hawking.mass}
m :=\f{r}2 (1-g_{\mathcal Q}(\nabla r,\nabla r)) = \f{r}2 \left(1+\f{4\rd_u r\rd_v r}{\Om^2}\right).
\end{equation}
By \eqref{eq:transpeqr}, \eqref{eq:raychu} and \eqref{eq:raychv},
%\begin{equation}
%\begin{split}
%\rd_u m = &\f{m \rd_u r}{r} + 2 \f{\rd_u r}{\Om^2}(-\frac{1}{4}\Omega^2-\partial_ur \partial_vr+\mathfrak{m}^2\pi r^2 \Omega^2 |\phi|^2+\frac{1}{4}\Omega^2 r^{-2} Q^2) - 8\pi r^2 (\rd_v r)\Omega^{-2}|D_u\phi|^2 \\
%= &\f{m \rd_u r}{r} -\f 12 (\rd_u r) (1+\f{4\partial_ur \partial_vr}{\Omg^2})\\
%&-8\pi \f{r^2(\rd_V r)}{\widetilde{\Om}^2}|D_u\phi|^2+2(\rd_u r)\mfm^2 \pi r^2|\phi|^2+\f 12 \f{(\rd_u r)Q^2}{r^2}
%\end{split}
%\end{equation}
\begin{equation}\label{eq:dum}
\rd_u m = -8\pi \f{r^2(\rd_v r)}{\Om^2}|D_u\phi|^2+2(\rd_u r)\mfm^2 \pi r^2|\phi|^2+\f 12 \f{(\rd_u r)Q^2}{r^2},
\end{equation}
\begin{equation}\label{eq:dvm}
\rd_v m = -8\pi \f{r^2(\rd_u r)}{\Om^2}|D_v\phi|^2+2(\rd_v r)\mfm^2 \pi r^2|\phi|^2+\f 12 \f{(\rd_v r)Q^2}{r^2}.
\end{equation}

\subsection{Global gauge transformations}\label{sec:global.gauge}
Consider the following global gauge transformation induced by the function $\chi: \mathcal{D}\to \R$, $\mathcal{D}\subset \R^2$:
\begin{align*}
\widetilde{\phi}(u,v)=&\:e^{-i\mathfrak{e}\chi(u,v)}\phi(u,v),\\
\widetilde{A}_{\mu}(u,v)=&\:A_{\mu}(u,v)+\partial_{\mu}\chi(u,v),
\end{align*}
with $\mu=u,v$. Let us denote $\widetilde{D}=d+i\mathfrak{e} \widetilde{A}$. Then,
\begin{equation*}
\widetilde{D}_{\mu}\widetilde{\phi}=e^{-i\mathfrak{e}\chi}D_{\mu}{\phi}.
\end{equation*}
As a result we conclude that the following norms are (globally) \emph{gauge-invariant}: $|\phi|=|\widetilde{\phi}|$ and $|D_{\mu}\phi|=|\widetilde{D}_{\mu}\widetilde{\phi}|$.

In most of this paper, the choice of gauge will not be important. We will only explicitly choose a gauge when discussing local existence or when we need to construct an extension of the solution. Instead, most of the time we will estimate the gauge invariant quantities $|\phi|$ and $|D_{\mu}\phi|$. For this purpose, let us note that we have the following estimates regarding these quantities:
\begin{lemma}
\label{lm:fundthmcalc}
The following estimates hold:
\begin{equation}
\label{eq:fundthmcalcphi1}
|\phi|(u,v)\leq |\phi|(u_1,v)+ \int_{u_1}^u |{D}_u {\phi}|(u',v)\,du'
\end{equation}
and
\begin{equation}
\label{eq:fundthmcalcphi2}
|\phi|(u,v)\leq |\phi|(u,v_1)+ \int_{v_1}^v |{D}_v {\phi}|(u,v')\,dv'.
\end{equation}
\end{lemma}

\begin{proof}
We can always pick $\chi$ such that  $A_u=0$ and $\widetilde{D}_u\widetilde{\phi}=\partial_u\widetilde{\phi}$. This fact, together with the fundamental theorem of calculus and the gauge-invariance property above, imply 
\begin{equation*}
|\phi|(u,v)=|\widetilde{\phi}|(u,v)\leq |\widetilde{\phi}|(u_1,v) + \int_{u_1}^u |\widetilde{D}_u \widetilde{\phi}|(u',v)\,du'= |\phi|(u_1,v)+ \int_{u_1}^u |{D}_u {\phi}|(u',v)\,du',
\end{equation*}
which implies \eqref{eq:fundthmcalcphi1}.

Similarly, by choosing $\chi$ such that $A_v=0$, we obtain \eqref{eq:fundthmcalcphi2}.\qedhere
\end{proof}

\section{Interior of extremal Reissner--Nordstr\"om black holes}\label{sec:geometry}

The \emph{interior region of the extremal Reissner--Nordstr\"om solution} with mass $M>0$ is the Lorentzian manifold $(\mathcal M_{eRN},g_{eRN})$, where $\mathcal M_{eRN} = (0,M)_r\times (-\infty,\infty)_t\times \mathbb S^2$ and the metric $g_{\textnormal{eRN}}$ in the $(t,r,\theta,\varphi)$ coordinate system is given by
$$g_{\textnormal{eRN}} = -\Omg_0^2 dt^2 + \Omg_0^{-2} dr^2 + r_0^2(d\theta^2+\sin^2\theta\,d\varphi^2),$$
where 
$$\Omg_0 = \left(1-\f{M}{r_0}\right).$$

We define the Eddington--Finkelstein $r^*$ coordinate (as a function of $r$) by
\begin{equation}\label{def:r*}
r^*= \f{M^2}{M-r}+2M\log (M-r) + r,
\end{equation}
and define the Eddington--Finkelstein double-null coordinates by
\begin{equation}\label{def:uv}
u= t-r^*,\quad v =t+r^*.
\end{equation}
In Eddington--Finkelstein double-null coordinates $(u,v,\theta,\varphi)$, the metric takes the form as in Section~\ref{doublenull}:
\begin{equation*}
g_{\textnormal{eRN}}=-\Omega^2_0(u,v)\,dudv+r_0^2(u,v)(d\theta^2+\sin^2\theta\,d\varphi^2),
\end{equation*}
where $r_0$ is defined implicitly by \eqref{def:r*} and \eqref{def:uv} and $\Omega^2_0(u,v) = (1-\f{M}{r_0(u,v)})^2$.

For the purpose of this paper, we do not need the explicit expressions for $r_0$ and $\Om_0$ as functions of $(u,v)$, but it suffices to have some simple estimates. Since we will only be concerned with the region of the spacetime close to timelike infinity $i^+$ (see Figure~\ref{fig:fullspacetime})\footnote{Formally, it is the ``$2$-sphere at $u=-\infty$, $v=\infty$''.}, we will assume $v\geq 1$ and $u\leq -1$. In this region, we have the following estimates (the proof is simple and will be omitted): 
\begin{lemma}\label{RN.est}
For $v\geq 1$ and $u\leq -1$, there exists $C>0$ (depending on $M$) such that for $v\geq 1$ and $u\leq -1$,
$$|r_0-M|(u,v)\leq \f{C}{(v+|u|)},\quad |\rd_v r_0|(u,v)+|\rd_u r_0|(u,v)\leq\f{C}{(v+|u|)^2}.$$
Given any $\beta>0$, we can find a constant $C_{\beta}>0$ (depending on $M$ and $\beta$) such that for $v\geq 1$ and $u\leq -1$,
\begin{equation}\label{eq:Om0.est}
\left|\Omega_0-\frac{2M}{v+|u|}\right|(u,v)\leq C_{{\beta}}(v+|u|)^{-2+{\beta}},
\end{equation}
and
\begin{equation}
\label{eq:weightestOmega0}
|\partial_v(v^2\Omega_0^2)+\partial_u(u^2\Omega_0^2)|(u,v)\leq C_{{\beta}} (v+|u|)^{-2+{\beta}}.
\end{equation}
\end{lemma}

\subsection{Regular coordinates}

We would like to think of $\mathcal M_{eRN}$ as as having the ``event horizon'' and the ``Cauchy horizon'' as null boundaries, which are formally the boundaries $\{u=-\infty\}$ and $\{v=\infty\}$ respectively. To properly define them, we will introduce double null coordinate systems which are regular at the event horizon and at the Cauchy horizon respectively. We will also use these coordinate systems later in the paper
\begin{itemize}
\item to pose the characteristic initial value problem near the event horizon; and
\item to extend the solution up to the Cauchy horizon.
\end{itemize}

\subsubsection{Regular coordinates at the event horizon}\label{sec:coord.EH}
Define $U$ by the relation
\begin{equation}\label{U.def}
\f{dU}{du} = \Omega_0^2(u,1)=\left(1-\frac{M}{r_0(u,1)}\right)^2,\quad U(-\infty) = 0.
\end{equation}
By Lemma \ref{RN.est}, there exists a constant $C$ (depending on $M$), so that we can estimate
\begin{equation}
\label{U.behaviour}
0\leq \f{dU}{du}\leq C(1+|u|)^{-2}.
\end{equation}

Define the \emph{event horizon} as the boundary $\{U=0\}$. We will abuse notation to denote the event horizon as both the boundary in the quotient manifold $\{(U,v):U=0\}\subset \mathcal Q$ and the original manifold $\{(U,v):U=0\}\times \s^2 \subset \mathcal M_{eRN}$ (cf.\ Section~\ref{sec:geom}).

After denoting $u(U)$ as the inverse of $u\mapsto U$, we abuse notation to write $r_0(U,v)=r_0(u(U),v)$ and $\hat{\Omega}_0$ is defined by
$$\hat{\Omega}_0^2(U,v) =  \Omega_0^{-2}(u(U),1) \Omega_0^2(u(U),v),$$
the extremal Reissner--Nordstr\"om metric takes the following form in the $(U,v,\theta,\varphi)$ coordinate system
\begin{equation*}
g_{\textnormal{eRN}}=-\hat{\Omega}^2_0(U,v)\,dU dv+r_0^2(U,v)(d\theta^2+\sin^2\theta\,d\varphi^2),
\end{equation*}
In particular, by \eqref{eq:Om0.est}, it holds that
\begin{equation}\label{hat.Om0}
\hat{\Omega}_0(0,v) = 1
\end{equation}
for all $v$. Additionally, we have, for all $v$,
$$r_0(0,v) = M.$$
Hence, in the $(U,v,\theta,\varphi)$ coordinate system, $\hat{\Omega}_0(U,v)$ and $r_0(U,v)$ extend continuously (in fact smoothly) to the event horizon. Moreover, for every $v\geq 1$ and $u(U)\leq -1$, $\hat{\Omega}_0^2(U,v)$ is bounded above and below as follows:
\begin{equation}\label{Omhat0.bounds}
\f{2}{v+1}\leq \hat{\Om}_0(U,v) \leq 1.
\end{equation}

\subsubsection{Regular coordinates at the Cauchy horizon}\label{sec:coord.CH}
Define $V$ by the relation
\begin{equation}\label{V.def}
\f{dV}{dv} = \Omega_0^{2}(-1,v)=\left(1-\frac{M}{r_0(-1,v)}\right)^2 ,\quad V(\infty) = 0.
\end{equation}
By Lemma \ref{RN.est}, there exists a constant $C$ (depending on $M$), so that we can estimate
\begin{equation}
\label{V.behaviour}
0\leq \f{dV}{dv}\leq C(1+v)^{-2}.
\end{equation}

Define the \emph{Cauchy horizon} as the boundary $\{V=0\}$. (Again, this is to be understood either as $\{(u,V):V=0\}\subset \mathcal Q$ or the original manifold $\{(u,V):V=0\}\times \subset \mathcal M_{eRN}$.) After denoting $v(V)$ as the inverse of $v\mapsto V$, we abuse notation to write $r_0(u,v(V))=r_0(u,v(v))$ and $\widetilde{\Omega}_0$ is defined by
$$\widetilde{\Omega}_0^2(u,v(V)) = \Omega_0^{-2}(-1,v(V))\Omega_0^2(u,v(V)),$$ 
the extremal Reissner--Nordstr\"om metric takes the following form in the $(u,V,\theta,\varphi)$ coordinate system
\begin{equation*}
g_{\textnormal{eRN}}=-\widetilde{\Omega}^2_0(u,V)\,dudV+r_0^2(u,V)(d\theta^2+\sin^2\theta\,d\varphi^2),
\end{equation*}
In analogy with Section~\ref{sec:coord.EH}, it is easy to see that $\widetilde{\Omega}^2_0$ and $r_0$ extend smoothly to the Cauchy horizon.

\section{Initial data assumptions}\label{sec:data}

We will consider the characteristic initial value problem for \eqref{EMCSFS} with initial data given on two transversally intersecting null hypersurfaces, which in the double null coordinates $(U,v)$ are denoted by
$$H_0:=\{(U,v):U=0\},\quad \underline{H}_{v_0}:=\{(U,v):v=v_0\}.$$
Here, the $(U,v)$ coordinates should be thought of as comparable to the Reissner--Nordstr\"om $(U,v)$ coordinates in Section~\ref{sec:coord.EH}; see Section~\ref{sec:null.comment} for further comments.

The initial data consist of $(\phi,r,\Omega,Q)$ on both $H_0$ and $\underline{H}_{v_0}$, subject to the equations \eqref{eq:Q1} and \eqref{eq:raychu} on $\underline{H}_{v_0}$, as well as the equations \eqref{eq:Q2} and \eqref{eq:raychv} on $H_0$.

We impose the following \emph{gauge conditions} on the initial hypersurfaces $\underline{H}_{v_0}$ and $H_0$:
\begin{equation}\label{eq:gauge.cond}
\hat{\Omega}(U,v_0)=\hat{\Omega}_0(U,v_0) \mbox{ for $U\in [0,U_0]$,}\quad \hat{\Omega}(0, v)=\hat{\Omega}_0(0,v)=1 \mbox{ for $v\in [v_0,\infty)$,}
\end{equation}
which can be thought of as a normalisation condition for the null coordinates. 

The initial data for $(\phi,r,\Omega,Q)$ will be prescribed in Sections~\ref{sec:data.phi}--\ref{sec:data.dur}, but before that, we will give some remarks in Sections~\ref{sec:null.comment} and \ref{sec:data.comment}: in \textbf{Section~\ref{sec:null.comment}}, we discuss our conventions on null coordinates; in \textbf{Section~\ref{sec:data.comment}}, we discuss which parts of the data are freely prescribable and which parts are determined by the constraints. We then proceed to discuss the initial data and the bounds that they satisfy. In \textbf{Section~\ref{sec:data.phi}}, we discuss the data for $\phi$; in \textbf{Section~\ref{sec:data.r}}, we discuss the data for $r$; in \textbf{Section~\ref{sec:data.Q}}, we discuss the data for $Q$; in \textbf{Section~\ref{sec:data.dur}}, we discuss the data for $\rd_U r$ on $H_0$.

\subsection{A comment about the use of the null coordinates}\label{sec:null.comment}

In the beginning of Section~\ref{sec:data}, we normalised the null coordinates $(U,v)$ on the initial hypersurfaces by the condition \eqref{eq:gauge.cond} so that they play a similar role to the $(U,v)$ coordinates on extremal Reissner--Nordstr\"om spacetimes introduced in Section~\ref{sec:coord.EH}. This set of null coordinates has the advantage of being \emph{regular near the event horizon} and therefore it is easy to see that the Einstein--Maxwell--Klein--Gordon system is locally well-posed with the prescribed initial data. 

However, in the remainder of the paper, it will be useful to pass to other sets of null coordinates. For this we introduce the following convention. \textbf{We use all of the coordinate systems $(U,v)$, $(u,v)$ and $(u,V)$, where $u$ and $V$ are defined (as functions of $U$ and $v$ respectively) by \eqref{U.def} and \eqref{V.def}.}

All the data will be prescribed in the $(U,v)$ coordinate system and we will prove estimates for $\phi$, $r$ and $Q$ in these coordinates. Nevertheless, using \eqref{U.def}, they imply immediately also estimates in the $(u,v)$ coordinate system, and it is those estimates that will be used in the later parts of the paper.

\subsection{A comment about freely prescribable data}\label{sec:data.comment}

Since the initial data need to satisfy \eqref{eq:Q1}, \eqref{eq:Q2}, \eqref{eq:raychu} and \eqref{eq:raychv}, not all of the data are freely prescribed. Instead, we have freely prescribable and constrained data:
\begin{itemize}
\item A normalisation condition for $u$ and $v$ can be specified. In our case, we specify the condition \eqref{eq:gauge.cond}.
\item $\phi$ on $H_0$ and $\underline{H}_0$ can be prescribed freely.
\item $r$ and $Q$ can then be obtained by solving \eqref{eq:Q1}, \eqref{eq:Q2}, \eqref{eq:raychu} and \eqref{eq:raychv} with appropriate initial conditions, namely,
\begin{itemize}
\item $r$ and $Q$ are to approach their corresponding values in extremal Reissner--Nordstr\"om with mass $M>0$, i.e.
$$\lim_{v\to \infty} r(0,v) = \lim_{v\to \infty} Q(0,v) = M.$$
\item $(\rd_U r)(0,v_0)$ can be freely prescribed, cf.~\eqref{eq:initdur}.
\end{itemize}
\end{itemize}

We remark that in order to fully specify the initial data, it only remains to pick a gauge condition for $A$. For this purpose, it will be most convenient to set $A_U(U,v_0)=0$ and $A_v(0,v)=0$. (This can always be achieved as each of these are only set to vanish on \emph{one hypersurface}, cf. discussions following \eqref{A.gauge.con}.) Nevertheless, the choice of gauge will not play a role in the rest of this subsection, since all the estimates we will need for $\phi$ and its derivatives can be phrased in terms of the \emph{gauge invariant} quantities $|\phi|$, $|D_v\phi|$ and $|D_u\phi|$.

\subsection{Initial data for $\phi$}\label{sec:data.phi}

We assume that there exists constants $\mathcal{D}_{\rm i}$ and $\mathcal{D}_{\rm o}$ such that
\begin{align}
\label{eq:idataphi} 
\int_{0}^{U_0} |D_{U} \phi |^2(U,v_0)\,dU\leq &\:\mathcal{D}_{\rm i} ,\\
\label{eq:odataphi}
\int_{v_0}^{\infty}v'^{2+\alpha}|D_{v}\phi|^2(0,v')\,dv'\leq &\: \mathcal{D}_{\rm o},
\end{align}
where we will take $\alpha>0$. We additionally assume that 
\begin{equation}\label{eq:philimit}
\lim_{v\to \infty}\phi(0,v)=0.
\end{equation}

\begin{lemma}\label{lm:phiinit}
The following estimate holds:
\begin{equation}
\label{eq:oinitphi}
|\phi|(0,v) \leq \sqrt{\mathcal{D}_{\rm o}}v^{-\frac{1}{2}-\frac{\alpha}{2}}.
\end{equation}
\end{lemma}
\begin{proof}
By \eqref{eq:philimit}, \eqref{eq:fundthmcalcphi2} and Cauchy--Schwarz inequality, we have that
\begin{equation*}
|\phi|(0,v)\leq \int_v^{\infty} |D_v\phi|(0,v')\,dv'\leq \sqrt{\int_{v}^{\infty} v'^{-2-\alpha}dv'}\cdot \sqrt{\int_{v}^{\infty} v'^{2+\alpha}|D_v\phi|(0,v')dv'},
\end{equation*}
so we can conclude using \eqref{eq:odataphi}.\qedhere
\end{proof}

\subsection{Initial data for $r$}\label{sec:data.r}
We assume that
\begin{equation}\label{eq:rlimit}
\lim_{v\to \infty}r(0,v)=M
\end{equation}
and we prescribe freely $\partial_Ur(0,v_0)$. Let us assume that
\begin{equation}
\label{eq:initdur}
\rd_U r(0,v_0)<0,\quad |\partial_Ur|(0,v_0)\leq M\mathcal{D}_{\rm i}.
\end{equation}

We use the equations \eqref{eq:raychu} and \eqref{eq:raychv} as constraint equations for the variable $r$ along $H_0=\{U=0\}$ and $\underline{H}_{v_0}=\{v=v_0\}$.

\subsubsection{Initial data for $r$ on $H_0$}
We obtain along ${H}_0$:
\begin{equation}\label{raychv.ODE}
\partial_{v}^2r(0,v)=-4\pi  r(0,v)|D_v\phi|^2,
\end{equation}
The above ODE can be solved to obtain $r(0,v)$.  

\begin{lemma}\label{lm:r.on.H0}
There exists a unique smooth solution to \eqref{raychv.ODE} satisfying \eqref{eq:rlimit}. Moreover, if $v_0$ satisfies the inequality
\begin{equation}\label{v0.first.est}
\mathcal{D}_{\rm o}v^{-1}_0\leq \f 1{8\pi},
\end{equation}
then the following estimates hold for $v\geq v_0$:
$$\frac{M}{2}\leq r(0,v)\leq M,\quad |r-M|(0,v)\leq 4 \pi M \mathcal{D}_{\rm o} v^{-1},\quad |\rd_v r|(0,v)\leq 4 \pi M \mathcal{D}_{\rm o} v^{-2}. $$
\end{lemma}
\begin{proof}
Existence and uniqueness can be obtained using a standard ODE argument. We will focus on proving the estimates.

First, observe that by integrating \eqref{raychv.ODE}, and using the assumptions $r(0,v)\to M$ and \eqref{eq:odataphi}, it follows that the limit $\lim_{v\to\infty} (\rd_v r)(0,v)$ exists. Now using again the assumption $r(0,v)\to M$, we deduce that 
\begin{equation}\label{dvr.limit.data}
\lim_{v\to\infty}(\rd_v r)(0,v) = 0.
\end{equation}
Together with \eqref{raychv.ODE} this implies that
\begin{equation*}
\partial_vr(0,v)\geq 0.
\end{equation*}
Since $r(0,v)\to M$, we can then bound
\begin{equation}
\label{eq:upperbinitr}
r(0,v)\leq M.
\end{equation}

By \eqref{raychv.ODE} and \eqref{dvr.limit.data},
\begin{equation*}
|\partial_{v}r(0,v)|\leq 4\pi \sup_{v_0\leq v<\infty}r(0,v)\cdot \int_{v}^{\infty} |D_{v} \phi |^2(u,v')\,dv'.
\end{equation*}
We deduce, using \eqref{eq:odataphi} and \eqref{eq:upperbinitr}, that
\begin{equation}
\label{eq:initdvr}
|\partial_{v}r(0,v)|\leq 4\pi M \mathcal{D}_{\rm o}v^{-2-\alp},
\end{equation}
and therefore
\begin{equation}\label{eq:initr}
|r(0,v)-M|\leq \f{4\pi M}{1+\alp} \mathcal{D}_{\rm o}v^{-1-\alp}.
\end{equation}
for all $v_0\leq v <\infty$. In particular, given \eqref{v0.first.est}, it holds that for all $v\in [v_0,\infty)$,
\begin{equation}\label{eq:initr.lowerbound}
r(0,v)\geq \f{M}{2}.
\end{equation}
The estimates stated in the lemma hence follow from \eqref{eq:upperbinitr}, \eqref{eq:initdvr}, \eqref{eq:initr} and \eqref{eq:initr.lowerbound}.
\end{proof}

\subsubsection{Initial data for $r$ on $\underline{H}_0$}

We similarly obtain along $\underline{H}_0$:
\begin{equation}\label{raychu.ODE}
\partial_{U}(\hat{\Omega}^{-2}(U,v_0) \rd_U r(U,v_0))=-4\pi  \hat{\Omega}^{-2}(U,v_0)r(U,v_0)|D_U\phi|^2.
\end{equation}
The above ODE can be solved to obtain $r(U,v_0)$. 

\begin{lemma}\label{lm:r.on.Hb0}
There exists a unique smooth solution to \eqref{raychu.ODE} satisfying \eqref{eq:rlimit}. Moreover, if $U_0$, $v_0$ satisfy the inequality
\begin{equation}\label{U0.first.est}
\mathcal{D}_{\rm o}v^{-1}_0\leq \f 1{8\pi},\quad M \mathcal D_{\rm i} v_0^2 U_0\leq \f{1}{36\pi},
\end{equation}
then the following estimates hold for $U\in [0,U_0]$:
$$\f{M}{4}\leq r(U,v_0)\leq \f{5M}{4},\quad |\rd_U r|(U,v_0)\leq 9\pi M \mathcal D_{\rm i}v_0^2.$$
\end{lemma}
\begin{proof}
As in Lemma~\ref{lm:r.on.Hb0}, since existence and uniqueness is standard, we focus on the estimates. To this end, we introduce a bootstrap argument. Introduce the following bootstrap assumption
\begin{equation}\label{data.Hb0.bootstrap}
r(U,v_0) \leq 2M.
\end{equation}
Integrating \eqref{raychu.ODE} and using \eqref{eq:gauge.cond},
\begin{equation*}
|\hat{\Omega}^{-2}(U,v_0)\partial_{U}r(U,v_0)-\partial_{U}r(0,v_0)|\leq 4\pi \sup_{0\leq U'\leq U_0} \hat{\Omega}^{-2}(U',v_0) r(U',v_0)\cdot \int_{0}^{U} |D_{U} \phi |^2(U'',v_0)\,dU''.
\end{equation*}
By \eqref{Omhat0.bounds}, \eqref{eq:idataphi} and \eqref{eq:initdur}, this implies
\begin{equation}\label{eq:durinit}
|\partial_{U}r(U,v_0)|\leq M\mathcal D_{\rm i} + 8\pi M \cdot\left(\f{v_0+1}{2}\right)^2\cdot \mathcal D_{\rm i} = M \mathcal D_{\rm i}(1+2\pi(v_0+1)^2)\leq 9 \pi M \mathcal D_{\rm i} v_0^2.
\end{equation}
%\begin{equation}\label{eq:durinit}
%|\partial_{U}r(U,v_0)|\leq (2M)^2 \left(\f 14 M^{-1}\mathcal D_{\rm i} + 8\pi M \cdot\left(\f{v_0+1}{4M}\right)^2\cdot \mathcal D_{\rm i}\right) = M \mathcal D_{\rm i}(1+2\pi(v_0+1)^2)\leq 9 \pi M \mathcal D_{\rm i} v_0^2.
%\end{equation}
Integrating in $U$, this yields (for $U\in [0,U_0]$),
\begin{equation*}
|r(U,v_0)-r(0,v_0)|\leq 9 \pi M \mathcal D_{\rm i} v_0^2 U_0,
\end{equation*}
which implies, using Lemma~\ref{lm:r.on.H0} (or more precisely \eqref{eq:upperbinitr} and \eqref{eq:initr.lowerbound}),
$$\f {M} 2 - 9 \pi M \mathcal D_{\rm i} v_0^2 U_0\leq r(U,v_0)\leq M + 9 \pi M \mathcal D_{\rm i} v_0^2 U_0.$$
Hence, by \eqref{U0.first.est}, it holds that
\begin{equation}\label{eq:rinit.Hb.est}
\f{{M}}{4} \leq r(U,v_0)\leq \f {5M}4,
\end{equation}
and we have improved the bootstrap assumption \eqref{data.Hb0.bootstrap}. This closes the bootstrap argument, and the desired estimates follow from \eqref{eq:durinit} and \eqref{eq:rinit.Hb.est}.\qedhere
\end{proof}

\subsection{Initial data for $Q$}\label{sec:data.Q}
In view of the equations \eqref{eq:Q1} and \eqref{eq:Q2} which have to be satisfied on the initial hypersurfaces, it suffices to impose $Q$ on one initial sphere.
We assume that
\begin{equation}\label{eq:Qlimit}
\lim_{v\to \infty}Q(0,v)=M.
\end{equation}

\begin{lemma}\label{lm:Qinit}
Assume \eqref{v0.first.est} holds. Then the following estimate holds on $H_0$:
\begin{equation}
\label{eq:oinitQ}
|Q(0,v)-M|(0,v)\leq 4\pi |\mfe|M \mathcal{D}_{\rm o}v^{-1-\alpha}.
\end{equation}
\end{lemma}
\begin{proof}
Using \eqref{eq:Q2}, \eqref{eq:upperbinitr} and Cauchy--Schwarz inequality we estimate
\begin{equation*}
\begin{split}
|Q(0,v)-M|(0,v)\leq &\: 4\pi |\mfe|\sup_{v''\in [v_0,\infty} r^2(0,v'') \int_{v}^{\infty} |\phi|\cdot|D_v\phi|(0,v')\,dv'\\
\leq &\:4\pi |\mfe|M^2\sup_{v\leq v'<\infty}|\phi|(0,v')\cdot \sqrt{\int_{v}^{\infty} v'^{-2-\alpha}dv'}\cdot \sqrt{\int_{v}^{\infty} v'^{2+\alpha}|D_v\phi|^2(0,v')dv'},
\end{split}
\end{equation*}
so we can use \eqref{eq:odataphi} and \eqref{eq:oinitphi} to conclude \eqref{eq:oinitQ}.
\end{proof}

\subsection{$\partial_Ur$ along $H_0$}\label{sec:data.dur}
The function $\rd_u r$ along $H_0$ is not freely prescribable, but is dictated by \eqref{eq:transpeqr} and the freely prescribable data for $(\rd_U r)(0,v_0)$ (which obeys \eqref{eq:initdur}). We will need the following estimate for $\partial_Ur$ along $H_0$. 
\begin{lemma}
Suppose \eqref{v0.first.est} holds. Then there exists a constant $C>0$ depending only on $M$ and $\mfm$ such that for every $v\in [v_0,\infty)$,
\begin{equation}
\label{eq:durhorizon}
|\partial_Ur(0,v)|\leq C(\mathcal{D}_{\rm o}+\mathcal{D}_{\rm i}).
\end{equation}
\end{lemma}
\begin{proof}
By \eqref{eq:transpeqr} we have that
\begin{equation*}
\begin{split}
\partial_v(r\partial_Ur)(0,v)=&\:4M^2\mathfrak{m}^2\pi r^2 |\phi|^2+M^2r^{-2}( Q^2-r^2)\\
=&\:4M^2\mathfrak{m}^2\pi r^2 |\phi|^2+M^2r^{-2}( Q^2-M^2)+M^2r^{-2}( M^2-r^2).
\end{split}
\end{equation*}
Hence,
\begin{equation*}
\begin{split}
|(r\partial_Ur)(0,v)-(r\partial_Ur)(0,v_0)|\lesssim&\: \left|\int_{v_0}^{\infty}\mathfrak{m}^2r^2 |\phi|^2+\frac{1}{4\pi}M^{-2}( Q^2-M^2)+\frac{1}{4\pi}M^{-2}(M^2-r^2)\,dv' \right|\lesssim \:\mathcal{D}_{\rm o},
\end{split}
\end{equation*}
where we have used Lemmas~\ref{lm:phiinit}, \ref{lm:r.on.H0} and \ref{lm:Qinit} (and we crucially used that $\alpha>0$). Together with \eqref{eq:initdur}, we can therefore conclude \eqref{eq:durhorizon}.

\end{proof}

\section{Statement of the main theorem}\label{sec:main.thm}

We are now ready to give a precise statement of the main theorem. Let us recall (from Section~\ref{sec:null.comment}) that we also consider the coordinate system $(u,v)$, where $u(U)$ is defined via the relation \eqref{U.def}. It will be convenient from this point onwards to use the $u$ (instead of $U$) coordinate.
\begin{theorem}\label{thm:main}
Suppose 
\begin{itemize}
\item the parameters $M$ and $\mathfrak e$ obey\footnote{Note that $$\left(10+5\sqrt{6}-3\sqrt{9+4\sqrt{6}}\right)\sim 9.24\dotsc$$.} 
\begin{equation}\label{par.con.thm}
1-\left(10+5\sqrt{6}-3\sqrt{9+4\sqrt{6}}\right)|\mathfrak{e}|M>0.
\end{equation}
\item the initial data are smooth and satisfy \eqref{eq:gauge.cond}, \eqref{eq:idataphi}, \eqref{eq:odataphi}, \eqref{eq:philimit}, \eqref{eq:rlimit}, \eqref{eq:initdur} and \eqref{eq:Qlimit} for some finite $\mathcal{D}_{\rm o}$ and $\mathcal D_{\rm i}$. 
\end{itemize}
Then for $|u_0|$ sufficiently large depending on $M$, $\mfm$, $\mfe$, $\alp$, $\mathcal{D}_{\rm o}$ and $\mathcal D_{\rm i}$ and $v_0$ sufficiently large depending on $M$, $\mfm$, $\mfe$, $\alp$ and $\mathcal{D}_{\rm o}$ (but not $\mathcal D_{\rm i}$!), the following holds:
\begin{itemize}
\item (Existence of solution) There exists a unique smooth, spherically symmetric solution to \eqref{EMCSFS} in the double null coordinate system in $(u,v)\in (-\infty,u_0]\times [v_0,\infty)$.
\item (Extendibility to the Cauchy horizon) In an appropriate coordinate system, a Cauchy horizon can be attached to the solution so that the metric, the scalar field and the Maxwell field extend continuously to it.
\item (Quantitative estimates) The following estimates hold for all $(u,v)\in (-\infty,u_0)\times [v_0,\infty)$ for some implicit constant depending on $M$, $\mfm$ and $\mfe$ (which shows that the solution is close to extremal Reissner--Nordstr\"om in an appropriate sense):
$$|\phi|(u,v)\ls  \mathcal{D}_{\rm o} v^{-\f 12 - \f{\alp}{2}} + (\mathcal{D}_{\rm o} + \mathcal D_{\rm i}) |u|^{-\f 12},\quad |r-M|(u,v)\ls  \mathcal{D}_{\rm o} v^{-1} + (\mathcal{D}_{\rm o} + \mathcal D_{\rm i})|u|^{-1}, $$
$$|\Om^2-\Om_0^2|(u,v)\ls |u|^{-\f 12}(|u|+v)^{-2},$$
$$|\rd_u r|(u,v)\ls (\mathcal{D}_{\rm o} + \mathcal D_{\rm i})|u|^{-2},\quad |\rd_v r|(u,v)\ls (\mathcal{D}_{\rm o} + \mathcal D_{\rm i})v^{-2}, $$
$$\int_{v_0}^\infty v^2|\rd_v\phi|^2(u,v)\, dv + \int_{-\infty}^{u_0} u^2|\rd_u\phi|^2(u,v)\, du\ls \mathcal{D}_{\rm o}+\mathcal D_{\rm i},$$
\begin{equation*} 
\int_{-\infty}^{u_0} u^2 \left(\partial_u\left(\log \frac{\Omega}{\Omega_0} \right)\right)^2(u,v)\,du+\int_{v_0}^{v_{\infty}} v^2 \left(\partial_v\left(\log \frac{\Omega}{\Omega_0} \right)\right)^2(u,v)\,dv\leq \f 12.
\end{equation*}
\item (Extendibility as a spherically symmetric solution) The solution can be extended non-uniquely in (spacetime) $C^{0,\f 12}\cap W^{1,2}_{loc}$ beyond the Cauchy horizon as a spherically symmetric solution to the Einstein--Maxwell--Klein--Gordon system.
\end{itemize}
\end{theorem}

\begin{remark}[$v_0\geq 1$, $u_0\leq -1$.]
Without loss of generality, we will from now on assume that $v_0\geq 1$ and $u_0\leq -1$.
\end{remark}

\begin{remark}[Validity of the estimates in Section~\ref{sec:data}]
Recall that in Section~\ref{sec:data}, some of the estimates that were proven depend on the assumptions \eqref{v0.first.est} and \eqref{U0.first.est}. From now on, we take $v_0$ sufficiently large and $u_0$ sufficiently negative (in a manner allowed by Theorem~\ref{thm:main}) so that \eqref{v0.first.est} and \eqref{U0.first.est} hold.
\end{remark}

\begin{remark}[Relaxing the largeness of $v_0$]
Note that $v_0$ is assumed to be large depending on $M$, $\mfm$, $\mfe$, $\alp$ and $\mathcal{D}_{\rm o}$ so that we restrict our attention to a region where the geometry is close to that of extremal Reissner--Nordstr\"om. However, in general, if we are given data with $v_{0,i}=1$ (say) and $\mathcal{D}_{\rm o}$ not necessarily small, we can do the following:
\begin{enumerate}
\item First, find a $v_0$ (sufficiently large) such that $v_0$ is sufficiently large depending on $M$, $\mfm$, $\mfe$, $\alp$ and $\mathcal{D}_{\rm o}$ in a way that is required by Theorem~\ref{thm:main}. 
\item Solve a \underline{finite} characteristic initial value problem in $(-\infty,u_0]\times [v_{0,i},v_0]$ for some $u_0$ sufficiently negative. (Such a problem can always be solved for $u_0$ sufficiently negative. This can be viewed as a restatement of the fact that for \underline{local} characteristic initial value problems, one only needs the smallness of \underline{one} characteristic length, as long as the other characteristic length is \underline{finite}, cf.~\cite{L.local}.)
\item Now, let $\mathcal D_{\rm i}$ be the size of the \emph{new data} on $\{v=v_0\}$ which is obtained from the previous step. By choosing $u_0$ smaller if necessary, it can be arranged so that $|u_0|$ is large in terms of $M$, $\mfm$, $\mfe$ and $\mathcal{D}_{\rm o} + \mathcal D_{\rm i}$ in a way consistent with Theorem~\ref{thm:main}.
\item Theorem~\ref{thm:main} can now be applied to obtain a solution in $(-\infty,u_0]\times [v_0,\infty)$ such that the conclusions of Theorem~\ref{thm:main} hold.
\end{enumerate}
\end{remark}

In the case $\mfe=\mfm=0$, we obtain the following additional regularity of the scalar field:
\begin{theorem}\label{thm:Lipschitz}
In the case $\mfe=\mfm=0$, suppose that in addition to the assumptions of Theorem~\ref{thm:main}, the following pointwise bounds hold for the initial data:
\begin{equation}\label{eq:extra.assump}
\sup_{u\in (-\infty,u_0]} |u|^2|\rd_u \phi|(u,v_0) + \sup_{v\in [v_0,\infty)} v^2 |\rd_v\phi|(-\infty,v)<\infty.
\end{equation}
Then, taking $u_0$ more negative if necessary, in the $(u,V)$ coordinate system (see \eqref{Vvrelation}), the scalar field is Lipschitz up to the Cauchy horizon.
\end{theorem}

\section{The main bootstrap argument}\label{sec:bootstrap}

\subsection{Setup of the bootstrap}
We will assume that 
\begin{itemize}
\item there exists a smooth solution $(\phi,\Omega,r,A)$ to the system of equations \eqref{eq:transpeqr}--\eqref{eq:raychv} in the rectangle $D_{U_0,[v_0,v_{\infty})}=\{(U,v)\,|\, 0\leq U\leq U_0, \,v_0\leq v< v_{\infty}\}$, such that 
\item the initial gauge conditions are satisfied, i.e.~$\hat{\Omega}^2(U,v_0)=\hat{\Omega}_0^2(U,v_0)$ and $\hat{\Omega}^2(0,v)=\hat{\Omega}_0^2(0,v)$, and
\item the initial conditions for $\phi$, $r$, $Q$ are attained.
\end{itemize}
On this region, we will moreover assume that certain \emph{bootstrap assumptions} hold (cf.~Section~\ref{sec:BA}). Our goal will then be to improve these bootstrap assumptions, which then by continuity, implies that the above three properties hold in for all $v\geq v_0$, i.e.~in the region $D_{U_0,v_0}=D_{U_0,[v_0,\infty)}$.

Recall again that we often use the $(u,v)$ instead the $(U,v)$ coordinates. Abusing notation, we will also write
\begin{equation*}
D_{u_0,[v_0,v_{\infty})}=D_{U_0,[v_0,v_{\infty})}=\{(u,v)\,|\, -\infty<u\leq u_0:=u(U_0),\, v_0\leq v\leq v_{\infty}\}.
\end{equation*}
%We will refer to the coordinates $(u,v)$ as the Eddington Finkelstein gauge and we will use the notation $H_0=\{u=-\infty\}$. 

\subsection{Bootstrap assumptions}\label{sec:BA}

Fix $\eta>0$ sufficiently small (depending only on $\mfe$ and $M$) so that 
\begin{equation}\label{eta.def}
1-\left(10+5\sqrt{6}-3\sqrt{9+4\sqrt{6}}\right)(1+\eta)|\mathfrak{e}|M> 0.
\end{equation}
(Such an $\eta$ exists in view of \eqref{par.con.thm}.) Define
\begin{equation}\label{mu.def}
\mu=\left(1-\left(10+5\sqrt{6}-3\sqrt{9+4\sqrt{6}}\right)(1+\eta)|\mathfrak{e}|M\right).
\end{equation}
Let us make the following bootstrap assumptions for the quantities $(\phi,\Omega,r)$ in $D_{U_0,v_{\infty}}$, for some $\mathcal A_\phi\geq 1$ to be chosen later: 
\begin{align}
\label{eq:bootstrapOmega}\tag{A1}
&\sup_{v\in [v_0,v_\infty]}\int_{-\infty}^{u_0} u'^2 \left(\partial_u\left(\log \frac{\Omega}{\Omega_0} \right)\right)^2(u',v)\,du'+\sup_{u\in(-\infty,u_0]}\int_{v_0}^{v_{\infty}} v'^2 \left(\partial_v\left(\log \frac{\Omega}{\Omega_0} \right)\right)^2(u,v')\,dv'\leq \: M,\\
\label{eq:bootstrapphi} \tag{A2}
&\sup_{v\in [v_0,v_\infty]}\int_{-\infty}^{u_0} u'^2 |D_u \phi |^2(u',v)\,du'+\sup_{u\in(-\infty,u_0]}\int_{v_0}^{v_{\infty}}v'^2|D_v\phi|^2(u,v')\,dv'\leq \:\mathcal A_\phi (\mathcal{D}_{\rm o}+\mathcal D_{\rm i}),\\
\label{eq:bootstrapr} \tag{A3}
&\sup_{u\in(-\infty,u_0],\,v\in [v_0,v_\infty]}|r-M|(u,v)\leq \: \frac{M}{2}.
\end{align}
Our goal will be to show that under these assumptions, for $|u_0|$ sufficiently large depending on $M$, $\mfm$, $\mfe$, $\alp$, $\eta$, $\mathcal{D}_{\rm o}$ and $\mathcal D_{\rm i}$ and $v_0$ sufficiently large depending on $M$, $\mfm$, $\mfe$, $\alp$, $\eta$ and $\mathcal{D}_{\rm o}$,
\begin{itemize}
\item the estimate \eqref{eq:bootstrapOmega} can be improved so that the RHS can be replaced by $\f {M}2$;
\item the estimate \eqref{eq:bootstrapphi} can be improved so that the RHS can be replaced by $C(\mathcal{D}_{\rm o}+\mathcal D_{\rm i})$, where $C$ is a constant depending only on $M$, $\mfm$, $\mfe$, $\alp$ and $\eta$;
\item the estimate \eqref{eq:bootstrapr} can be improved to $|r-M| \leq \frac{M}{4}$.
\end{itemize}

\subsection{Conventions regarding constants}

In closing the bootstrap argument, the main source of smallness will come from choosing $|u_0|$ and $v_0$ appropriately large. We remark on our conventions regarding the constants that will be used in the bootstrap argument.
\begin{itemize}
\item All the implicit constants (either in the form of $C$ or $\ls$) are allowed to depend on the parameters $M$, $\mfm$, $\mfe$ and $\alp$. In particular, they are allowed to depend on $\eta$ (defined in \eqref{eta.def}) and $\mu$ (defined in \eqref{mu.def}). There will be places where the exact values of these parameters matter (hence the corresponding restriction in Theorem~\ref{thm:main}): at those places the constants will be explicitly written.
\item $|u_0|$ is taken to be large depending on $M$, $\mfm$, $\mfe$, $\alp$, $\eta$, $\mathcal{D}_{\rm o}$ and $\mathcal D_{\rm i}$, and $v_0$ is taken to be large depending on $M$, $\mfm$, $\mfe$, $\alp$, $\eta$ and $\mathcal{D}_{\rm o}$. In particular, we will use
$$\mathcal{D}_{\rm o} v_0^{-\f 1{10}} \ll 1,\quad (\mathcal{D}_{\rm o} + \mathcal D_{\rm i}) |u_0|^{-\f 1{10}} \ll 1$$
without explicit comments, where by $\ll 1$, we mean that it is small with respect to the constants appearing in the argument that depend on $M$, $\mfm$, $\mfe$, $\alp$ and $\eta$.
\item $\mathcal A_\phi\geq 1$ will eventually be chosen to be large depending $M$, $\mfm$, $\mfe$ and $\alp$, \textbf{but not on $\mathcal{D}_{\rm o}$ and $\mathcal D_{\rm i}$}. In particular, we will also use
\begin{equation}\label{const.adm}
\mathcal A_\phi^3\mathcal{D}_{\rm o} v_0^{-\f 1{10}} \ll 1,\quad \mathcal A_\phi^3(\mathcal{D}_{\rm o} + \mathcal D_{\rm i}) |u_0|^{-\f 1{10}} \ll 1
\end{equation}
without explicit comments.
\end{itemize}

\section{Pointwise estimates}\label{sec:pointwise}

\begin{proposition}
For all $-\infty<u\leq u_0$, $v_0\leq v\leq v_{\infty}$ we have that:
\begin{align}
\label{eq:Linftyphi}
|\phi|(u,v)\lesssim &\: \sqrt{\mathcal{D}_{\rm o}}v^{-\frac{1}{2}-\f{\alp}{2}}+ \mathcal A_{\phi}^{\f 12}\sqrt{\mathcal{D}_{\rm o}+\mathcal D_{\rm i}}|u|^{-\frac{1}{2}},\\
\label{eq:LinftyQ}
 |Q-M|(u,v)\lesssim &\: \mathcal{D}_{\rm o}v^{-1-\alp}+ \mathcal A_{\phi}(\mathcal{D}_{\rm o}+\mathcal D_{\rm i}) |u|^{-1},\\ 
\label{eq:LinftyOmegav3}
\left|\Omega-\Omega_0\right|(u,v)\lesssim&\:  |u|^{-\frac{1}{2}}(v+|u|)^{-1},\\
\label{eq:LinftyOmegav4}
\left|\Omega^2-\Omega_0^2\right|(u,v)\lesssim&\:  |u|^{-\frac{1}{2}}(v+|u|)^{-2},\\
\label{eq:LinftyOmegav1}
\left|\Omega-\frac{2M}{(v+|u|)}\right|(u,v)\lesssim&\:  |u|^{-\frac{1}{2}}(v+|u|)^{-1},\\
\label{eq:LinftyOmegav2}
\left|\Omega^2-\frac{4M^2}{(v+|u|)^2}\right|(u,v)\lesssim&\:  |u|^{-\frac{1}{2}}(v+|u|)^{-2}.
\end{align}
In particular, 
\begin{align}
\label{eq:estrangeQ}
\frac{1}{2}M\leq |Q|(u,v)\leq&\: \frac{3}{2}M,\\
\label{eq:estrangeOmega}
 2M^2(v+|u|)^{-2}\leq&\: \Omega^2(u,v)\leq 6M^2(v+|u|)^{-2}.
\end{align}
\end{proposition}
\begin{proof}
\textbf{Proof of \eqref{eq:Linftyphi}.} By \eqref{eq:fundthmcalcphi1} and the Cauchy--Schwarz inequality, we obtain
\begin{equation}
\label{eq:fundcalc1}
|\phi|(u,v)\leq |\phi|(-\infty,v)+ |u|^{-\frac{1}{2}} \sqrt{\int_{-\infty}^u u'^2 |{D}_u {\phi}|^2(u',v)\,du'}.
\end{equation}
Using \eqref{eq:oinitphi} and the bootstrap assumption \eqref{eq:bootstrapphi} to control the first and second term respectively, we obtain \eqref{eq:Linftyphi}.

\textbf{Proof of \eqref{eq:LinftyQ}.} Using the estimate \eqref{eq:Linftyphi} for $\phi$, the bootstrap assumption \eqref{eq:bootstrapr} for $r$ together with \eqref{eq:Q1}, we obtain a pointwise estimate for $|Q-M|$:
\begin{equation*}
\begin{split}
|Q-M|(u,v)\leq&\: |Q-M|(-\infty,v)+\int_{-\infty}^u|\partial_uQ|(u',v)\,du'\\
\leq &\: |Q-M|(-\infty,v)+4\pi |\mathfrak{e}|\int_{-\infty}^ur^2|\phi||D_u\phi|(u',v)\,du'\\
\leq &\: |Q-M|(-\infty,v)+4\pi |\mathfrak{e}||u|^{-\frac{1}{2}}\sup_{-\infty<u'\leq u}r^2|\phi|(u',v)\cdot  \sqrt{\int_{-\infty}^uu^2|D_u\phi|^2(u',v)\,du'}\\
\leq &\: |Q-M|(-\infty,v)+C\sqrt{\mathcal{A}_{\phi}(\mathcal{D}_{\rm o}+\mathcal D_{\rm i})}|u|^{-\frac{1}{2}}\left(|u|^{-\frac{1}{2}}\sqrt{\mathcal{A}_{\phi}(\mathcal{D}_{\rm o}+\mathcal D_{\rm i})}+v^{-\frac{1}{2}-\f{\alp}2}\sqrt{\mathcal{D}_{\rm o}}\right).
\end{split}
\end{equation*}
Using \eqref{eq:oinitQ} and Young inequality, we therefore conclude that
\begin{equation*}
|Q-M|(u,v) \ls \mathcal{D}_{\rm o}v^{-1-\alp}+ \mathcal{A}_{\phi}(\mathcal{D}_{\rm o}+\mathcal D_{\rm i}) |u|^{-1}.
\end{equation*}

\textbf{Proof of \eqref{eq:LinftyOmegav3}, \eqref{eq:LinftyOmegav4}, \eqref{eq:LinftyOmegav1} and \eqref{eq:LinftyOmegav2}.} By our choice of initial gauge \eqref{eq:gauge.cond}, we have that
\begin{equation*}
\log \frac{\Omega}{\Omega_0}(-\infty,v)=\log \frac{\hat{\Omega}}{\hat{\Omega}_0}(U=0,v)=0,
\end{equation*}
so we can estimate using \eqref{eq:bootstrapOmega}
\begin{align*}
\left|\log \frac{\Omega}{\Omega_0}\right|(u,v)\leq  |u|^{-\frac{1}{2}} \sqrt{\int_{-\infty}^{u} u'^2 \left(\partial_u\left(\log \frac{\Omega}{\Omega_0} \right)\right)^2(u',v')\,du'}\leq M^{\f 12}|u|^{-\f 12}.
\end{align*}
Using \eqref{eq:bootstrapOmega} and the simple inequality $|e^\vartheta - 1|\leq |\vartheta| e^{|\vartheta|}$, we have
$$\left|\f{\Omega}{\Omega_0} -1\right|(u,v) \leq M^{\f 12}|u|^{-\f 12} \max\left\{\f{\Omega}{\Omega_0}, \f{\Omega_0}{\Omega}\right\}(u,v),\quad \left|\f{\Omega_0}{\Omega} -1\right|(u,v) \leq M^{\f 12}|u|^{-\f 12} \max\left\{\f{\Omega}{\Omega_0}, \f{\Omega_0}{\Omega}\right\}(u,v).$$
We now consider two cases. Suppose $\f{\Omega}{\Omega_0}(u,v)>1$ for some $(u,v)$, we have
$$\left|\f{\Omega}{\Omega_0} -1\right|(u,v)\leq M^{\f 12}|u|^{-\f 12}\left(\f{\Omega}{\Omega_0}-1\right)(u,v) + M^{\f 12}|u|^{-\f 12},$$
which, after choosing $u_0$ to satisfy $M^{\f 12}|u_0|^{-\f 12}\leq \f 12$, implies
\begin{equation}\label{Om.diff.case.1}
\left|\f{\Omega}{\Omega_0} -1\right|(u,v)\leq 2M^{\f 12}|u|^{-\f 12}.
\end{equation}
Multiplying \eqref{Om.diff.case.1} by $\Om_0$ in particular implies 
\begin{equation}\label{eq:Om.diff.est.basic}
|\Omega-\Omega_0|(u,v)\leq 4M^{\f 12}|u|^{-\frac{1}{2}}\Omega_0(u,v).
\end{equation}
in this case.
On the other hand, if $\f{\Omega}{\Omega_0}(u,v)<1$ for some $(u,v)$, we have by a similar argument that for $M^{\f 12}|u_0|^{-\f 12}\leq \f 12$,
\begin{equation*}
\left|\f{\Omega_0}{\Omega} -1\right|(u,v)\leq 2M^{\f 12}|u|^{-\f 12}.
\end{equation*}
This then implies
$$|\Om_0-\Om|(u,v)\leq 2M^{\f 12}|u|^{-\f 12} |\Om-\Om_0|(u,v) + 2M^{\f 12}|u|^{-\f 12}\Om_0(u,v).$$
Choosing $M^{\f 12}|u_0|^{-\f 12}\leq \f 14$ implies that we also have \eqref{eq:Om.diff.est.basic} in this case. Using \eqref{eq:Om.diff.est.basic} and \eqref{eq:Om0.est}, we conclude \eqref{eq:LinftyOmegav3}.

For \eqref{eq:LinftyOmegav4}, we use \eqref{eq:Om.diff.est.basic} twice to obtain
$$|\Omega^2-\Omega_0^2|(u,v)\leq 4M^{\f 12}|u|^{-\frac{1}{2}}\Omega_0(\Om+\Om_0)\leq 16M|u|^{-1}\Omega_0^2 + 8 M^{\f 12}|u|^{-\frac{1}{2}}\Omega_0^2,$$
which, after choosing $|u_0|$ to be sufficiently large and using \eqref{eq:Om0.est}, implies \eqref{eq:LinftyOmegav4}.

Finally, \eqref{eq:LinftyOmegav1} and \eqref{eq:LinftyOmegav2} follow from \eqref{eq:LinftyOmegav3}, \eqref{eq:LinftyOmegav4} and \eqref{eq:Om0.est} (with $\beta = \f 12$).

\textbf{Proof of \eqref{eq:estrangeQ} and \eqref{eq:estrangeOmega}.} \eqref{eq:estrangeQ} is an immediate consequence of \eqref{eq:LinftyQ} while \eqref{eq:estrangeOmega} is an immediate consequence of \eqref{eq:LinftyOmegav2}. \qedhere

\end{proof}

\begin{proposition}\label{prop:rest}
The following estimates hold:
\begin{align}
\label{dur.est}|\partial_ur|(u,v)\lesssim &\:\mathcal A_{\phi}(\mathcal{D}_{\rm o}+\mathcal{D}_{\rm i})|u|^{-2},\\
\label{dvr.est}|\partial_vr|(u,v)\lesssim &\:\mathcal{D}_{\rm o}v^{-2}+\mathcal A_{\phi}(\mathcal{D}_{\rm o}+\mathcal{D}_{\rm i})\min\{v^{-2},|u|^{-2}\},\\
\label{r.est}|r(u,v)-M|\lesssim &\: \mathcal{D}_{\rm o}v^{-1}+\mathcal A_{\phi}(\mathcal{D}_{\rm o}+\mathcal{D}_{\rm i})|u|^{-1},\\
\label{r.est.2}|r(u,v)-r_0(u,v)|\lesssim &\: \mathcal{D}_{\rm o} v^{-1}+\mathcal A_{\phi}(\mathcal{D}_{\rm o}+\mathcal{D}_{\rm i})|u|^{-1}+ (v+|u|)^{-1}.
\end{align}
In particular,
\begin{equation}\label{r.BS.improve}
|r(u,v)-M|\leq \f M 4,
\end{equation}
which improves the bootstrap assumption \eqref{eq:bootstrapr}.
\end{proposition}
\begin{proof}
\textbf{Proof of \eqref{dur.est}.} By \eqref{eq:raychu}, \eqref{eq:bootstrapr}, \eqref{eq:estrangeOmega} and \eqref{eq:durhorizon}, we can estimate
\begin{equation*}
\begin{split}
|\partial_ur|(u,v)\leq &\: \Omega^2(u,v)|\partial_Ur|(-\infty,v)+C\Omega^2(u,v)\int_{-\infty}^{u}\Omega^{-2}(u',v)|D_u\phi|^2(u',v)\,du'\\
\leq &\: C(\mathcal{D}_{\rm o}+\mathcal{D}_{\rm i})|u|^{-2}+C(v+|u|)^{-2}\int_{-\infty}^{u}\frac{(v+|u'|)^2}{|u'|^2}|u'|^2|D_u\phi|^2(u',v)\,du'
\end{split}
\end{equation*}
Note that for $|u'|\geq |u|$, we have that
\begin{equation*}
\left(\frac{v+|u'|}{|u'|}\right)^2=\left(1+\frac{v}{|u'|}\right)^2\leq \left(1+\frac{v}{|u|}\right)^2=\left(\frac{v+|u|}{|u|}\right)^2,
\end{equation*}
which we can use to further estimate
\begin{equation*}
|\partial_ur|(u,v)\leq C(\mathcal{D}_{\rm o}+\mathcal{D}_{\rm i})|u|^{-2}+C|u|^{-2}\int_{-\infty}^u|u'|^2|D_u\phi|^2(u',v)\,du'
\end{equation*}
and hence, by \eqref{eq:bootstrapphi},
\begin{equation*}
|\partial_ur|(u,v)\leq C\mathcal A_{\phi}(\mathcal{D}_{\rm o}+\mathcal{D}_{\rm i})|u|^{-2}.
\end{equation*}
\textbf{Proof of \eqref{r.BS.improve}.} Using the fundamental theorem of calculus and integrating \eqref{dur.est} in $u$, we obtain
\begin{equation*}
\begin{split}
|r(u,v)-M|\leq&\: |r(-\infty,v)-M|+\int_{-\infty}^u |\partial_ur|(u',v)'\,du'
\lesssim \mathcal{D}_{\rm o}v^{-1}+\mathcal A_{\phi}(\mathcal{D}_{\rm o}+\mathcal{D}_{\rm i})|u|^{-1},
\end{split}
\end{equation*}
where in the last inequality we have used Lemma~\ref{lm:r.on.H0} and \eqref{dur.est}.
Combining this with the estimates for $r_0$ in Lemma~\ref{RN.est}, we thus obtain the estimate for $r-r_0$ in \eqref{r.est.2}.
In particular, for $\mathcal{D}_{\rm o}v_0^{-1}$, $(\mathcal{D}_{\rm o}+\mathcal{D}_{\rm i})|u_0|^{-1}$ suitably small, we obtain:
\begin{equation*}
|r(u,v)-M|<\frac{1}{4}M,
\end{equation*}
for all $(u,v)\in D_{u_0,v_{\infty}}$, which is the estimate \eqref{r.est}.

\textbf{Proof of \eqref{dvr.est}: the region $\{v\leq |u|\}$.} We first rewrite \eqref{eq:transpeqr} as follows:
\begin{equation*}
\begin{split}
\partial_u(r\partial_vr)=&\frac{1}{4}\Omega^2r^{-2}(Q^2-M^2)+ \frac{1}{4}\Omega^2r^{-2}(M^2-r^2)+\mathfrak{m}^2\pi r^2\Omega^2 |\phi|^2.
\end{split}
\end{equation*}
By \eqref{eq:Linftyphi} (for $\phi$), \eqref{r.est} (for $r-M$) and \eqref{eq:LinftyQ} (for $Q-M$), the $u$-integral of the RHS of the above equation can be estimated (up to a constant) by
\begin{equation*}
\int_{-\infty}^u \Omega^2 (\mathcal{D}_{\rm o} v^{-1} + \mathcal A_{\phi}(\mathcal{D}_{\rm o}+\mathcal D_{\rm i}) |u'|^{-1})\, du'.
\end{equation*}
Using \eqref{eq:LinftyOmegav2}, we have, in the region $\{v\leq |u|\}$,
$$\int_{-\infty}^u \Omega^2 \mathcal{D}_{\rm o} v^{-1} \, du'\ls \mathcal{D}_{\rm o} v^{-1}\int_{-\infty}^u (v+|u'|)^{-2}  \, du' \ls \mathcal{D}_{\rm o} v^{-1}(v+|u|)^{-1}$$
and 
\begin{equation*}
\begin{split}
\int_{-\infty}^u \Omega^2 \mathcal A_{\phi}(\mathcal{D}_{\rm o}+\mathcal D_{\rm i}) |u'|^{-1}\, \, du' \ls &\:\mathcal A_{\phi}(\mathcal{D}_{\rm o}+\mathcal D_{\rm i})\int_{-\infty}^u (v+|u'|)^{-2} |u'|^{-1} \, du' \\
\ls &\:\mathcal A_{\phi}(\mathcal{D}_{\rm o}+\mathcal D_{\rm i}) |u|^{-2}.
\end{split}
\end{equation*}
Together with the bound on $\rd_v r$ on the event horizon in Lemma~\ref{lm:r.on.H0}, this implies that when $v\leq |u|$, we have the following estimate:
$$|\rd_v r|(u,v)\ls \mathcal{D}_{\rm o}v^{-2}+\mathcal A_{\phi}(\mathcal{D}_{\rm o}+ \mathcal D_{\rm i})\min\{v^{-2},|u|^{-2}\}.$$

\textbf{Proof of \eqref{dvr.est}: the region $\{v\geq |u|\}$.} Notice that if we estimate in this region the same manner as before, we lose a factor of $\log v$ in the bound. So instead of \eqref{eq:transpeqr}, we will use the Raychaudhuri's equation \eqref{eq:raychv}\footnote{On the other hand, in the region $\{v\leq |u|\}$ that we considered above, it also does not seem that Raychaudhuri can give the desired bound.}. More precisely, given $(u,v)$ with $v\geq |u|$, we integrate \eqref{eq:raychv} along a constant-$u$ curve starting from its intersection with the curve $\{|u|=v\}$ to obtain
$$\left|\f{\rd_v r}{\Omg^2}\right|(u,v)\ls \left|\f{\rd_v r}{\Omg^2}\right|(u,-u) + \int_{-u}^v \Omg^{-2} |D_v\phi|^2(u,v') \, dv'.$$
By the estimates for $\rd_v r$ in the previous step and \eqref{eq:LinftyOmegav4}, we have
$$\left|\f{\rd_v r}{\Omg^2}\right|(u,-u)\ls \mathcal A_\phi(\mathcal{D}_{\rm o}+ \mathcal D_{\rm i}).$$
Since $v'\geq |u|$ in the domain of integration, we have, after using \eqref{eq:LinftyOmegav4} and \eqref{eq:bootstrapphi}, that
$$\int_{-u}^v \Omg^{-2} |D_v\phi|^2(u,v') \, dv'\ls \int_{-u}^v (v'+|u|)^2 |D_v\phi|^2(u,v') \, dv'\ls \int_{-u}^v v^2 |D_v\phi|^2(u,v') \, dv'\ls \mathcal A_\phi(\mathcal{D}_{\rm o}+\mathcal D_{\rm i}).$$
Combining the three estimates above with \eqref{eq:LinftyOmegav2} yields that when $v\geq |u|$, 
$$|\rd_v r|(u,v) \ls \mathcal A_\phi(\mathcal{D}_{\rm o}+ \mathcal D_{\rm i})\min\{v^{-2},|u|^{-2}\}.$$
Together with the previous step, we have thus completed the proof of \eqref{dvr.est}. \qedhere

\end{proof}

\section{Energy estimates}\label{sec:energy}

In this section, we prove the energy estimates for (derivatives of) $\phi$ and $\Om$. In particular, we will improve our bootstrap assumptions \eqref{eq:bootstrapOmega} and \eqref{eq:bootstrapphi}. As we discussed in the introduction, the argument leading to energy estimates for $\phi$ will go through the introduction of a \emph{renormalised energy}, the analysis of which forms the most technical part of the paper.

In \textbf{Section~\ref{sec:renorm.en}}, we will motivate the introduction of our renormalised energy for the matter field by considering the stress-energy-momentum tensor associated to the matter field. In \textbf{Section~\ref{sec:coercive.phi}}, we show that the renormalised energy we introduce is coercive; and in \textbf{Section~\ref{sec:eephi}}, we show that one can bound this renormalised energy. Combining these facts yields the desired control for the matter field.

Finally, in \textbf{Section~\ref{sec:ee.Om}}, we prove the energy estimates for $\f{\Om}{\Om_0}$.

\subsection{The stress energy tensor and the renormalised energy fluxes}\label{sec:renorm.en}
The null components of the stress-energy-momentum tensor $\mathbb{T}_{\mu\nu}$ corresponding to the scalar field and electromagnetic tensor are given by:
\begin{align*}
\mathbb{T}_{uv}=\:\frac{1}{4}\Omega^2\left(\mathfrak{m}^2|\phi|^2+\frac{1}{4\pi}r^{-4}Q^2\right),\quad \mathbb{T}_{uu}=\:|D_u\phi|^2,\quad
\mathbb{T}_{vv}=\: |D_v\phi|^2.
\end{align*}
The above expressions suggest that the natural energy fluxes along the null hypersurfaces of constant $u$ and $v$ are obtained by integrating the contraction $\mathbb{T}(X,\partial_{v})$ and $\mathbb{T}(X,\partial_u)$, respectively.  With the choice $X=u^2\partial_u+v^2\partial_v$, this gives:
\begin{align*}
&\int_{v_1}^{v_2} v^2r^2|D_v\phi|^2+ \frac{1}{4}u^2r^2\Omega^2\left(\mathfrak{m}^2|\phi|^2+\frac{1}{4\pi}r^{-4}Q^2\right)\,dv,\\
&\int_{u_1}^{u_2} u^2r^2|D_u\phi|^2+ \frac{1}{4}v^2r^2\Omega^2\left(\mathfrak{m}^2|\phi|^2+\frac{1}{4\pi}r^{-4}Q^2\right)\,du,
\end{align*}
with $-\infty<u_1<u_2\leq u_0$, $v_0 \leq v_1<v_2<\infty$.

However, with the above energy fluxes, the initial energy flux is not finite even initially on the event horizon with the chosen initial data. We therefore introduce the following \emph{renormalised} energy fluxes.\footnote{Notice that in the renormalisation, not only have we ``added an infinity term'' to each of the fluxes, we have also ``replaced several factors of $r$ by factors of $M$''. The replacement of $r$ by $M$ is strictly speaking not necessary to make the renormalised energy fluxes finite, but this simplifies the computations below.}
\begin{align}
\label{Ev.def} E_v(u):=&\:\int_{v_0}^{v_{\infty}} v^2M^2|D_v\phi|^2(u,v)+ \frac{1}{4}u^2M^2\Omega^2\left(\mathfrak{m}^2|\phi|^2+\frac{1}{4\pi}M^{-4}\left(Q^2-M^2\right)\right)(u,v)\,dv,\\
\label{Eu.def} E_u(v):=&\:\int_{-\infty}^{u_0} u^2M^2|D_u\phi|^2(u,v)+ \frac{1}{4}v^2M^2\Omega^2\left(\mathfrak{m}^2|\phi|^2+\frac{1}{4\pi}M^{-4}\left(Q^2-M^2\right)\right)(u,v)\,du.
\end{align}

\subsection{Coercivity of the renormalised energy flux}\label{sec:coercive.phi}

Our goal in this subsection is to prove that the renormalised energy flux we defined in \eqref{Ev.def} and \eqref{Eu.def} is coercive and controls the quantity on the LHS of \eqref{eq:bootstrapphi}. The statement of the main result can be found in Proposition~\ref{prop:coercenergy} at the end of the subsection. This will be achieved in a number of steps. We will first need the following preliminary results:
\begin{enumerate}
\item We need an improved version of \eqref{eq:LinftyQ}, which keeps track of the constants appearing in the leading-order terms (\textbf{Lemma~\ref{eq:preciseestQ}})
\item We then show that $\int_{v_0}^v \left(\frac{|u|}{v'+|u|}\right)^2|\phi|^2(u,v')\,dv'$ can be controlled by the LHS of \eqref{eq:bootstrapphi} (\textbf{Lemma~\ref{lm:Hardyv}}).
\item Similarly, we show that $\int^{u_0}_{-\infty} \left(\frac{v}{v+|u'|}\right)^2|\phi|^2(u',v)\,du'$ can be controlled by the LHS of \eqref{eq:bootstrapphi} (\textbf{Lemma~\ref{lm:Hardyu}}).
\end{enumerate}
Both 2. and 3. above are based on a Hardy-type inequality; see \textbf{Lemma~\ref{lm:hardy}}.
After these preliminary steps, we turn to the terms in the renormalised energy (cf.~\eqref{Eu.def}, \eqref{Ev.def}) which are not manifestly non-negative. Precisely, we control
\begin{itemize}
\item $\frac{1}{16\pi}\int_{v_0}^{v_{\infty}} u^2M^{-2}\Omega^2\left(Q^2-M^2\right)(u,v)\,dv$ in \textbf{Proposition~\ref{prop:coercive.error.v}}, and
\item $\frac{1}{16\pi}\int_{-\infty}^{u_0} v^2M^{-2}\Omega^2\left(Q^2-M^2\right)(u,v)\,du$ in \textbf{Proposition~\ref{prop:coercive.error.u}}.
\end{itemize}
Putting all these together, we thus obtain the main result in \textbf{Proposition~\ref{prop:coercenergy}}.

We now turn to the details, beginning with the following lemma:
\begin{lemma}
\label{eq:preciseestQ}
Let $\eta>0$ be as in \eqref{eta.def}. Then there exists $C>0$ such that
\begin{equation}
\label{eq:pointestQ2minM2}
\begin{split}
|Q^2-M^2|(u,v)\leq&\: 8\pi M|\mathfrak{e}|(1+\eta)  |u|^{-1}\int_{-\infty}^uu'^2M^2|D_u\phi|^2(u',v)\,du'+C\mathcal{D}_{\rm o}v^{-1}+C\mathcal A_\phi(\mathcal{D}_{\rm o}+\mathcal{D}_{\rm i})|u|^{-2}.
\end{split}
\end{equation}
\end{lemma}
\begin{proof}
We compute
\begin{equation*}
\begin{split}
(Q^2-M^2)(u,v)=&\:(Q(u,v)+M)(Q(u,v)-M)\\
=&\:(Q(u,v)+M)\left(Q(-\infty,v)-M+\int_{-\infty}^u\partial_uQ(u',v)\,du'\right)
\end{split}
\end{equation*}
and use \eqref{eq:Q1} to estimate
\begin{equation*}
|Q^2-M^2|(u,v)\leq |Q(u,v)+M|\left(|Q(-\infty,v)-M|+4\pi|\mathfrak{e}|\int_{-\infty}^ur^2|\phi||D_u\phi|\,du'\right).
\end{equation*}
We further use \eqref{eq:fundcalc1}  to estimate:
\begin{equation}
\label{eq:mainestQ1}
\begin{split}
&\:\int_{-\infty}^ur^2|\phi||D_u\phi|(u',v)\,du'\\
\leq&\: \left[M^2+\sup_{-\infty<u'\leq u}(r^2-M^2)\right]\cdot \sup_{-\infty<u'\leq u}|\phi|(u',v)\cdot \int_{-\infty}^u|D_u\phi|(u',v)\,du'\\
\leq& \: \left[M^2+\sup_{-\infty<u'\leq u}(r^2-M^2)\right]\cdot |u|^{-\frac{1}{2}}|\phi|(-\infty,v)\sqrt{\int_{-\infty}^uu'^2|D_u\phi|^2(u',v)\,du'}\\
& +\left[M^2+\sup_{-\infty<u'\leq u}(r^2-M^2)\right]\cdot |u|^{-\frac{1}{2}}\int_{-\infty}^u |D_u \phi|(u',v)\,du'\cdot \sqrt{\int_{-\infty}^uu'^2|D_u\phi|^2(u',v)\,du'}\\
\leq& \: \left[M^2+\sup_{-\infty<u'\leq u}(r^2-M^2)\right]\cdot |u|^{-\frac{1}{2}}|\phi|(-\infty,v)\sqrt{\int_{-\infty}^uu'^2|D_u\phi|^2(u',v)\,du'}\\
& +\left[M^2+\sup_{-\infty<u'\leq u}(r^2-M^2)\right]\cdot |u|^{-1} \int_{-\infty}^uu'^2|D_u\phi|^2(u',v)\,du'.
\end{split}
\end{equation}
Using \eqref{r.est}, it follows that
\begin{equation*}
|r^2-M^2|(u,v)=|r-M+2M|\cdot |r-M|\lesssim \mathcal{D}_{\rm o}v^{-1}+\mathcal A_\phi(\mathcal{D}_{\rm o}+ \mathcal{D}_{\rm i})|u|^{-1},
\end{equation*}
where we have taken $\mathcal{D}_{\rm o}v_0^{-1}$ and $(\mathcal{D}_{\rm o}+\mathcal D_{\rm i})|u_0|^{-1}$ to be suitably small.

Therefore, we can further estimate the right-hand side of \eqref{eq:mainestQ1} to obtain
\begin{equation}
\label{eq:mainestQ}
\begin{split}
\int_{-\infty}^ur^2|\phi||D_u\phi|(u',v)\,du'\leq&\: \frac{1}{4}\eta^{-1} M^2|\phi|^2(-\infty,v)+(1+\eta)M^2|u|^{-1} \int_{-\infty}^uu'^2|D_u\phi|^2(u',v)\,du'\\
&+C\mathcal{D}_{\rm o}v^{-1}|u|^{-1}+C\mathcal A_\phi(\mathcal{D}_{\rm o}+\mathcal D_{\rm i})|u|^{-2},
\end{split}
\end{equation}
where we moreover applied Young's inequality with an $\eta$ weight, where $\eta>0$ is as in the statement of the proposition.

By applying \eqref{eq:LinftyQ} and \eqref{eq:mainestQ} together with the initial data estimate \eqref{eq:oinitphi} and \eqref{eq:oinitQ}, we therefore obtain:
\begin{equation*}
\begin{split}
|Q^2-M^2|(u,v)\leq&\: |Q(u,v)+M|\Bigg(|Q(-\infty,v)-M|+\pi \eta^{-1}M^2|\mathfrak{e}| |\phi|^2(-\infty,v)\\
&+4\pi|\mathfrak{e}|(1+\eta)|u|^{-1} \int_{-\infty}^uu'^2M^2|D_u\phi|^2(u',v)\,du'+C\mathcal{D}_{\rm o}v^{-1}|u|^{-1}+C\mathcal A_\phi(\mathcal{D}_{\rm o}+\mathcal D_{\rm i})|u|^{-2}\Bigg)\\
\leq &\: \left(2M+C\mathcal{D}_{\rm o}v^{-1-\alp}+C\mathcal A_\phi(\mathcal{D}_{\rm o}+\mathcal D_{\rm i}) |u|^{-1}\right)\\
\cdot&\:\Bigg(4\pi|\mathfrak{e}|(1+\eta)|u|^{-1} \int_{-\infty}^uu'^2M^2|D_u\phi|^2(u',v)\,du'+C\mathcal{D}_{\rm o}v^{-1}+C\mathcal A_\phi(\mathcal{D}_{\rm o}+\mathcal D_{\rm i})|u|^{-2}\Bigg)\\
\leq &\: 8\pi M|\mathfrak{e}|(1+\eta)  |u|^{-1}\int_{-\infty}^uu'^2M^2|D_u\phi|^2(u',v)\,du'+C\mathcal{D}_{\rm o}v^{-1}+C\mathcal A_\phi(\mathcal{D}_{\rm o}+\mathcal D_{\rm i})|u|^{-2}.
\end{split}
\end{equation*}
\end{proof}

In order to estimate the $|\phi|^2$ integral, we will make use of a Hardy inequality:
\begin{lemma}
\label{lm:hardy}
Let $f:\R\to \R$ be in $H^1_{\rm loc}(\R)$ and let $c\geq 0$ be a constant. Then for any $x_1,x_2\in \R$ with $0<x_1<x_2<\infty$:
\begin{align*}
\int_{x_1}^{x_2} (x+c)^pf^2(x)\,dx\leq &\:2 (x_2+c)^{p+1}f^2(x_2)+\frac{4}{(p+1)^2}\int_{x_1}^{x_2} (x+c)^{p+2}f'^2(x)\,dx\quad \textnormal{if}\: p>-1\\
\int_{x_1}^{x_2} (x+c)^pf^2(x)\,dx\leq &\:2 (x_1+c)^{p+1}f^2(x_1)+\frac{4}{(p+1)^2}\int_{x_1}^{x_2} (x+c)^{p+2}f'^2(x)\,dx\quad \textnormal{if}\: p<-1.
\end{align*}
\end{lemma}
\begin{proof}
We will only prove the inequality in the case $p>-1$. The other inequality is similar. We compute
\begin{equation*}
\begin{split}
0\leq &\:\int_{x_1}^{x_2} (x+c)^p (\f{p+1}2 f+ (x+c)f')^2(x)\,dx \\
= &\: \int_{x_1}^{x_2} \f{(p+1)^2}4 (x+c)^p f^2 +  \f{p+1}2 (x+c)^{p+1} \f{d}{dx} (f^2) +  (x+c)^{p+2} f'^2(x)\,dx\\
=& \: -\f{(p+1)^2}4 \int_{x_1}^{x_2} (x+c)^p f^2(x)\, dx + \f{p+1}2 (x_2+c)^{p+1} f^2(x_2) \\
&\:- \f{p+1}2 (x_1+c)^{p+1} f^2(x_1)+ \int_{x_1}^{x_2} (x+c)^{p+2} f'^2(x)\,dx.
\end{split}
\end{equation*}
Rearranging and dropping a manifestly non-negative term yield the conclusion.
\end{proof}
Using Lemma~\ref{lm:hardy}, we prove a Hardy-type estimate on a constant-$u$ hypersurface in our setting.
\begin{lemma}\label{lm:Hardyv}
Given $\eta>0$ as in \eqref{eta.def}, there exists $C>0$ \textbf{independent of $\mathcal A_\phi$} such that the following holds:
\begin{equation*}
\begin{split}
&\:\int_{v_0}^v \left(\frac{|u|}{v'+|u|}\right)^2|\phi|^2(u,v')\,dv' \\
\leq &\: C\mathcal{D}_{\rm o}+ 4 \int_{v_0}^{v} v'^2 |D_v \phi|^2(u,v')\,dv'+ 3(1+\eta)\int_{-\infty}^u |u'|^2|D_u\phi|^2(u',|u|)\,du'.
\end{split}
\end{equation*}
\end{lemma}
\begin{proof}
We will apply the Hardy inequality Lemma~\ref{lm:hardy} on a fixed constant-$u$ hypersurface. For this purpose, it is useful to choose an auxiliary gauge for $A$ (cf.~Section~\ref{sec:global.gauge} and proof of Lemma~\ref{lm:fundthmcalc}) where $|\rd_v\phi|=|D_v\phi|$ (and we will perform all estimates in terms of the gauge invariant quantities $|\phi|$ and $|D_v\phi|$). Let $\gamma>0$ be a constant to be chosen. If $v>|u|$, we also partition the integration interval: $[v_0,v]=[v_0,\gamma |u|]\cup [\gamma |u|,v]$, so we can estimate:
\begin{equation}\label{lm:Hardyv.1}
\begin{split}
& \:\int_{v_0}^v \left(\frac{|u|}{v'+|u|}\right)^2|\phi|^2(u,v')\,dv'\\
=&\:\int_{v_0}^{\gamma|u|} \left(\frac{|u|}{v'+|u|}\right)^2|\phi|^2(u,v')\,dv'+\int_{\gamma|u|}^v \left(\frac{|u|}{v'+|u|}\right)^2|\phi|^2(u,v')\,dv'\\
\leq &\: \int_{v_0}^{\gamma|u|} |\phi|^2\,dv'+|u|^2\int_{\gamma|u|}^v(v'+|u|)^{-2}|\phi|^2(u,v')\,dv'\\
\leq &\: 4\int_{v_0}^{\gamma|u|}v'^2 |D_v \phi|^2(u,v')\,dv'+4|u|^2\int_{\gamma|u|}^v |D_v \phi|^2(u,v')\,dv'+ 2\gamma|u||\phi|^2(u,\gamma|u|)+\frac{2|u|}{(\gamma+1)}|\phi|^2(u,\gamma|u|)\\
\leq &\: 4 \max\{1,\gamma^{-2}\} \int_{v_0}^{v} v'^2 |D_v \phi|^2(u,v')\,dv'+\left(2\gamma+\f{2}{\gamma+1}\right)|u||\phi|^2(u,\gamma|u|),
\end{split}
\end{equation}
where we applied Lemma~\ref{lm:hardy} with $p=0$ and $p=-2$ respectively in the two intervals. By \eqref{eq:fundcalc1}, Young's inequality and \eqref{eq:oinitphi} (to control the initial data), it moreover holds that for any $\eta>0$, there exists $C_{\eta}>0$ so that the following is verified:
\begin{equation}
\label{eq:estphispacelikecurve}
|u||\phi|^2(u,v)\leq C_{\eta}\mathcal{D}_{\rm o}\frac{|u|}{v}+(1+\eta)\int_{-\infty}^u |u'|^2|D_u\phi|^2(u',v)\,du'.
\end{equation}
Using this to control the last term in \eqref{lm:Hardyv.1}, and setting $\gamma=1$, we obtain the desired conclusion.
\end{proof}

We have a similar Hardy-type estimate on a constant-$v$ hypersurface, with appropriate modifications of the weights.
\begin{lemma}\label{lm:Hardyu}
Let $\eta>0$ be as in \eqref{eta.def}. Then there exists $C>0$ \textbf{independent of $\mathcal A_\phi$} such that the following holds:
$$\int^{u_0}_{-\infty} \left(\frac{v}{v+|u'|}\right)^2|\phi|^2(u',v)\,du' \leq C \mathcal{D}_{\rm o} + 6(1+\eta) \int^{u_0}_{-\infty} u'^2 |D_u \phi|^2(u',v)\,du'. $$
\end{lemma}
\begin{proof}
Choose a gauge so that $|D_u\phi|=|\rd_u\phi|$. (As in Lemma~\ref{lm:Hardyv}, we will only estimate the gauge invariant quantities.)

Let $\sigma>0$ be a constant that will be chosen suitably later on. We partition the integration interval: $[-\infty,u_0]=[-\infty,\sigma v]\cup [\sigma v,u_0]$. For the sake of convenience, we will change our integration variable from $u$ to $-u$ and we will denote $-u$ by $|u|$, so we can estimate:
\begin{equation*}
\begin{split}
&\:\int_{|u_0|}^{\infty} \left(\frac{v}{v+|u'|}\right)^2|\phi|^2(u',v)\,d|u'|\\
=&\:\int_{|u_0|}^{\sigma v} \left(\frac{v}{v+|u'|}\right)^2|\phi|^2(u',v)\,d|u'|+\int_{\sigma v}^{\infty} \left(\frac{v}{v+|u'|}\right)^2|\phi|^2\,d|u'|\\
\leq &\: \int_{|u_0|}^{\sigma v} |\phi|^2(u',v)\,d|u'|+v^2\int_{\sigma v}^{\infty}(v+|u'|)^{-2}|\phi|^2(u',v)\,d|u'|\\
\leq &\: 4\int_{|u_0|}^{\sigma v}|u'|^2 |D_u\phi|^2(u',v)\,d|u'|+4v^2\int_{\sigma v}^{\infty}|D_u \phi|^2(u',v)\,d|u'|+2 \sigma v|\phi|^2(\sigma v,v)+2\frac{v^2}{\sigma v+v}|\phi|^2(\sigma v,v)\\
\leq &\: 4\max\{1,\sigma^{-2}\}\int^{u_0}_{-\infty} u'^2 |D_u \phi|^2(u',v)\,du'+\left(2+\frac{2}{\sigma(1+\sigma)}\right)\sigma v|\phi|^2(\sigma v,v),
\end{split}
\end{equation*}
where we applied Lemma~\ref{lm:hardy} with $p=0$ and $p=-2$ in the two intervals. We apply \eqref{eq:estphispacelikecurve} to conclude that for any $\eta',\,\sigma>0$, there exists $C_{\eta',\sigma}>0$ such that
\begin{equation*}
 \int_{u}^{u_0}\left(\frac{v}{v+|u'|}\right)^2|\phi|^2\,du'\leq \left(4\max\{1,\sigma^{-2}\}+(1+\eta')\left(2+\frac{2}{\sigma(1+\sigma)}\right)\right)\int^{u_0}_{-\infty}u'^2 |D_u \phi|^2\,du'+C_{\eta',\sigma}\mathcal{D}_{\rm o}.
\end{equation*}
Finally, choosing $\eta'$ sufficiently small and $\sigma$ sufficiently large, we can choose
$$4\max\{1,\sigma^{-2}\}+4\sigma^{-2}+(1+\eta')\left(2+\frac{2}{\sigma(1+\sigma)}\right) \leq 6(1+\eta),$$
which then gives the desired conclusion. \qedhere

\end{proof}

With Lemmas~\ref{lm:Hardyv} and \ref{lm:Hardyu} in place, we are now ready to control the terms in $E_v(u)$ and $E_u(v)$ which are not manifestly non-negative. These are the terms 
$$\frac{1}{16\pi}\int_{v_0}^{v_{\infty}} u^2M^{-2}\Omega^2\left(Q^2-M^2\right)(u,v)\,dv,\quad \frac{1}{16\pi}\int_{-\infty}^{u_0} v^2M^{-2}\Omega^2\left(Q^2-M^2\right)(u,v)\,du$$ in \eqref{Ev.def} and \eqref{Eu.def}, which will be handled in Propositions~\ref{prop:coercive.error.v} and \ref{prop:coercive.error.u} respectively.

\begin{proposition}\label{prop:coercive.error.v}
For $\eta>0$ as in \eqref{eta.def} and for $\kappa>0$ arbitrary, there exists a constant $C>0$ \textbf{independent of $\mathcal A_{\phi}$} (but dependent on $\kappa$ in addition to $M$, $\mfe$, $\mfm$, $\alp$ and $\eta$) such that 
\begin{equation*}
\begin{split}
&\:\left| \frac{1}{16 \pi}\int_{v_0}^v u^2M^{-2}\Omega^2\left(Q^2-M^2\right)(u,v')\,dv'\right|\\
\leq &\:(\kappa+4\kappa^{-1}) |\mathfrak{e}|M^3 \int_{v_0}^vv'^2|D_v\phi|^2\,dv'+(2+3\kappa^{-1})(1+\eta)|\mathfrak{e}|M^3 \int_{-\infty}^u |u'|^2|D_u\phi|^2\,du' + C(\mathcal{D}_{\rm o}+\mathcal D_{\rm i}).
\end{split}
\end{equation*}
\end{proposition}
\begin{proof}

\textbf{The main integration by parts.} We begin the argument by a simple integration by parts. 
\begin{equation}
\begin{split}
\label{eq:intbypartsv}
&\: \left| \frac{1}{16 \pi}\int_{v_0}^v u^2M^{-2}\Omega^2\left(Q^2-M^2\right)(u,v')\,dv'\right|\\
\leq &\:\left|-\frac{1}{16 \pi}\int_{v_0}^v\left(\int_{v_0}^{v'}u^2M^{-2}\Omega^2(u,v')\,dv''\right)\partial_v\left(Q^2-M^2\right)(u,v')\,dv'\right| \\
&\:+\left|\frac{1}{16 \pi}\left(\int_{v_0}^{v'}u^2M^{-2}\Omega^2(u,v'')\,dv''\right)\left(Q^2-M^2\right)(u,v')\bigg|^{v'=v}_{v'=v_0}\right|\\
\leq &\:\frac{1}{16 \pi}\left|\int_{v_0}^v u^2M^{-2}\Omega^2(u,v')\,dv'\right||Q^2-M^2|(u,v)\\
&\:+\frac{1}{8 \pi}\int_{{v_0}}^v \left|\int_{v_0}^{v'} u^2M^{-2}\Omega^2(u,v'')\,dv'''\right||Q||\partial_vQ|(u,v')\,dv' .
\end{split}
\end{equation}
Note that there are two terms that we need to estimate.

\textbf{An auxiliary computation.} Before we proceed, we first estimate the integral
$$\int_{v_0}^v u^2M^{-2}\Omega^2(u,v')\,dv'$$
which appears in \eqref{eq:intbypartsv}. From \eqref{eq:LinftyOmegav2} it follows that
\begin{equation}\label{eq:Linfty.Omega.cons}
M^{-2}\Omega^2\leq 4(v+|u|)^{-2}+C|u|^{-\frac{1}{2}}(v+|u|)^{-2},
\end{equation}
and hence,
\begin{equation*}
\begin{split}
\int_{v_0}^v u^2M^{-2}\Omega^2(u,v')\,dv'\leq&\:4\int_{v_0}^v \left(\frac{|u|}{v'+|u|}\right)^2\,dv'+C|u|^{-\frac{1}{2}}\int_{v_0}^v\left(\frac{|u|}{v'+|u|}\right)^2\,dv'.
\end{split}
\end{equation*}
Note that
\begin{equation}
\label{eq:uvderivatives}
\partial_v(|u|v(v+|u|)^{-1})=\left(\frac{|u|}{v+|u|}\right)^2.
\end{equation}
Hence,
\begin{equation}
\label{eq:auxintegralest}
\begin{split}
\int_{v_0}^v u^2M^{-2}\Omega^2(u,v')\,dv'\leq&\:4|u|v(v+|u|)^{-1}+C|u|^{\frac{1}{2}}v(v+|u|)^{-1}.
\end{split}
\end{equation}

\textbf{Estimating the boundary term.} We first estimate the boundary term in \eqref{eq:intbypartsv} above by using \eqref{eq:pointestQ2minM2} and \eqref{eq:auxintegralest} 
\begin{equation}\label{eq:coercive.bdryv}
\begin{split}
&\:\frac{1}{16 \pi}\left|\int_{v_0}^v u^2M^{-2}\Omega^2(u,v')\,dv'\right||Q^2-M^2|(u,v)\\
\leq&\: 2(1+\eta)M|\mathfrak{e}| \int_{-\infty}^uu^2|D_u\phi|^2(u',v)M^2\,du'+C\mathcal{D}_{\rm o}+C\mathcal A_{\phi}(\mathcal{D}_{\rm o}+\mathcal{D}_{\rm i})|u|^{-1}+C(\mathcal{D}_{\rm o}+\mathcal{D}_{\rm i})|u|^{-\frac{1}{2}}.
\end{split}
\end{equation}

\textbf{Estimating the remaining integral I: the error term.} Now, we estimate the remaining integral on the very right-hand side of \eqref{eq:intbypartsv} by applying \eqref{eq:Q2} and \eqref{eq:auxintegralest}:
\begin{equation}\label{remainingintegralv}
\frac{1}{8 \pi}\int_{v_0}^v \left|\int_{v_0}^{v'} u^2M^{-2}\Omega^2(u,v'')\,dv''\right||Q||\partial_vQ|(u,v')\,dv'\leq \underbrace{2M^3 |\mathfrak{e}|\int_{v_0}^v |u|v'(v'+|u|)^{-1}|\phi||D_v\phi|\,dv'}_{Main\,\,term}+\textnormal{Err},
\end{equation}
where
\begin{equation}\label{Err1.def}
\begin{split}
\textnormal{Err}=&\: C\int_{v_0}^v |u|v'(v'+|u|)^{-1}|Q-M ||\phi||D_v\phi|(u,v')\,dv'\\
&+C\int_{v_0}^v |M+(Q-M)| |u|^{\f 12}v' (v'+|u|)^{-1} |\phi||D_v\phi|(u,v')\,dv'\\
&+C\int_{v_0}^v |u|v'(v'+|u|)^{-1}|r^2-M^2||\phi||D_v\phi|(u,v')\,dv'.
\end{split}
\end{equation}
Let us start with the error term, beginning with the second term \eqref{Err1.def}, which is the hardest since the decay is the weakest. By \eqref{eq:Linftyphi},
\begin{equation*}
\begin{split}
&\: \int_{v_0}^v |M+(Q-M)| |u|^{\f 12} v (v+|u|)^{-1} |\phi||D_v\phi|(u,v')\,dv'\\
\ls & \:\sqrt{\mathcal{D}_{\rm o}} \int_{v_0}^v |u|^{\f 12} v^{\f 12-\f{\alpha}{2}} (v+|u|)^{-1} |D_v\phi|(u,v')\,dv' + \sqrt{\mathcal A_{\phi}(\mathcal{D}_{\rm o} + \mathcal D_{\rm i})} \int_{v_0}^v v (v+|u|)^{-1} |D_v\phi|(u,v')\,dv' \\
\ls & \:\sqrt{\mathcal A_{\phi}(\mathcal{D}_{\rm o} +\mathcal D_{\rm i})} (v_0+|u|)^{-\f 14} \left(\int_{v_0}^v  v'^{-\f 32}\,dv'\right)^{\f 12}\left(\int_{v_0}^v v'^2 |D_v\phi|^2 (u,v')\,dv'\right)^{\f 12} \\
\ls & \: \mathcal A_\phi(\mathcal{D}_{\rm o} +\mathcal D_{\rm i}) (v_0+|u|)^{-\f 14} v_0^{-\f 14}.
\end{split}
\end{equation*}

The first term in \eqref{Err1.def} can be treated similarly, with the only caveat that there is a contribution where $|\phi|$ and $|Q-M|$ only give $v$-decay, and therefore one does not get any smallness in $(v_0+|u|)^{-\f 14}$. Nevertheless, this term has a coefficient that depends only on $\mathcal{D}_{\rm o}$ (and not on $\mathcal D_{\rm i}$). More precisely, using \eqref{eq:Linftyphi} and \eqref{eq:LinftyQ}, we have
\begin{equation*}
\begin{split}
& \:\int_{v_0}^v |u|v'(v'+|u|)^{-1}|Q-M ||\phi||D_v\phi|(u,v')\,dv'\\
\ls & \:\mathcal{D}_{\rm o}^{\f 32} \int_{v_0}^v v'^{-\f 12}|D_v\phi|(u,v')\,dv' + \mathcal A_\phi^2(\mathcal{D}_{\rm o} +\mathcal D_{\rm i})^2 (v_0+|u|)^{-\f 14} v_0^{-\f 14}\\
\ls & \:\mathcal{D}_{\rm o}^{\f 32} v_0^{-\f 12}\left(\int_{v_0}^v v'^2|D_v\phi|^2(u,v')\,dv'\right)^{\f 12} + \mathcal A_\phi^2(\mathcal{D}_{\rm o} +\mathcal D_{\rm i})^2 (v_0+|u|)^{-\f 14} v_0^{-\f 14}\\
\ls & \: \mathcal A_\phi^{\f 12}\mathcal{D}_{\rm o}^2 v_0^{-\f 12} + \mathcal A_\phi^2(\mathcal{D}_{\rm o} +\mathcal D_{\rm i})^2 (v_0+|u|)^{-\f 14} v_0^{-\f 14}.
\end{split}
\end{equation*}
Finally, the third term in \eqref{Err1.def} can be handled in an identical manner as the first term, except for using \eqref{r.est} instead of \eqref{eq:LinftyQ}, so that we have
\begin{equation*}
\begin{split}
\int_{v_0}^v |u|v'(v'+|u|)^{-1}|r^2-M^2||\phi||D_v\phi|(u,v')(u,v')\,dv'
\ls & \: \mathcal A_\phi^{\f 12}\mathcal{D}_{\rm o}^2 v_0^{-\f 12} + \mathcal A_\phi^2(\mathcal{D}_{\rm o} +\mathcal D_{\rm i})^2 (v_0+|u|)^{-\f 14} v_0^{-\f 14}.
\end{split}
\end{equation*}
Putting all these together, choosing $|u_0|$ and $v_0$ appropriately large (where the largeness of $v_0$ does \emph{not} depend on $\mathcal D_{\rm i}$), and returning to \eqref{Err1.def}, we thus obtain
\begin{equation}\label{Errv}
\begin{split}
\textnormal{Err}\lesssim \mathcal{D}_{\rm o}+\mathcal{D}_{\rm i}.
\end{split}
\end{equation}

\textbf{Estimating the remaining integral II: the main term.} Now for the main term in \eqref{remainingintegralv}, we have, using Young's inequality,
\begin{equation}\label{eq:maintermv}
\begin{split}
&\: 2M^3 |\mathfrak{e}|\int_{v_0}^v |u|v'(v'+|u|)^{-1}|\phi||D_v\phi|\,dv'\\
\leq &\: \kappa |\mathfrak{e}| M^3 \int_{v_0}^vv'^2|D_v\phi|^2\,dv'+\kappa^{-1}|\mathfrak{e}| M^3 \int_{v_0}^v \left(\frac{|u|}{v'+|u|}\right)^2|\phi|^2\,dv'\\
\leq &\:  (\kappa+4\kappa^{-1}) |\mathfrak{e}|M^3\int_{v_0}^vv'^2|D_v\phi|^2\,dv'+3(1+\eta)\kappa^{-1}|\mathfrak{e}|M^3 \int_{-\infty}^u |u'|^2|D_u\phi|^2\,du' + C_{\kappa} (\mathcal{D}_{\rm o}+\mathcal D_{\rm i}),
\end{split}
\end{equation}
where in the last line we have used Lemma~\ref{lm:Hardyv}.

\textbf{Putting everything together.}  Combining \eqref{eq:intbypartsv}, \eqref{eq:coercive.bdryv}, \eqref{remainingintegralv}, \eqref{Errv}and \eqref{eq:maintermv}, we obtain the desired conclusion.
\end{proof}

We now turn to the analogue of Proposition~\ref{prop:coercive.error.v} on constant-$v$ hypersurfaces.
\begin{proposition}\label{prop:coercive.error.u}
For $\eta>0$ as in \eqref{eta.def}, there exists a constant $C>0$ \textbf{independent of $\mathcal A_{\phi}$} such that
\begin{equation*}
\begin{split}
&\:\left|\frac{1}{16 \pi}\int_{-\infty}^{u_0} v^2M^{-2}\Omega^2\left(Q^2-M^2\right)(u',v)\,du' \right|\\
\leq &\:2(1+\sqrt{6})|\mfe| M^3(1+\eta) \int_{-\infty}^{u_0} |u'|^2|D_u\phi|^2(u',v)\,du' + C(\mathcal{D}_{\rm o} + \mathcal D_{\rm i}).
\end{split}
\end{equation*}
\end{proposition}
\begin{proof}
We will consider the integral
\begin{equation*}
\frac{1}{16 \pi}\int_{u}^{u_0} v^2M^{-2}\Omega^2\left(Q^2-M^2\right)(u',v)\,du'
\end{equation*}
and take the limit $u\downarrow -\infty$.

\textbf{The main integration by parts.} We integrate by parts as in \eqref{eq:intbypartsv}
\begin{equation}
\begin{split}
\label{eq:intbypartsu}
&\:\left|\frac{1}{16 \pi}\int_{u}^{u_0} v^2M^{-2}\Omega^2\left(Q^2-M^2\right)(u',v)\,du'\right|\\
\leq &\:\left|\frac{1}{16 \pi}\int_{u}^{u_0}\left(\int_{u'}^{u_0}v^2M^{-2}\Omega^2\,du''\right)\partial_u\left(Q^2-M^2\right)(u',v)\,du'\right|\\
&+\left|\frac{1}{16 \pi}\left(\int_{u'}^{u_0}v^2M^{-2}\Omega^2(u'',v)\,du''\right)\left(Q^2-M^2\right)(u',v)\bigg|^{u'=u_0}_{u'=u}\right|\\
\leq &\:\frac{1}{16 \pi}\int_{u}^{u_0}v^2M^{-2}\Omega^2(u',v)\,du' \cdot|Q^2-M^2|(u,v)\\
&+\frac{1}{8 \pi}\int_{u}^{u_0} \left[\int_{\infty}^{u'} v^2M^{-2}\Omega^2(u'',v)\,du''\right] |Q||\partial_uQ|(u',v)\,du' .
\end{split}
\end{equation}

\textbf{An auxiliary computation.} By \eqref{eq:Linfty.Omega.cons} and
$$-\rd_u(|u|v(v+|u|)^{-1}) = \left(\f{v^2}{v+|u'|}\right)^2,$$
we obtain
\begin{equation}\label{eq:estuintOmega}
\begin{split}
\int_{u}^{u_0} v^2M^{-2}\Omega^2(u',v)\,du' \leq &\: 4\int_{u}^{u_0} \left(\f{v^2}{v+|u'|}\right)^2\, du' + C v^2\int_{u}^{u_0} \f{1}{|u'|^{\f 12}(v+|u'|)^2}\, du'\\
\leq &\: 4|u|v(v+|u|)^{-1} + C|u|^{\f 12} v^2(v+|u|)^{-2}.
\end{split}
\end{equation}

\textbf{Estimating the boundary term.}
We now control the boundary term in \eqref{eq:intbypartsu}. By \eqref{eq:estuintOmega} combined with \eqref{eq:pointestQ2minM2}, we obtain 
\begin{equation}\label{eq:coercive.bdryu}
\begin{split}
&\:\frac{1}{16 \pi}\left|\int_{u}^{u_0} v^2M^{-2}\Omega^2(u',v)\,du'\right||Q^2-M^2|(u,v)\\
\leq&\: 2(1+\eta)M^3|\mathfrak{e}| \int_{-\infty}^u u'^2|D_u\phi|^2(u',v)\,du'+C\mathcal{D}_{\rm o}+C\mathcal{A}_{\phi}(\mathcal{D}_{\rm o}+\mathcal{D}_{\rm i})|u|^{-1}+C(\mathcal{D}_{\rm o}+\mathcal D_{\rm i})|u|^{-\frac{1}{2}}.
\end{split}
\end{equation}

\textbf{Estimating the remaining integral.} The remaining integral in \eqref{eq:intbypartsu} can be controlled as follows:
\begin{equation*}
\frac{1}{8 \pi}\int_{u}^{u_0} \left|\int_{u'}^{u_0} v^2M^{-2}\Omega^2(u'',v)\,du''\right||Q||\partial_uQ|(u',v)\,du'\leq \underbrace{2M^3 |\mathfrak{e}|\int_{u}^{u_0} |u'|v(v+|u'|)^{-1}|\phi||D_u\phi|(u',v)\,du'}_{Main\,\,term}+\textnormal{Err},
\end{equation*}
where
\begin{equation*}
\begin{split}
\textnormal{Err}=&\: C\int^{u_0}_u |u'|v(v+|u'|)^{-1}|Q-M ||\phi||D_u\phi|(u',v)\,du'\\
&+C\int^{u_0}_u |M+(Q-M)| |u'|^{\f 12} v (v+|u'|)^{-1} |\phi||D_u\phi|(u',v)\,du'\\
&+C\int^{u_0}_u |u'|v(v+|u'|)^{-1}|r^2-M^2||\phi||D_u\phi|(u',v)\,du'.
\end{split}
\end{equation*}
The error term can be estimated in essentially the same manner as the error term in the proof of Proposition~\ref{prop:coercive.error.v}, except we use \eqref{eq:bootstrapphi} for $\int^{u_0}_u u'^2|D_u\phi|^2(u',v)\,du'$ instead of $\int_{v_0}^v v'^2|D_v\phi|^2(u,v')\,dv'$. We omit the details and just record the following estimate:
\begin{equation}\label{eq:coercive.erru}
\textnormal{Err}\lesssim \mathcal{D}_{\rm o}+\mathcal{D}_{\rm i}.
\end{equation}
For the main term in \eqref{eq:intbypartsu}, we apply H\"older's inequality, Lemma~\ref{lm:Hardyu} and Young's inequality to obtain
\begin{equation}\label{eq:coercive.mainterm.u}
\begin{split}
&\: 2M^3 |\mathfrak{e}|\int_{u}^{u_0} |u'|v(v+|u'|)^{-1}|\phi||D_u\phi|(u',v)\,du'\\
\leq &\: 2|\mathfrak{e}|M^3 \left(\int_{u}^{u_0}|u'|^2|D_u\phi|^2(u',v)\,du'\right)^{\f 12}\left(\int_{u}^{u_0}\left(\frac{v}{v+|u'|}\right)^2|\phi|^2(u',v)\,du'\right)^{\f 12}\\
\leq &\: 2|\mathfrak{e}|M^3 \sqrt{6}(1+\eta)\int_{u}^{u_0}|u'|^2|D_u\phi|^2(u',v)\,du' + C\mathcal{D}_{\rm o}.
\end{split}
\end{equation}

\textbf{Putting everything together.} Putting together \eqref{eq:intbypartsu}, \eqref{eq:coercive.bdryu}, \eqref{eq:coercive.erru} and \eqref{eq:coercive.mainterm.u}, and taking the limit $u\downarrow -\infty$, we obtain the desired conclusion. \qedhere
\end{proof}

We can now prove the main result of this subsection, namely, the coercivity of the renormalised energy flux (up to controllable error terms).
\begin{proposition}
\label{prop:coercenergy}
We can estimate
\begin{equation}
\label{eq:coercenergy}
\begin{split}
\int_{-\infty}^{u_0} |u|^2M^2 |D_u\phi|^2\,du+ \int_{v_0}^{v_{\infty}} v^2M^2 |D_v\phi|^2\,dv\leq \:&\mu^{-1}\left(E_u(v)+E_v(u)\right)+C(\mathcal{D}_{\rm o}+\mathcal{D}_{\rm i}),
\end{split}
\end{equation}
with $\mu$ as in \eqref{mu.def}.
\end{proposition}
\begin{proof}

Plugging in the estimates in Propositions~\ref{prop:coercive.error.v} and \ref{prop:coercive.error.u} into \eqref{Ev.def} and \eqref{Eu.def} respectively, and using $\mfm\geq 0$, we deduce that
\begin{equation*}
\begin{split}
E_u(v)+E_v(u)\geq&\: \left(1-\left(4+2\sqrt{6}+3\kappa^{-1}\right)(1+\eta)|\mathfrak{e}|M \right) \int_{-\infty}^{u_0} |u|^2r^2 |D_u\phi|^2\,du'\\
&+ \left(1- (\kappa+4\kappa^{-1})|\mathfrak{e}|M\right) \int_{v_0}^{v_{\infty}} v^2M^2 |D_v\phi|^2\,dv'-C(\mathcal{D}_{\rm o}+\mathcal{D}_{\rm i}).
\end{split}
\end{equation*}
Now we choose $\kappa=2 + \sqrt{6} + \sqrt{9 + 4 \sqrt{6}}$ to obtain the conclusion.
\end{proof}

\subsection{Energy estimates for $\phi$}\label{sec:eephi}
We define
\begin{align*}
E_u(v;u):=&\:\int_{-\infty}^{u} u^2M^2|D_u\phi|^2(u',v)+ \frac{1}{4}v^2M^2\Omega^2\left(\mathfrak{m}^2|\phi|^2+\frac{1}{4\pi}M^{-4}\left(Q^2-M^2\right)\right)(u',v)\,du',\\
E_v(u;v):=&\:\int_{v_0}^{v} v^2M^2|D_v\phi|^2(u,v')+ \frac{1}{4}u^2M^2\Omega^2\left(\mathfrak{m}^2|\phi|^2+\frac{1}{4\pi}M^{-4}\left(Q^2-M^2\right)\right)(u,v')\,dv'.
\end{align*}
With the above definitions, we have that $E_u(v)=E_u(v;u_0)$ and $E_v(u)=E_v(u;v_{\infty})$.

\begin{proposition}
\label{prop:eestimatephi}
$E_u(v;u)$ and $E_v(u;v)$ obey the following estimate:
\begin{equation*}
\sup_{v_0\leq v\leq v_{\infty}}E_u(v;u)+\sup_{-\infty<u<u_0}E_v(u;v)\leq C(\mathcal{D}_{\rm o}+\mathcal{D}_{\rm i}).
\end{equation*}
\end{proposition}
\begin{proof}
In order to simplify the notation, in this proof, we omit the arguments in the integrals, which will typically be taken as $(u',v')$. For any $u\in (-\infty, u_0]$ and $v\in [v_0,v_{\infty}]$, we decompose:
\begin{equation*}
\begin{split}
[E_u(v;u)-E_u(v_0;u)]+[E_v(u;v)-E_v(-\infty;v)]=&\:\int_{v_0}^v \partial_v E_u(v';u)\,dv'+\int_{-\infty}^u \partial_uE_v(u';v)\,du'\\
=&\:J_1+J_2+J_3+J_4+J_5+J_6,
\end{split}
\end{equation*}
where
\begin{align*}
J_1=&\:M^2\int_{v_0}^v \int_{-\infty}^u u'^2  \partial_v(|D_u\phi|^2)\,du'dv',\\
J_2=&\:\frac{1}{4} M^2 \mathfrak{m}^2\int_{v_0}^v \int_{-\infty}^u  v'^2 \Omega^2 \partial_v(|\phi|^2)+\partial_v(v'^2\Omega^2)\cdot |\phi|^2\,du'dv',\\
J_3=&\: \frac{1}{16\pi} M^{-2}\int_{v_0}^v \int_{-\infty}^u 2v'^2\Omega^2Q\partial_vQ+\partial_v(v'^2\Omega^2)(Q^2-M^2)\,du'dv',\\
J_4=&\:M^2\int_{v_0}^v \int_{-\infty}^u v'^2 M^2 \partial_u(|D_v\phi|^2)\,du'dv',\\
J_5=&\:\frac{1}{4} M^2 \mathfrak{m}^2\int_{v_0}^v \int_{-\infty}^u u'^2 \Omega^2 \partial_u(|\phi|^2)+\partial_u(u'^2\Omega^2)\cdot |\phi|^2\,du'dv',\\
J_6=&\: \frac{1}{16\pi} M^{-2}\int_{v_0}^v \int_{-\infty}^u 2u'^2\Omega^2Q\partial_uQ+\partial_u(u'^2\Omega^2)(Q^2-M^2)\,du'dv'.
\end{align*}

We first use equations \eqref{eq:phi1} and \eqref{eq:phi2} to rewrite the integral $J_1$ in terms of expressions that are zeroth- or first-order derivatives of the variables $\phi,\Omega,r$. For this, we use that we have the following identity for complex-valued functions $f$:
\begin{equation*}
\partial_v(|f|^2)=(D_vf-i\mathfrak{e} A_vf)\bar{f}+f\overline{(D_vf-i\mathfrak{e} A_vf)}=\bar{f}D_vf+f\overline{D_vf}
\end{equation*}
and similarly
\begin{equation*}
\partial_u(|f|^2)=\bar{f}D_uf+f\overline{D_uf}.
\end{equation*}

We therefore obtain:
\begin{equation*}
\begin{split}
J_1=&\:\int_{v_0}^v \int_{-\infty}^u u'^2 M^2\left(\overline{D_u\phi}D_vD_u\phi+D_u\phi \overline{D_vD_u \phi}\right)\,du'dv'\\
=&\:-\frac{1}{2}\int_{v_0}^v \int_{-\infty}^u \frac{1}{2}u'^2 M^2\mathfrak{m}^2\Omega^2(\phi\overline{D_u\phi}+\bar{\phi}{D_u\phi} )+2M^2r^{-1}u'^2\partial_ur(\overline{D_u\phi}D_v\phi+\overline{D_v\phi}D_u\phi)\\
&+4M^2r^{-1}u'^2\partial_vr|D_u\phi|^2+\frac{1}{2} i\mathfrak{e} M^2r^{-2}u'^2\Omega^2Q(\phi \overline{D_u\phi}-\bar{\phi}{D_u\phi}) \,du'dv'
\end{split}
\end{equation*}
and similarly,
\begin{equation*}
\begin{split}
J_4=&\:\int_{v_0}^v \int_{-\infty}^u v'^2 M^2\left(\overline{D_v\phi}D_uD_v\phi+D_v\phi \overline{D_uD_v \phi}\right)\,du'dv'\\
=&\:-\frac{1}{2}\int_{v_0}^v \int_{-\infty}^u \frac{1}{2}v'^2 M^2\mathfrak{m}^2\Omega^2(\phi\overline{D_v\phi}+\bar{\phi}{D_v\phi} )+2M^2r^{-1}v'^2\partial_vr(\overline{D_u\phi}D_v\phi+\overline{D_v\phi}D_u\phi)\\
&+4M^2r^{-1}v'^2\partial_ur|D_v\phi|^2-\frac{1}{2} i\mathfrak{e}M^2r^{-2}\Omega^2v'^2 Q(\phi \overline{D_v\phi}-\bar{\phi}{D_v\phi}) \,du'dv'.
\end{split}
\end{equation*}
We also rewrite $J_2$ and $J_5$ to obtain
\begin{align*}
J_2=&\:\frac{1}{4} M^2 \mathfrak{m}^2\int_{v_0}^v \int_{-\infty}^u  v'^2 \Omega^2 (\bar{\phi}D_v\phi+\phi\overline{D_v\phi})+\partial_v(v'^2\Omega^2)\cdot |\phi|^2\,du'dv',\\
J_5=&\:\frac{1}{4} M^2 \mathfrak{m}^2\int_{v_0}^v \int_{-\infty}^u  u'^2 \Omega^2 (\bar{\phi}D_u\phi+\phi\overline{D_u\phi})+\partial_u(u'^2\Omega^2)\cdot |\phi|^2\,du'dv'.
\end{align*}
Finally, we use equations \eqref{eq:Q1} and \eqref{eq:Q2} to rewrite $J_3$ and $J_6$:
\begin{align*}
J_3=&\:\int_{v_0}^v \int_{-\infty}^u -\frac{1}{4}i\mathfrak{e}M^{-2}r^2v'^2\Omega^2Q(\phi \overline{D_v\phi}-\bar{\phi}D_v\phi)+\frac{1}{16\pi} M^{-2}\partial_v(v'^2\Omega^2)(Q^2-M^2)\,du'dv',\\
J_6=&\:\int_{v_0}^v \int_{-\infty}^u \frac{1}{4}i\mathfrak{e}M^{-2}r^2u'^2\Omega^2Q(\phi \overline{D_u\phi}-\bar{\phi}D_u\phi)+\frac{1}{16\pi} M^{-2}\partial_u(u'^2\Omega^2)(Q^2-M^2)\,du'dv'.
\end{align*}

By incorporating the cancellations in the terms in $J_i$, we can write:
\begin{equation}\label{E.main.formula}
\begin{split}
[E_u(v;u)-E_u(v_0;u)]+[E_v(u;v)-E_v(-\infty;v)]=&\:\sum_{i=1}^7F_i,
\end{split} 
\end{equation}
with
\begin{align*}
F_{1}=&\:-M^2 \int_{v_0}^v \int_{-\infty}^u r^{-1}(u'^2\partial_ur+v'^2\partial_vr)(\overline{D_u\phi}D_v\phi+\overline{D_v\phi}D_u\phi)\,du'dv',\\
F_{2}=&\:-2M^2\int_{v_0}^v \int_{-\infty}^u r^{-1}\partial_vr\cdot u'^2|D_u\phi|^2+r^{-1}\partial_ur\cdot v'^2|D_v\phi|^2\,du'dv',\\
F_{3}=&\: \frac{1}{4} M^2 \mathfrak{m}^2\int_{v_0}^v \int_{-\infty}^u[\partial_v(v'^2\Omega_0^2)+\partial_u(u'^2\Omega_0^2)]|\phi|^2\,du'dv',\\
F_{4}=&\: \frac{1}{16\pi}M^{-2}\int_{v_0}^v \int_{-\infty}^u [\partial_v(v'^2\Omega_0^2)+\partial_u(u'^2\Omega_0^2)](Q^2-M^2)\,du'dv' ,\\
F_{5}=&\:\frac{1}{4}\int_{v_0}^v \int_{-\infty}^u [\partial_v(v'^2\cdot (\Omega^2-\Omega_0^2)+\partial_u(u'^2\cdot (\Omega^2-\Omega_0^2))](M^2 \mathfrak{m}^2|\phi|^2+\frac{1}{4\pi}M^{-2}(Q^2-M^2))\,du'dv',\\
F_{6}=&\: \frac{1}{4}i\mathfrak e M^{-2}\int_{v_0}^v \int_{-\infty}^u r^{-2}(r^4-M^4)u'^2Q\Omega^2(\phi\overline{D_u\phi}-\overline{\phi}D_u\phi)\,du'dv',\\
F_{7}=&\: \frac{1}{4}i\mathfrak e M^{-2}\int_{v_0}^v \int_{-\infty}^u r^{-2}(M^4-r^4)v'^2Q\Omega^2(\phi\overline{D_v\phi}-\overline{\phi}D_v\phi)\,du'dv'.
\end{align*}
We estimate using Cauchy--Schwarz inequality, Young's inequality, Proposition~\ref{prop:rest} and \eqref{eq:bootstrapphi}
\begin{equation*}
\begin{split}
|F_1|\lesssim &\:\int_{v_0}^v \int_{-\infty}^u  (u'^2|\partial_ur|+v'^2|\partial_vr|)|D_u\phi||D_v\phi|\,du'dv'\\
\lesssim&\: \int_{v_0}^v \int_{-\infty}^u \mathcal{A}_{\phi}(\mathcal{D}_{\rm o}+\mathcal{D}_{\rm i})|D_u\phi||D_v\phi|\,du'dv'\\
\lesssim&\: \int_{v_0}^v \int_{-\infty}^u \mathcal{A}_{\phi}(\mathcal{D}_{\rm o}+\mathcal{D}_{\rm i}) v'^{-2}u'^{\f 32}|D_u\phi|^2\,du'dv'+\int_{v_0}^v \int_{-\infty}^u \mathcal{A}_{\phi}(\mathcal{D}_{\rm o}+\mathcal{D}_{\rm i})u'^{-\f 32} v'^{2}|D_v\phi|^2\,du'dv'\\
\lesssim&\:\mathcal{A}_{\phi}(\mathcal{D}_{\rm o}+\mathcal{D}_{\rm i})v_0^{-1}|u_0|^{-\f 12}\cdot \sup_{v_0\leq v\leq v_{\infty}}  \int_{-\infty}^u u'^{2}|D_u\phi|^2\,du'\\
&+ \mathcal{A}_{\phi}(\mathcal{D}_{\rm o}+\mathcal{D}_{\rm i})|u_0|^{-\f 12}\cdot \sup_{-\infty<u<u_0}  \int_{v_0}^{v_{\infty}} v'^{2}|D_v\phi|^2\,dv'\\
\lesssim&\: \mathcal A_\phi^2(\mathcal{D}_{\rm o}+\mathcal{D}_{\rm i})^2|u_0|^{-\f 12}.
\end{split}
\end{equation*}
We can similarly estimate
\begin{equation*}
\begin{split}
|F_2|\lesssim &\:\int_{v_0}^v \int_{-\infty}^u  |\partial_vr|\cdot u'^2|D_u\phi|^2\,du'dv'+\int_{v_0}^v \int_{-\infty}^u  |\partial_ur|\cdot v'^2|D_v\phi|^2\,du'dv'\\
\lesssim &\:\mathcal{A}_{\phi}(\mathcal{D}_{\rm o}+\mathcal{D}_{\rm i})\left(\int_{v_0}^v \int_{-\infty}^u  v'^{-\f32}\cdot u'^{\f32}|D_u\phi|^2\,du'dv'+\int_{v_0}^v \int_{-\infty}^u u'^{-2}\cdot v'^2|D_v\phi|^2\,du'dv'\right)\\
&+\mathcal{D}_{\rm o}\int_{v_0}^v \int_{-\infty}^u  v'^{-2}\cdot u'^{2}|D_u\phi|^2\,du'dv'\\
\lesssim &\:[\mathcal{D}_{\rm o} v_0^{-1} +\mathcal{A}_{\phi} (\mathcal{D}_{\rm o} +\mathcal{D}_{\rm i})(v_0^{-\frac 12}|u_0|^{-\frac 12} +|u_0|^{-1})] \mathcal{A}_{\phi}(\mathcal{D}_{\rm o}+\mathcal{D}_{\rm i}).
\end{split}
\end{equation*}
In order to estimate $|F_3|$, we use \eqref{eq:weightestOmega0} (with a constant depending on $\beta>0$) and \eqref{eq:Linftyphi} to get
\begin{equation*}
\begin{split}
|F_3|\lesssim &\:\int_{v_0}^v \int_{-\infty}^u |\partial_v(v^2\Omega_0^2)+\partial_u(u^2\Omega_0^2)|\cdot |\phi|^2\,du'dv'\\
\lesssim&\: \int_{v_0}^v \int_{-\infty}^u (v'+|u'|)^{-2+\beta} (\mathcal{D}_{\rm o}v'^{-1}+\mathcal{A}_{\phi}(\mathcal{D}_{\rm o}+ \mathcal D_{\rm i})|u'|^{-1})\,du'dv'\\
\lesssim&\: \mathcal{A}_{\phi}(\mathcal{D}_{\rm o}+ \mathcal D_{\rm i})(v_0+|u_0|)^{-1+2\beta}.
\end{split}
\end{equation*}
Similarly, using \eqref{eq:weightestOmega0} and \eqref{eq:LinftyQ},
\begin{equation*}
\begin{split}
|F_4|\lesssim &\:\int_{v_0}^v \int_{-\infty}^u |\partial_v(v^2\Omega_0^2)+\partial_u(u^2\Omega_0^2)|\cdot|Q-M||Q+M|\,du'dv'\\
\lesssim&\: \int_{v_0}^v \int_{-\infty}^u (v'+|u'|)^{-2+\beta} (\mathcal{D}_{\rm o}v'^{-1}+\mathcal{A}_{\phi}(\mathcal{D}_{\rm o}+ \mathcal D_{\rm i})|u'|^{-1})\,du'dv'\\
\lesssim&\:\mathcal{A}_{\phi}(\mathcal{D}_{\rm o}+ \mathcal D_{\rm i})(v_0+|u_0|)^{-1+2\beta}.
\end{split}
\end{equation*}
Before we estimate $|F_5|$, it is convenient to rewrite the following expression:
\begin{equation*}
\begin{split}
\partial_v(v^2(\Omega^2-\Omega^2_0))+\partial_u(u^2(\Omega^2-\Omega^2_0))=&\:\partial_v\left(v^2\Omega_0^2\left(\frac{\Omega^2}{\Omega_0^2}-1\right)\right)+\partial_u\left(u^2\Omega_0^2\left(\frac{\Omega^2}{\Omega_0^2}-1\right)\right)\\
=&\: \left(\partial_v(v^2\Omega_0^2)+\partial_u(u^2\Omega_0^2)\right)\left(\frac{\Omega^2}{\Omega_0^2}-1\right)\\
&+v^2\Omega_0^2\partial_v\left(\frac{\Omega^2}{\Omega_0^2}-1\right)+u^2\Omega_0^2\partial_u\left(\frac{\Omega^2}{\Omega_0^2}-1\right)\\
=&\: \left(\partial_v(v^2\Omega_0^2)+\partial_u(u^2\Omega_0^2)\right)\left(\frac{\Omega^2}{\Omega_0^2}-1\right)\\
&+2v^2\Omega^2\partial_v\left(\log\frac{\Omega}{\Omega_0}\right)+2u^2\Omega^2\partial_u\left(\log\frac{\Omega}{\Omega_0}\right).
\end{split}
\end{equation*}
Therefore,
\begin{equation*}
\begin{split}
|F_5|\lesssim &\:\int_{v_0}^v \int_{-\infty}^u |\partial_v(v'^2\Omega_0^2)+\partial_u(u'^2\Omega_0^2)|\left|\frac{\Omega^2}{\Omega_0^2}-1\right|\cdot(|\phi|^2+|Q-M||Q+M|)\,du'dv'\\
&+ \int_{v_0}^v \int_{-\infty}^u v'^2\Omega^2\left|\partial_v\left(\log\frac{\Omega}{\Omega_0}\right)\right|\cdot(|\phi|^2+|Q-M||Q+M|)\,du'dv'\\
&+ \int_{v_0}^v \int_{-\infty}^u u'^2\Omega^2\left|\partial_u\left(\log\frac{\Omega}{\Omega_0}\right)\right|\cdot(|\phi|^2+|Q-M||Q+M|)\,du'dv'=:F_{5,1}+F_{5,2}+F_{5,3}.
\end{split}
\end{equation*}
Using \eqref{eq:weightestOmega0}, \eqref{eq:Linftyphi}, \eqref{eq:LinftyQ}, \eqref{eq:LinftyOmegav3}, we can estimate $|F_{5,1}|$ in the same way as $|F_3|$ and $|F_4|$ to obtain
\begin{equation*}
\begin{split}
|F_{5,1}|\lesssim&\: \mathcal{A}_{\phi}(\mathcal{D}_{\rm o}+ \mathcal D_{\rm i})|u_0|^{-\frac{1}{2}} (v_0+|u_0|)^{-1+2\beta}.
\end{split}
\end{equation*}
For $|F_{5,2}|$, we use \eqref{eq:Linftyphi}, \eqref{eq:LinftyQ} and \eqref{eq:bootstrapOmega} to estimate
\begin{equation*}
\begin{split}
&\:|F_{5,2}|\\
\lesssim&\:\int_{v_0}^v \int_{-\infty}^u (\mathcal{D}_{\rm o}v'^{-1}+\mathcal{A}_{\phi}(\mathcal{D}_{\rm o}+\mathcal D_{\rm i})|u'|^{-1})v'^2\Omega^2\left|\partial_v\left(\log\frac{\Omega}{\Omega_0}\right)\right|\,du'dv'\\
\lesssim&\: \mathcal{D}_{\rm o}\left(\sup_{u'\in (-\infty,u_0]} \int_{v_0}^v v'^2\left(\partial_v\left(\log\frac{\Omega}{\Omega_0}\right)\right)^2(u',v')\,dv'\right)^{\f 12}\int_{-\infty}^u\left(\int_{v_0}^v  (v'+|u'|)^{-4} \,dv'\right)^{\f 12}\,du'\\
&+ \mathcal{A}_{\phi}(\mathcal{D}_{\rm o}+ \mathcal D_{\rm i})\left(\sup_{u'\in (-\infty,u_0]} \int_{v_0}^v v'^2\left(\partial_v\left(\log\frac{\Omega}{\Omega_0}\right)\right)^2(u',v')\,dv'\right)^{\f 12}\int_{-\infty}^u\left(\int_{v_0}^v  \f{v'^2}{|u'|^2(v'+|u'|)^{4}} \,dv'\right)^{\f 12}\,du'\\
\lesssim&\: \mathcal{D}_{\rm o}(v_0+|u_0|)^{-\f 12}+\mathcal{A}_{\phi}(\mathcal{D}_{\rm o}+ \mathcal D_{\rm i})|u_0|^{-\f 12},
\end{split}
\end{equation*}
where in the last line we have evaluated an integral as follows: (We only include this estimate for completeness. In what follows, we will bound similar integrals in analogous manner without spelling out the full details.)
\begin{equation*}
\begin{split}
&\: \int_{-\infty}^u\left(\int_{v_0}^v  \f{v'^2}{|u'|^2(v'+|u'|)^{4}} \,dv'\right)^{\f 12}\,du'\\
\ls &\: \int_{-\infty}^u\left(\int_{v_0}^{|u|}  \f{v'^2}{|u'|^2(v'+|u'|)^{4}} \,dv' + \int_{|u|}^v  \f{v'^2}{|u'|^2(v'+|u'|)^{4}} \,dv'\right)^{\f 12}\,du' \\
\ls &\: \int_{-\infty}^u\left(\int_{v_0}^{|u|}  \f{v'^2}{|u'|^6} \,dv' + \int_{|u|}^v  \f{1}{|u'|^2 v'^2} \,dv'\right)^{\f 12}\,du'
\ls  \int_{-\infty}^u |u'|^{-\f 32}\,du' \ls |u_0|^{-\f 12}.
\end{split}
\end{equation*}
For $|F_{5,3}|$, we similarly use \eqref{eq:Linftyphi}, \eqref{eq:LinftyQ} and \eqref{eq:bootstrapOmega} to estimate as follows:
\begin{equation*}
\begin{split}
|F_{5,3}|\lesssim&\:\int_{v_0}^v \int_{-\infty}^u (\mathcal{D}_{\rm o}v'^{-1}+\mathcal{A}_{\phi}(\mathcal{D}_{\rm o}+ \mathcal D_{\rm i})|u'|^{-1})u'^2\Omega^2\left|\partial_u\left(\log\frac{\Omega}{\Omega_0}\right)\right|\,du'dv'\\
\lesssim&\: \mathcal{D}_{\rm o}\left(\sup_{u'\in(-\infty,u_0]}\int_{-\infty}^u u'^2\left(\partial_u\left(\log\frac{\Omega}{\Omega_0}\right)\right)^2(u',v')\,du'\right)^{\f 12}\int_{v_0}^v\left( \int_{-\infty}^u \frac{u'^2}{v'^2(v'+|u'|)^{4}}\,du'\right)^{\f 12}\,dv'\\
&+ \mathcal{A}_{\phi}(\mathcal{D}_{\rm o}+\mathcal D_{\rm i})\left(\sup_{u'\in(-\infty,u_0]}\int_{-\infty}^u u'^2\left(\partial_u\left(\log\frac{\Omega}{\Omega_0}\right)\right)^2(u',v')\,du'\right)^{\f 12}\int_{v_0}^v\left( \int_{-\infty}^u (v'+|u'|)^{-4}\,du'\right)^{\f 12}\,dv'\\
\lesssim&\: \mathcal{D}_{\rm o} v_0^{-\frac{1}{2}}+\mathcal{A}_{\phi}(\mathcal{D}_{\rm o}+\mathcal D_{\rm i})(v_0+|u_0|)^{-\frac{1}{2}}.
\end{split}
\end{equation*}
Thus, combining the estimates for $F_{5,1}$, $F_{5,2}$ and $F_{5,3}$, we obtain
$$|F_5|\ls \mathcal{D}_{\rm o} v_0^{-\frac{1}{2}}+\mathcal{A}_{\phi}(\mathcal{D}_{\rm o}+\mathcal D_{\rm i})|u_0|^{-\frac{1}{2}}.$$
We are left with $|F_6|$ and  $|F_7|$, which are slightly easier because more decay is available. For $F_6$, we use Cauchy--Schwarz inequality, \eqref{eq:Linftyphi}, \eqref{eq:LinftyOmegav2}, \eqref{r.est}, and \eqref{eq:bootstrapphi} to obtain
\begin{equation*}
\begin{split}
|F_6|\lesssim &\:\int_{v_0}^v \int_{-\infty}^u |r-M|u'^2\Omega^2|\phi||D_u\phi|\,du'dv'\\
\lesssim&\:\int_{v_0}^v \int_{-\infty}^u (\mathcal{D}_{\rm o}v'^{-1}+\mathcal A_\phi(\mathcal{D}_{\rm o}+\mathcal D_{\rm i})|u'|^{-1})(\sqrt{\mathcal{D}_{\rm o}}v'^{-\frac{1}{2}}+\sqrt{\mathcal{A}_{\phi}(\mathcal{D}_{\rm o}+\mathcal D_{\rm i})}|u'|^{-\frac{1}{2}})u'^2\Omega^2|D_u\phi|\,du'dv'\\
\lesssim&\: \mathcal{D}_{\rm o}^{\f 32}\left(\sup_{v'\in [v_0,v)}\int_{-\infty}^u u'^2|D_u\phi|^2(u',v')\,du'\right)^{\f 12}\int_{v_0}^v \left( \int_{-\infty}^u v'^{-3} u'^2 (v'+|u'|)^{-4}\,du' \right)^{\f 12}\,dv' \\
&\:+ \mathcal A_\phi^{\f 32}(\mathcal{D}_{\rm o}+\mathcal D_{\rm i})^{\f 32}\left(\sup_{v'\in [v_0,v)} \int_{-\infty}^u u'^2|D_u\phi|^2(u',v')\,du'\right)^{\f 12} \int_{v_0}^v\left( \int_{-\infty}^u |u'|^{-1}(v'+|u'|)^{-4}\,du' \right)^{\f 12}\, dv'\\
\ls & \: \mathcal A_{\phi}^{\f 12}\mathcal{D}_{\rm o}^{\f 32}(\mathcal{D}_{\rm o}+\mathcal D_{\rm i})^{\f 12}v_0^{-\f 12}|u_0|^{-\f 12} + \mathcal A_\phi^2(\mathcal{D}_{\rm o}+\mathcal D_{\rm i})^2|u_0|^{-\f 12}(v_0+|u_0|)^{-\f 12}.   
\end{split}
\end{equation*}
Similarly, we use Cauchy--Schwarz inequality, \eqref{eq:Linftyphi}, \eqref{eq:LinftyOmegav2}, \eqref{r.est}, and \eqref{eq:bootstrapphi} to obtain
\begin{equation*}
\begin{split}
|F_7|\lesssim &\:\int_{v_0}^v \int_{-\infty}^u |r-M|v'^2\Omega^2|\phi||D_v\phi|\,du'dv'\\
\lesssim&\:\int_{v_0}^v \int_{-\infty}^u (\mathcal{D}_{\rm o}v'^{-1}+\mathcal A_\phi(\mathcal{D}_{\rm o}+\mathcal D_{\rm i})|u'|^{-1})(\sqrt{\mathcal{D}_{\rm o}}v'^{-\frac{1}{2}}+\sqrt{\mathcal{A}_{\phi}(\mathcal{D}_{\rm o}+\mathcal D_{\rm i})}|u'|^{-\frac{1}{2}})v'^2\Omega^2|D_v\phi|\,du'dv'\\
\lesssim&\:\mathcal{D}_{\rm o}^{\f 32} \left(\sup_{u'\in  (-\infty,u_0]}\int_{v_0}^v v'^2|D_v\phi|^2(u',v')\,dv'\right)^{\f 12} \int_{-\infty}^u\left(\int_{v_0}^v v'^{-1} (v'+|u'|)^{-4}\,dv'\right)^{\f 12}\,du'\\
&+\mathcal A_{\phi}^{\f 32}(\mathcal{D}_{\rm o}+\mathcal D_{\rm i})^{\f 32} \left(\sup_{u'\in  (-\infty,u_0]}\int_{v_0}^v v'^2|D_v\phi|^2(u',v')\,dv'\right)^{\f 12}\int_{-\infty}^u \left(\int_{v_0}^v v'^2 |u'|^{-3} (v'+|u'|)^{-4}\,dv'\right)^{\f 12} \,du'\\
\lesssim&\: \mathcal A_{\phi}^{\f 12}\mathcal{D}_{\rm o}^{\f 32}(\mathcal{D}_{\rm o}+\mathcal D_{\rm i})^{\f 12} v_0^{-\f 12}(v_0+|u_0|)^{-\f 12}+ \mathcal A_{\phi}^2(\mathcal{D}_{\rm o}+\mathcal D_{\rm i})^2 |u_0|^{-\f 12}(v_0+|u_0|)^{-\f 12}.
\end{split}
\end{equation*}
Choosing $v_0$ and $|u_0|$ large in a manner allowed by \eqref{const.adm}, we obtain
$$|F_1|+\dots + |F_7|\ls \mathcal{D}_{\rm o}+\mathcal{D}_{\rm i}.$$
Finally, noting that the initial data contributions $E_u(v_0;u)$ and $E_v(-\infty;v)$ are by definition bounded by $\mathcal{D}_{\rm o}+\mathcal{D}_{\rm i}$, and returning to \eqref{E.main.formula}, we obtain
$$\sup_{v_0\leq v< v_{\infty}}E_u(v;u)+\sup_{-\infty<u<u_0}E_v(u;v)\ls \mathcal{D}_{\rm o}+\mathcal{D}_{\rm i},$$
which is to be proved.
\end{proof}

Combining Propositions~\ref{prop:coercenergy} and \ref{prop:eestimatephi}, we obtain the following estimate. In particular, this is an improvement over the bootstrap assumption \eqref{eq:bootstrapphi} for $\mathcal A_{\phi}$ sufficiently large depending on $M$, $\mfm$, $\mfe$ and $\eta$. 
\begin{corollary}\label{cor:phi}
Choosing $\mathcal A_\phi$ sufficiently large (depending on $M$, $\mfm$ and $\mfe$), the following estimate holds:
\begin{align*}
\sup_{v\in [v_0,v_{\infty})}\int_{-\infty}^{u_0} u^2 |D_u \phi |^2(u,v)\,du+\sup_{u\in (-\infty,u_0)}\int_{v_0}^{v_{\infty}}v^2|D_v\phi|^2(u,v)\,dv\leq &\: C(\mathcal{D}_{\rm o}+\mathcal D_{\rm i}) \leq \f{A_{\phi}}{2}(\mathcal{D}_{\rm o}+\mathcal D_{\rm i}).
\end{align*}
\end{corollary}

\textbf{At this point we fix $\mathcal A_{\phi}$} so that Corollary~\ref{cor:phi} holds.

\subsection{Energy estimates for $\log \frac{\Omega}{\Omega_0}$}\label{sec:ee.Om}
Finally, we carry out the energy estimates for $\log \frac{\Omega}{\Omega_0}$. As we noted in the introduction, the essential point is to establish that $\log \frac{\Omega}{\Omega_0}$ obeys an equation of the form \eqref{Om.Klein.Gordon} up to lower order terms. More precisely, starting with \eqref{eq:waveqOmega}, the estimates that we have obtained so far show that the $D_u\phi \overline{D_v\phi}$ and $\rd_u r \rd_v r$ terms have better decay properties, and that $r$ and $Q$ both decay to $M$. Therefore, \eqref{eq:waveqOmega} can indeed be thought of as \eqref{Om.Klein.Gordon}.

We split the proof of the energy estimates into two parts. First, in Lemma~\ref{lem:Om.eqn}, we consider an energy inspired by the form \eqref{Om.Klein.Gordon} and write down the error terms that arise when controlling this energy. Then, in Proposition~\ref{prop:eestimateOmega}, we will then bound all the error terms arising in Lemma~\ref{lem:Om.eqn} to obtain the desired estimate for $\log \frac{\Omega}{\Omega_0}$.

\begin{lemma}\label{lem:Om.eqn}
The following identity holds for any $u\in (-\infty,u_0)$ and $v\in [v_0,v_\infty)$:
\begin{equation*}
\begin{split}
\int_{-\infty}^u &u'^2\left(\partial_v \log \left(\frac{\Omega}{\Omega_0}\right)\right)^2+\frac{1}{8}M^{-4}u'^2\Omega^{-2}\left(\Omega^2-\Omega_0^2\right)^2(u',v)\,du'\\
&+\int_{v_0}^vv'^2\left(\partial_v \log \left(\frac{\Omega}{\Omega_0}\right)\right)^2+\frac{1}{8}M^{-4}v'^2\Omega^{-2}\left(\Omega^2-\Omega_0^2\right)^2(u,v')\,dv'\\
=&\: \sum_{i=1}^{6}O_{i},
\end{split}
\end{equation*}
where
\begin{align*}
O_1=&\:-4\pi \int_{v_0}^v \int_{-\infty}^u (D_u\phi\overline{D_v\phi}+\overline{D_u\phi}D_v\phi)\left(v'^2\partial_v \log \left(\frac{\Omega}{\Omega_0}\right)+u'^2\partial_u \log \left(\frac{\Omega}{\Omega_0}\right)\right)\,du'dv',\\
O_2=&\:\frac{1}{8}M^{-4}\int_{v_0}^v \int_{-\infty}^u \left[\partial_v\left(v'^2\Omega_0^2\right)+\partial_u\left(u'^2\Omega_0^2\right)\right]\frac{\Omega_0^2}{\Omega^2}\left(\frac{\Omega^2}{\Omega_0^2}-1\right)^2\,du'dv',\\
O_3=&\:-\frac{1}{4}M^{-4}\int_{v_0}^v \int_{-\infty}^u \Omega^{-2}\left(\Omega^2-\Omega_0^2\right)^2\left(v'^2\partial_v \log \left(\frac{\Omega}{\Omega_0}\right)+u'^2\partial_u \log \left(\frac{\Omega}{\Omega_0}\right)\right)\,du'dv',\\
O_4=&\:\int_{v_0}^v \int_{-\infty}^u \Omega_0^2\left[Q^2(r_0^{-4}-r^{-4})+r_0^{-4}(M^{2}-Q^2)-\frac{1}{2}(r_0^{-2}-r^{-2})\right]\\
&\cdot\left(v'^2\partial_v \log \left(\frac{\Omega}{\Omega_0}\right)+u'^2\partial_u \log \left(\frac{\Omega}{\Omega_0}\right)\right)\,du'dv',\\
O_5=&\:\int_{v_0}^v \int_{-\infty}^u (\Omega^2-\Omega_0^2)\left[r^{-4}(M^2-Q^2)+\frac{1}{2}(r^{-2}-M^{-2})-M^2(r^{-4}-M^{-4})\right]\\
&\cdot\left(v'^2\partial_v \log \left(\frac{\Omega}{\Omega_0}\right)+u'^2\partial_u \log \left(\frac{\Omega}{\Omega_0}\right)\right)\,du'dv',\\
O_6=&\:2\int_{v_0}^v \int_{-\infty}^u\left( r^{-2}\partial_ur\partial_vr-r_0^{-2}\partial_ur_0\partial_vr_0\right)\left(v'^2\partial_v \log \left(\frac{\Omega}{\Omega_0}\right)+u'^2\partial_u \log \left(\frac{\Omega}{\Omega_0}\right)\right)\,du'dv',
\end{align*}
where, as in Proposition~\ref{prop:eestimatephi}, we have suppressed the argument $(u',v')$ in the integrand in the $O_i$ terms.
\end{lemma}
\begin{proof}
By \eqref{eq:waveqOmega} we have that

\begin{equation*}
\begin{split}
\partial_u\partial_v \log \left(\frac{\Omega}{\Omega_0}\right)=&\:-2\pi (D_u\phi\overline{D_v\phi}+\overline{D_u\phi}D_v\phi)+r^{-2}\partial_ur\partial_vr-r_0^{-2}\partial_ur_0\partial_vr_0\\
&-\frac{1}{2}\Omega^2r^{-4}Q^2+\frac{1}{2}\Omega_0^2r_0^{-4}M^2+\frac{1}{4}\Omega^2r^{-2}-\frac{1}{4}\Omega_0^2r_0^{-2}\\
=&\:-2\pi (D_u\phi\overline{D_v\phi}+\overline{D_u\phi}D_v\phi)+r^{-2}\partial_u(r-r_0)\partial_vr+r^{-2}\partial_ur_0\partial_v(r-r_0)\\
&+\partial_ur_0\partial_vr_0(r^{-2}-r_0^{-2})\\
&-\frac{1}{2}(\Omega^2-\Omega^2_0)r^{-4}M^2+\frac{1}{2}(\Omega^2-\Omega_0^2)r^{-4}(M^2-Q^2)+\frac{1}{4}(\Omega^2-\Omega^2_0)M^{-2}\\
&+\frac{1}{4}(\Omega^2-\Omega^2_0)(r^{-2}-M^{-2})\\
&+\frac{1}{2}\Omega_0^2Q^2(r_0^{-4}-r^{-4})+\frac{1}{2}\Omega_0^2r_0^{-4}(M^{2}-Q^2)-\frac{1}{4}\Omega_0^2(r_0^{-2}-r^{-2})\\
=&\:-2\pi (D_u\phi\overline{D_v\phi}+\overline{D_u\phi}D_v\phi)+r^{-2}\partial_u(r-r_0)\partial_vr+r^{-2}\partial_ur_0\partial_v(r-r_0)\\
&+\partial_ur_0\partial_vr_0(r^{-2}-r_0^{-2})\\
&-\frac{1}{2}(\Omega^2-\Omega_0^2)M^{-2}-\frac{1}{2}(\Omega^2-\Omega_0^2)M^2(r^{-4}-M^{-4})+\frac{1}{2}(\Omega^2-\Omega_0^2)r^{-4}(M^2-Q^2)\\
&+\frac{1}{4}(\Omega^2-\Omega^2_0)M^{-2}+\frac{1}{4}(\Omega^2-\Omega^2_0)(r^{-2}-M^{-2})\\
&+\frac{1}{2}\Omega_0^2Q^2(r_0^{-4}-r^{-4})+\frac{1}{2}\Omega_0^2r_0^{-4}(M^{2}-Q^2)-\frac{1}{4}\Omega_0^2(r_0^{-2}-r^{-2}).
\end{split}
\end{equation*}
Using the above equation, we obtain
\begin{equation*}
\begin{split}
&\partial_u\left(v^2\left(\partial_v \log \left(\frac{\Omega}{\Omega_0}\right)\right)^2\right)\\
=&\:2v^2\partial_u\partial_v \log \left(\frac{\Omega}{\Omega_0}\right)\cdot \partial_v \log \left(\frac{\Omega}{\Omega_0}\right)\\
=&\:-4\pi v^2(D_u\phi\overline{D_v\phi}+\overline{D_u\phi}D_v\phi)\partial_v \log \left(\frac{\Omega}{\Omega_0}\right)+2r^{-2}\partial_u(r-r_0)\partial_vrv^2\partial_v \log \left(\frac{\Omega}{\Omega_0}\right)\\
&+2r^{-2}\partial_ur_0\partial_v(r-r_0)v^2\partial_v \log \left(\frac{\Omega}{\Omega_0}\right)+2\partial_ur_0\partial_vr_0(r^{-2}-r_0^{-2})v^2\partial_v \log \left(\frac{\Omega}{\Omega_0}\right)\\
&-\frac{1}{2}M^{-2}v^2\left(\Omega^2-\Omega_0^2\right)\partial_v \log \left(\frac{\Omega}{\Omega_0}\right)-(\Omega^2-\Omega_0^2)M^2(r^{-4}-M^{-4})v^2\partial_v \log \left(\frac{\Omega}{\Omega_0}\right)\\
&+(\Omega^2-\Omega_0^2)r^{-4}(M^2-Q^2)v^2\partial_v \log \left(\frac{\Omega}{\Omega_0}\right)+\frac{1}{2}(\Omega^2-\Omega^2_0)(r^{-2}-M^{-2})v^2\partial_v \log \left(\frac{\Omega}{\Omega_0}\right)\\
&+\Omega_0^2Q^2(r_0^{-4}-r^{-4})v^2\partial_v \log \left(\frac{\Omega}{\Omega_0}\right)+\Omega_0^2r_0^{-4}(M^{2}-Q^2)v^2\partial_v \log \left(\frac{\Omega}{\Omega_0}\right)\\
&-\frac{1}{2}\Omega_0^2(r_0^{-2}-r^{-2})v^2\partial_v \log \left(\frac{\Omega}{\Omega_0}\right).
\end{split}
\end{equation*}
Note that we can write
\begin{equation*}
\begin{split}
\left(\Omega^2-\Omega_0^2\right)\partial_v \log \left(\frac{\Omega}{\Omega_0}\right)=&\:\frac{1}{2}(\Omega^2-\Omega_0^2)\partial_v\left(\frac{\Omega^2}{\Omega_0^2}-1\right)\frac{\Omega_0^2}{\Omega^2}
=\: \frac{1}{4}\frac{\Omega_0^4}{\Omega^2}\partial_v\left(\left(\frac{\Omega^2}{\Omega_0^2}-1\right)^2\right).
\end{split}
\end{equation*}
Hence,
\begin{equation*}
\begin{split}
-\frac{1}{2}M^{-4}v^2\left(\Omega^2-\Omega_0^2\right)\partial_v \log \left(\frac{\Omega}{\Omega_0}\right)=&\:-\frac{1}{8}v^2M^{-4}\frac{\Omega_0^4}{\Omega^2}\partial_v\left(\left(\frac{\Omega^2}{\Omega_0^2}-1\right)^2\right)\\
=&\:-\partial_v\left(\frac{1}{8}M^{-4}v^2\Omega^{-2}\left(\Omega^2-\Omega_0^2\right)^2\right)+\frac{1}{8}M^{-4}\partial_v\left(v^2\Omega_0^2\frac{\Omega_0^2}{\Omega^2}\right)\left(\frac{\Omega^2}{\Omega_0^2}-1\right)^2\\
=&\:-\partial_v\left(\frac{1}{8}M^{-4}v^2\Omega^{-2}\left(\Omega^2-\Omega_0^2\right)^2\right)\\
&+\frac{1}{8}M^{-4}\partial_v\left(v^2\Omega_0^2\right)\frac{\Omega_0^2}{\Omega^2}\left(\frac{\Omega^2}{\Omega_0^2}-1\right)^2\\
&-\frac{1}{4}M^{-4}v^2\Omega^{-2}\partial_v \log \left(\frac{\Omega}{\Omega_0}\right)\left(\Omega^2-\Omega_0^2\right)^2.
\end{split}
\end{equation*}
We similarly consider $\partial_v\left(u^2\left(\partial_u \log \left(\frac{\Omega}{\Omega_0}\right)\right)^2\right)$ and use Leibniz rule (with $u$ replacing the role of $v$). Noting also that by the gauge condition \eqref{eq:gauge.cond}, $\log\f{\Om}{\Om_0}=0$ on the initial hypersurfaces, this yields the statement of the Lemma.
\end{proof}

\begin{proposition}
\label{prop:eestimateOmega}
The following estimate holds
\begin{equation*} 
\sup_{v\in [v_0,v_\infty]}\int_{-\infty}^{u_0} u'^2 \left(\partial_u\left(\log \frac{\Omega}{\Omega_0} \right)\right)^2(u',v)\,du'+\sup_{u\in(-\infty,u_0]}\int_{v_0}^{v_{\infty}} v'^2 \left(\partial_v\left(\log \frac{\Omega}{\Omega_0} \right)\right)^2(u,v')\,dv'\leq \f {M}2.
\end{equation*}
\end{proposition}

\begin{proof}

In order to obtain the stated estimates, we need to bound each of the terms in Lemma~\ref{lem:Om.eqn}. The basic idea is to use the bootstrap assumption \eqref{eq:bootstrapOmega} to control $\rd_v \log \left(\frac{\Omega}{\Omega_0}\right)$ and $\rd_u \log \left(\frac{\Omega}{\Omega_0}\right)$ and to use the estimates that we have previously obtained to deduce the decay and smallness of these terms.

We begin with the estimates for $O_1$. This turns out to be the most difficult term since we do not have any kind of pointwise estimates for $|D_v\phi|$, $|D_u\phi|$, $|\rd_v \log (\frac{\Omega}{\Omega_0})|$ and $|\rd_u \log (\frac{\Omega}{\Omega_0})|$. We bound it as follows using \eqref{eq:bootstrapOmega}:
\begin{equation}\label{eq:O1.est}
\begin{split}
|O_1|\lesssim &\: \int_{v_0}^v \int_{-\infty}^u |D_u\phi|\cdot|{D_v\phi}|\cdot \left(v'^2\left|\partial_v \log \left(\frac{\Omega}{\Omega_0}\right)\right|+u'^2\left|\partial_u \log \left(\frac{\Omega}{\Omega_0}\right)\right|\right)\,du'dv'\\
\lesssim &\: \left(\int_{v_0}^v v'^{2}\sup_{u' \in (-\infty,u]}|D_v\phi|^2(u',v')\,dv'\right)^{\f 12} \cdot \left(\int_{-\infty}^u u'^2 \sup_{v'\in [v_0, v]}|D_u\phi|^2(u',v')\,du'\right)^{\f 12}\\
&\times\left( |u_0|^{-\f 12}\left(\sup_{u' \in (-\infty,u]}\int_{v_0}^v v'^2\left|\partial_v \log \left(\frac{\Omega}{\Omega_0}\right)\right|^2(u',v')\,dv'\right)^{\f 12} \right.\\
&\:\left.\qquad+ v_0^{-\f 12}\left(\sup_{v' \in [v_0,v]}\int_{-\infty}^u u'^2\left|\partial_u \log \left(\frac{\Omega}{\Omega_0}\right)\right|^2(u',v')\,du'\right)^{\f 12}\right)\\
\lesssim &\: (|u_0|^{-\f 12}+v_0^{-\f 12})\left(\int_{v_0}^v \sup_{u'\in (-\infty, u]}\left[\int_{-\infty}^{u'}v'^{2}\partial_u(|D_v\phi|^2)(u'',v')\,du''\right]\,dv'\right)^{\f 12}\\
&\:\times\left(\int_{-\infty}^u\sup_{v'\in [v_0, v]}\left[\int_{v_0}^{v'}|u'|^{2}\partial_v(|D_u\phi|^2)(u',v'')\,dv''\right]\,du'\right)^{\f 12}.
\end{split}
\end{equation}
We first consider the last factor in \eqref{eq:O1.est}. From the computation for $J_1$ in the proof of Proposition~\ref{prop:eestimatephi} (which used \eqref{eq:phi1} and \eqref{eq:phi2}), it follows that
\begin{equation}\label{eq:O1.est.uu}
\begin{split}
\int_{-\infty}^u&\sup_{v'\in [v_0, v]}\left[\int_{v_0}^{v'}|u'|^{\f 32}\partial_v(|D_u\phi|^2)\,dv''\right]\,du'\\
\lesssim &\: \int_{v_0}^v \int_{-\infty}^{u}u'^{\f 32} \Omega^2|\phi||{D_u\phi}|+u'^{\f 32}|\partial_ur||D_u\phi|D_v\phi|+u'^{\f 32}|\partial_vr||D_u\phi|^2\,du'dv'\\
=: &\: O_{1,1} + O_{1,2} + O_{1,3}.
%\lesssim&\: (|u_0|^{-\frac{1}{2}}+v_0^{-\frac{1}{2}})(\mathcal{D}_{\rm o}+\mathcal{D}_{\rm i}+\mathcal{A}_{\Omega}),
\end{split}
\end{equation}
%where the final inequality follows from the estimates in the proof of Proposition \ref{prop:eestimatephi}.
The terms $O_{1,2}$ and $O_{1,3}$ have already been controlled in the proof of Proposition \ref{prop:eestimatephi}. More precisely, estimating as the terms $F_1$ and $F_2$ in the proof of Proposition \ref{prop:eestimatephi}, and noting that $O_{1,2}$ and $O_{1,3}$ have an additional $u'^{-\f 12}$ weight compared to $F_1$ and $F_2$, we have
$$O_{1,2}+O_{1,3}\ls \mathcal A_\phi^2(\mathcal{D}_{\rm o}+\mathcal{D}_{\rm i})^2|u_0|^{-1}+ \mathcal A_{\phi}\mathcal{D}_{\rm o} (\mathcal{D}_{\rm o}+\mathcal{D}_{\rm i})v_0^{-1}|u_0|^{-\f 12}\ls  \mathcal A_\phi^2(\mathcal{D}_{\rm o}+\mathcal{D}_{\rm i})^2|u_0|^{-\f 12}.$$
It thus remains to bound $O_{1,1}$, which has no analogue in Proposition~\ref{prop:eestimatephi}. To control this term, we use \eqref{eq:Linftyphi}, \eqref{eq:estrangeOmega}, the Cauchy--Schwarz inequality and the bootstrap assumption \eqref{eq:bootstrapphi} to obtain
\begin{equation*}
\begin{split}
%& \int_{v_0}^v \int_{-\infty}^{u}u'^{\f 32} \Omega^2|\phi||{D_u\phi}| \, du'dv'\\
O_{1,1} \ls &\:\sqrt{\mathcal A_\phi(\mathcal{D}_{\rm o}+\mathcal D_{\rm i})}\int_{v_0}^v \left(\int_{-\infty}^{u}u'^2 |{D_u\phi}|^2 \, du'\right)^{\f 12} \left(\int_{-\infty}^{u} (v'+|u'|)^{-4} \, du'\right)^{\f 12} \,dv'\\
 & + \sqrt{\mathcal{D}_{\rm o}} \int_{v_0}^v \left(\int_{-\infty}^{u}u'^2 |{D_u\phi}|^2 \, du'\right)^{\f 12} \left(\int_{-\infty}^{u}\f{|u'|}{v'^{1+\alp}} (v'+|u'|)^{-4} \, du'\right)^{\f 12} \,dv'\\
\ls & \:\sqrt{\mathcal A_\phi(\mathcal{D}_{\rm o}+\mathcal D_{\rm i})}\left(\sup_{v'\in [v_0, v]} \int_{-\infty}^{u}u'^2 |{D_u\phi}|^2(u',v') \, du'\right)^{\f 12} \left( \int_{v_0}^v (v'+|u|)^{-\f 32} dv'+ \int_{v_0}^v v'^{-\f 12-\f{\alp}{2}}(v'+|u|)^{-1} dv'\right)\\
\ls & \:\mathcal A_\phi(\mathcal{D}_{\rm o}+\mathcal D_{\rm i}) (v_0+|u_0|)^{-\f 12}.
\end{split}
\end{equation*}
Combining all these and plugging back into \eqref{eq:O1.est.uu}, we obtain
\begin{equation}\label{eq:O1.est.uu.1}
\begin{split}
\int_{-\infty}^u\sup_{v'\in [v_0, v]}\left[\int_{v_0}^{v'}|u'|^{\f 32}\partial_v(|D_u\phi|^2)\,dv''\right]\,du'
\ls &\mathcal A_\phi^2(\mathcal{D}_{\rm o}+\mathcal{D}_{\rm i})^2|u_0|^{-\f 12} + \mathcal A_\phi(\mathcal{D}_{\rm o}+\mathcal{D}_{\rm i})(v_0+|u_0|)^{-\f 12}\\
\ls & \: \mathcal A_\phi^2(\mathcal{D}_{\rm o}+\mathcal{D}_{\rm i})^2|u_0|^{-\f 12}.
\end{split}
\end{equation}

For the other factor in \eqref{eq:O1.est}, we estimate similarly by
\begin{equation*}
\begin{split}
\int_{v_0}^v &\sup_{-\infty\leq u'\leq u}\left[\int_{-\infty}^{u'}v'^{\f 32}\partial_u(|D_v\phi|^2)\,du''\right]\,dv'\\
\lesssim&\: \int_{v_0}^v \int_{-\infty}^{u}v'^2 \Omega^2|\phi||{D_v\phi}|+v'^2|\partial_{v}r||D_{u}\phi|D_v\phi|+v'^2|\partial_ur||D_v\phi|^2\,du'dv'=:O_{1,4}+O_{1,5}+O_{1,6}.
\end{split}
\end{equation*}
The terms $O_{1,5}$ and $O_{1,6}$, just as $O_{1,2}$ and $O_{1,3}$, can be bounded above by $\mathcal A_\phi^2(\mathcal{D}_{\rm o}+\mathcal{D}_{\rm i})^2|u_0|^{-\f 12}$ as the terms $F_1$ and $F_2$ in the proof of Proposition~\ref{prop:eestimatephi}. For the term $O_{1,4}$, we have, using \eqref{eq:Linftyphi}, \eqref{eq:estrangeOmega}, the Cauchy--Schwarz inequality and the bootstrap assumption \eqref{eq:bootstrapphi},

\begin{equation*}
\begin{split}
& \int_{v_0}^v \int_{-\infty}^{u}v'^{\f 32} \Omega^2|\phi||{D_v\phi}| \, du'dv'\\
\ls & \sqrt{\mathcal A_\phi(\mathcal{D}_{\rm o}+\mathcal D_{\rm i})} \int_{-\infty}^{u} (\int_{v_0}^v v'^2 |{D_v\phi}|^2 \, dv')^{\f 12}(\int_{v_0}^v |u'|^{-1} v' (v'+|u'|)^{-4} dv')^{\f 12} \,du'\\
& + \sqrt{\mathcal{D}_{\rm o}} \int_{-\infty}^{u} (\int_{v_0}^v v'^2 |{D_v\phi}|^2\, dv')^{\f 12}(\int_{v_0}^v v'^{-2\alp} (v'+|u'|)^{-4} dv')^{\f 12} \,du'\\
\ls & \sqrt{\mathcal A_\phi(\mathcal{D}_{\rm o}+\mathcal D_{\rm i})}(\sup_{u'\in (-\infty, u]}(\int_{v_0}^v v'^2 |{D_v\phi}|^2(u',v') \, dv')^{\f 12})(\int_{-\infty}^u |u'|^{-\f 12} (v_0+|u'|)^{-1}\ du' + \int_{-\infty}^u (v_0+|u'|)^{-\f 32}\, du')\\
\ls & \mathcal A_\phi(\mathcal{D}_{\rm o}+\mathcal D_{\rm i})|u_0|^{-\f 12}.
\end{split}
\end{equation*}
Combining, we obtain
\begin{equation}\label{eq:O1.est.vv}
\begin{split}
\int_{v_0}^v &\sup_{-\infty\leq u'\leq u}\left[\int_{-\infty}^{u'}v'^{\f 32}\partial_u(|D_v\phi|^2)\,du''\right]\,dv'
%\lesssim&\: \int_{v_0}^v \int_{-\infty}^{u}v'^2 \Omega^2|\phi||{D_v\phi}|+v'^2|\partial_{v}r||D_{u}\phi|D_v\phi|+v'^2|\partial_ur||D_v\phi|^2\,du'dv'\\
\ls  \mathcal A_\phi^2(\mathcal{D}_{\rm o}+\mathcal D_{\rm i})^2|u_0|^{-\f 12}.
\end{split}
\end{equation}

Combining \eqref{eq:O1.est}, \eqref{eq:O1.est.uu.1} and \eqref{eq:O1.est.vv}, we can therefore conclude that
\begin{equation*}
|O_1|\lesssim \mathcal A_\phi^2(\mathcal{D}_{\rm o}+\mathcal D_{\rm i})^2|u_0|^{-\f 12}(|u_0|^{-\f 12} + v_0^{-\f 12}).
\end{equation*}

We estimate $|O_2|$ by applying \eqref{eq:weightestOmega0} and \eqref{eq:LinftyOmegav3}. (Here, as before, the implicit constant may depend on $\beta$ for $\beta>0$.)
\begin{equation*}
\begin{split}
|O_2|\lesssim &\: \int_{v_0}^v \int_{-\infty}^u \left| \partial_v\left(v'^2\Omega_0^2\right)+\partial_u\left(u'^2\Omega_0^2\right)\right|\frac{\Omega_0^2}{\Omega^2}\left(\frac{\Omega^2}{\Omega_0^2}-1\right)^2\,du'dv'\\
\lesssim &\: \int_{v_0}^v \int_{-\infty}^u (v'+|u'|)^{-2+\beta}\left(\frac{\Omega^2}{\Omega_0^2}-1\right)^2\,du'dv'\\
\lesssim &\: \int_{v_0}^v \int_{-\infty}^u  (v'+|u'|)^{-2+\beta} |u'|^{-1}\,du'dv'\\
\lesssim &\: v_0^{-1+\beta}+|u_0|^{-1+\beta}.
\end{split}
\end{equation*}

It turns out that the remaining terms have a similar structure and is convenient to bound them in the same way. The following are the three basic estimates. First, using Cauchy--Schwarz inequality and \eqref{eq:bootstrapOmega}, we have
\begin{equation}\label{eq:Om.block.1}
\begin{split}
&\:\int_{v_0}^v \int_{-\infty}^u (v'+|u'|)^{-2}|u'|^{-1}\left(v'^2\left|\partial_v \log \left(\frac{\Omega}{\Omega_0}\right)\right|+u'^2\left|\partial_u \log \left(\frac{\Omega}{\Omega_0}\right)\right|\right)\,du'dv' \\
\ls & \: \left(\int_{v_0}^v \int_{-\infty}^u (v'+|u'|)^{-4}|u'|^{-\f 12}v'^2\,du'dv' \right)^{\f 12} \left(\int_{v_0}^v \int_{-\infty}^u |u'|^{-\f 32} v'^2\left|\partial_v \log \left(\frac{\Omega}{\Omega_0}\right)\right|^2 \,du'dv' \right)^{\f 12}\\
& \: +\left(\int_{v_0}^v \int_{-\infty}^u (v'+|u'|)^{-4} v'^{\f 32}\,du'dv' \right)^{\f 12} \left(\int_{v_0}^v \int_{-\infty}^u v'^{-\f 32} u'^2\left|\partial_u \log \left(\frac{\Omega}{\Omega_0}\right)\right|^2\,du'dv' \right)^{\f 12}\\
\ls & \: \left((|u_0|^{-\f 14} +v_0^{-\f 14})|u_0|^{-\f 14} + (v_0+|u_0|)^{-\f 14} v_0^{-\f 14}\right) \ls |u_0|^{-\f 14}.
\end{split}
\end{equation}
Again, using Cauchy--Schwarz inequality and \eqref{eq:bootstrapOmega}, we have
\begin{equation}\label{eq:Om.block.2}
\begin{split}
&\:\int_{v_0}^v \int_{-\infty}^u (v'+|u'|)^{-2}v'^{-1}\left(v'^2\left|\partial_v \log \left(\frac{\Omega}{\Omega_0}\right)\right|+u'^2\left|\partial_u \log \left(\frac{\Omega}{\Omega_0}\right)\right|\right)\,du'dv' \\
\ls & \: \left(\int_{v_0}^v \int_{-\infty}^u (v'+|u'|)^{-4}|u'|^{\f 32}\,du'dv' \right)^{\f 12} \left(\int_{v_0}^v \int_{-\infty}^u |u'|^{-\f 32} v'^2\left|\partial_v \log \left(\frac{\Omega}{\Omega_0}\right)\right|^2 \,du'dv' \right)^{\f 12}\\
& \: +\left(\int_{v_0}^v \int_{-\infty}^u (v'+|u'|)^{-4} v'^{-\f 12}|u'|^2\,du'dv' \right)^{\f 12} \left(\int_{v_0}^v \int_{-\infty}^u v'^{-\f 32} u'^2\left|\partial_u \log \left(\frac{\Omega}{\Omega_0}\right)\right|^2\,du'dv' \right)^{\f 12}\\
\ls & \: ((v_0+|u_0|)^{-\f 14}|u_0|^{-\f 14} + (v_0^{-\f 14} + |u_0|^{-\f 14}) v_0^{-\f 14}) \ls v_0^{-\f 14}.
\end{split}
\end{equation}
Thirdly, we have a another slight variant of the above estimates, for which we again use Cauchy--Schwarz inequality and \eqref{eq:bootstrapOmega}:
\begin{equation}\label{eq:Om.block.3}
\begin{split}
&\:\int_{v_0}^v \int_{-\infty}^u v'^{-2}|u'|^{-2} \left(v'^2\left|\partial_v \log \left(\frac{\Omega}{\Omega_0}\right)\right|+u'^2\left|\partial_u \log \left(\frac{\Omega}{\Omega_0}\right)\right|\right)\,du'dv' \\
\ls & \: \left(\int_{v_0}^v \int_{-\infty}^u v'^{-2}|u'|^{-\f 52}\,du'dv' \right)^{\f 12} \left(\int_{v_0}^v \int_{-\infty}^u |u'|^{-\f 32} v'^2\left|\partial_v \log \left(\frac{\Omega}{\Omega_0}\right)\right|^2 \,du'dv' \right)^{\f 12}\\
& \: +\left(\int_{v_0}^v \int_{-\infty}^u v'^{-\f 52}|u'|^{-2}\,du'dv' \right)^{\f 12} \left(\int_{v_0}^v \int_{-\infty}^u v'^{-\f 32} u'^2\left|\partial_u \log \left(\frac{\Omega}{\Omega_0}\right)\right|^2\,du'dv' \right)^{\f 12}\\
\ls & \: (v_0^{-\f 12}|u_0|^{-1} + v_0^{-1}|u_0|^{-\f 12}) \ls |u_0|^{-\f 14}.
\end{split}
\end{equation}

Using these basic estimates, we now estimate $|O_3|,\dots,|O_6|$. Using \eqref{eq:LinftyOmegav4} and \eqref{eq:estrangeOmega} to bound $\Omg^{-2}(\Omg^2-\Omg_0^2)^2$, we bound $O_3$ via \eqref{eq:Om.block.1}
\begin{equation*}
\begin{split}
|O_3|\lesssim &\: \int_{v_0}^v \int_{-\infty}^u \Omega^{-2}\left(\Omega^2-\Omega_0^2\right)^2\left(v'^2\left|\partial_v \log \left(\frac{\Omega}{\Omega_0}\right)\right|+u'^2\left|\partial_u \log \left(\frac{\Omega}{\Omega_0}\right)\right|\right)\,du'dv'\\
\lesssim &\:\int_{v_0}^v \int_{-\infty}^u (v'+|u'|)^{-2}|u'|^{-1}\left(v'^2\left|\partial_v \log \left(\frac{\Omega}{\Omega_0}\right)\right|+u'^2\left|\partial_u \log \left(\frac{\Omega}{\Omega_0}\right)\right|\right)\,du'dv'\lesssim \:|u_0|^{-\f 14}.
\end{split}
\end{equation*}
Similarly, using \eqref{eq:LinftyQ}, \eqref{r.est}, \eqref{r.est.2}, \eqref{eq:Om.block.1} and \eqref{eq:Om.block.2} to obtain the following estimate for $|O_4|$:
\begin{equation*}
\begin{split}
|O_4|\lesssim &\: \int_{v_0}^v \int_{-\infty}^u \Omega_0^{2}(|r-M|+|r-r_0|+|Q-M|)\left(v'^2\left|\partial_v \log \left(\frac{\Omega}{\Omega_0}\right)\right|+u'^2\left|\partial_u \log \left(\frac{\Omega}{\Omega_0}\right)\right|\right)\,du'dv'\\
\lesssim &\:\int_{v_0}^v \int_{-\infty}^u (v'+|u'|)^{-2}(\mathcal{A}_{\phi}(\mathcal{D}_{\rm o}+\mathcal D_{\rm i})|u'|^{-1}+\mathcal{D}_{o}v'^{-1}+(v'+|u'|)^{-1})\\
&\qquad\cdot\left(v'^2\left|\partial_v \log \left(\frac{\Omega}{\Omega_0}\right)\right|+u'^2\left|\partial_u \log \left(\frac{\Omega}{\Omega_0}\right)\right|\right)\,du'dv'\\
\lesssim &\: (\mathcal{A}_{\phi}(\mathcal{D}_{\rm o}+\mathcal D_{\rm i})+1)|u_0|^{-\f 14} + \mathcal{D}_{o} v_0^{-\f 14}.
\end{split}
\end{equation*}
By Lemma~\ref{RN.est} and \eqref{eq:estrangeOmega}, $|\Omega^2-\Omega_0^{2}|(u',v')\ls (v'+|u'|)^{-2}$. Hence $|O_5|$ can be controlled in a similar manner as $O_4$ as follows:
\begin{equation*}
\begin{split}
|O_5|\lesssim &\: \int_{v_0}^v \int_{-\infty}^u |\Omega^2-\Omega_0^{2}|(|r-M|+|r-r_0|+|Q-M|)\left(v'^2\left|\partial_v \log \left(\frac{\Omega}{\Omega_0}\right)\right|+u'^2\left|\partial_u \log \left(\frac{\Omega}{\Omega_0}\right)\right|\right)\,du'dv'\\
\lesssim &\:\int_{v_0}^v \int_{-\infty}^u (v'+|u'|)^{-2}(\mathcal{A}_{\phi}(\mathcal{D}_{\rm o}+\mathcal D_{\rm i})|u'|^{-1}+\mathcal{D}_{o}v'^{-1}+(v'+|u'|)^{-1})\\
&\qquad\cdot\left(v'^2\left|\partial_v \log \left(\frac{\Omega}{\Omega_0}\right)\right|+u'^2\left|\partial_u \log \left(\frac{\Omega}{\Omega_0}\right)\right|\right)\,du'dv'\\
\lesssim &\:(\mathcal{A}_{\phi}(\mathcal{D}_{\rm o}+\mathcal D_{\rm i})+1)|u_0|^{-\f 14} + \mathcal{D}_{o} v_0^{-\f 14}.
\end{split}
\end{equation*}
Finally, we estimate $|O_6|$. For this we use Lemma~\ref{RN.est}, \eqref{dur.est}, \eqref{dvr.est} and \eqref{eq:Om.block.3} to obtain
\begin{equation*}
\begin{split}
|O_6|\lesssim &\: \int_{v_0}^v \int_{-\infty}^u(|\partial_vr||\partial_ur|+|\partial_vr_0||\partial_ur_0|)\left(v'^2\left|\partial_v \log \left(\frac{\Omega}{\Omega_0}\right)\right|+u'^2\left|\partial_u \log \left(\frac{\Omega}{\Omega_0}\right)\right|\right)\,du'dv'\\
\lesssim &\:\int_{v_0}^v \int_{-\infty}^u [(v'+|u'|)^{-4}+\mathcal{A}_{\phi}^2(\mathcal{D}_{\rm o}+\mathcal D_{\rm i})^2|u'|^{-2}v'^{-2}]\\
&\cdot\left(v'^2\left|\partial_v \log \left(\frac{\Omega}{\Omega_0}\right)\right|+u'^2\left|\partial_u \log \left(\frac{\Omega}{\Omega_0}\right)\right|\right)\,du'dv'\\
\lesssim &\:\mathcal{A}_{\phi}^2(\mathcal{D}_{\rm o}+\mathcal D_{\rm i})^2|u_0|^{-\f 14}.
\end{split}
\end{equation*}

Hence, choosing $v_0$ and $|u_0|$ large in a manner allowed by \eqref{const.adm}, we obtain
\begin{equation*}
\begin{split}
&\:\sup_{v\in [v_0,v_{\infty}]} \int_{-\infty}^{u_0} u^2 \left(\partial_u\left(\log \frac{\Omega}{\Omega_0} \right)\right)^2(u,v)\,du+\sup_{u\in (-\infty,u_0]}\int_{v_0}^{v_{\infty}} v^2 \left(\partial_v\left(\log \frac{\Omega}{\Omega_0} \right)\right)^2(u,v)\,dv\leq \f {M}2,
\end{split}
\end{equation*}
which is to be proved.
\end{proof}

\section{Stability of the Cauchy horizon of extremal Reissner--Nordstr\"om}\label{sec:stability}
We now conclude the bootstrap argument and show that the solution exists and remains regular for all $v\geq v_0$ and that certain estimates hold. More precisely, we have
\begin{proposition}\label{prop:final.est}
There exists a smooth solution $(\phi,r,\Om,A)$ to \eqref{eq:transpeqr}--\eqref{eq:raychv} in the rectangle (cf.~Figure~\ref{fig:fullspacetime})
$$D_{u_0,v_0}=\{(u,v')\,|\, -\infty \leq u\leq u_0, \,v_0\leq v< \infty\}$$
with the prescribed initial data.
Moreover, all the estimates in Sections~\ref{sec:pointwise} and \ref{sec:energy} hold in $D_{u_0,v_0}$.
\end{proposition}
\begin{proof}

For every $v\in [v_0,\infty)$, consider the following conditions:
\begin{enumerate}[(A)]
\item A smooth solution $(\phi,r,\Om,A)$ to \eqref{eq:transpeqr}--\eqref{eq:raychv} exists in the rectangle 
$$D_{u_0,[v_0,v)}=\{(u,v')\,|\, -\infty \leq u\leq u_0, \,v_0\leq v'< v\}$$ 
with the prescribed initial data.
\item The estimates \eqref{eq:bootstrapOmega}, \eqref{eq:bootstrapphi} and \eqref{eq:bootstrapr} hold in $D_{u_0,[v_0,v)}$.
\end{enumerate}
Consider the set $\mathfrak I\subset [v_0,\infty)$ defined by
$$\mathfrak I := \{ v\in [v_0,\infty) : \mbox{(A) and (B) are both satisfied for all $v'\in [v_0, v)$}\}.$$
We will show that $\mathfrak I$ is non-empty, closed and open, which implies that $\mathfrak I = [v_0,\infty)$. Standard local existence implies that $\mathfrak I$ is non-empty. Closedness of $\mathfrak I$ follows immediately from the definition of $\mathfrak I$. The most difficult property to verify is the openness of $\mathfrak I$. For this, suppose $v\in \mathfrak I$. We then argue as follows:
\begin{itemize}
\item Under the bootstrap assumptions \eqref{eq:bootstrapOmega}, \eqref{eq:bootstrapphi} and \eqref{eq:bootstrapr}, all the estimates in Sections~\ref{sec:pointwise} and \ref{sec:energy} hold in $D_{u_0,[v_0,v)}$. A standard propagation of regularity result shows that the solution can be extended smoothly up to 
$$D_{u_0,[v_0,v]}=\{(u,v')\,|\, -\infty \leq u\leq u_0, \,v_0\leq v'\leq  v\}.$$
Hence, one can apply a local existence result for the characteristic initial value problem to show that there exists $\de>0$ such that a smooth solution $(\phi,r,\Om,A)$ to \eqref{eq:transpeqr}--\eqref{eq:raychv} exists in $D_{u_0,[v_0,v+\de)}$.
\item The estimates in \eqref{r.BS.improve}, Corollary~\ref{cor:phi} and Proposition~\ref{prop:eestimateOmega} \emph{improve} those in \eqref{eq:bootstrapOmega}, \eqref{eq:bootstrapphi} and \eqref{eq:bootstrapr}. Hence, by continuity, after choosing $\de>0$ smaller if necessary, \eqref{eq:bootstrapOmega}, \eqref{eq:bootstrapphi} and \eqref{eq:bootstrapr} hold in $D_{u_0,[v_0,v+\de)}$.
\end{itemize}
Combining the two points above, we deduce that after choosing $\de>0$ smaller if necessary, $(v-\de,v+\de)\subset \mathfrak I$. This proves the openness of $\mathfrak I$.
By connectedness of $[v_0,\infty)$, we deduce that $\mathfrak I = [v_0,\infty)$. This implies the existence of a smooth solution in $D_{u_0,v_0}$. Moreover, this implies the assumptions \eqref{eq:bootstrapOmega}, \eqref{eq:bootstrapphi} and \eqref{eq:bootstrapr} that are used in Sections~\ref{sec:pointwise} and \ref{sec:energy} in fact hold throughout $D_{u_0,v_0}$. Therefore, indeed all the estimates in Sections~\ref{sec:pointwise} and \ref{sec:energy} hold in $D_{u_0,v_0}$.
\end{proof}

We have therefore shown the existence of a solution in the whole region $D_{u_0,v_0}$. Since we have now closed our bootstrap argument, \textbf{in the remainder of the paper, we will suppress any dependence on $\mathcal A_\phi$} (which in turn depends only on $M$, $\mfm$ and $\mfe$). 

In the remainder of this section, we show that one can attach a \emph{Cauchy horizon} to the solution and prove regularity of the solution \emph{up to} the Cauchy horizon. More precisely, define $V$ to be a function of $v$ in exactly the same manner as in Section~\ref{sec:coord.CH}, i.e. 
\begin{equation}\label{Vvrelation}
\f{dV}{dv} = \Omega_0^2(1,v),\quad V(\infty) = 0.
\end{equation}
We will use also the convention that
$$V_0:=V(v_0).$$
Define moreover (as in Section~\ref{sec:coord.CH}) the \emph{Cauchy horizon} $\CH$ as the boundary $\{V=0\}$ in the $(u,V,\theta,\varphi)$ coordinate system. Note that this induces a natural differential structure on $D_{u_0,v_0}\cup \CH$. In the new coordinate system, in order to distinguish the ``new $\Om$'', we follow the convention in Section~\ref{sec:coord.CH} and denote $g(\rd_u,\rd_V) = -\f 12\widetilde{\Om}^{2}$ instead. We show that the solution $(\phi,r,\widetilde{\Om},A)$ (after choosing an appropriate gauge for $A$) extends to the Cauchy horizon continuously, and that in fact their derivatives are in $L^2_{loc}$ up to the Cauchy horizon. (In fact, as we will show in Section~\ref{sec:extension}, there are non-unique extensions as spherically symmetric solutions to \eqref{EMCSFS} beyond the Cauchy horizon.)

We begin by restating some of the estimates we have obtained in this new coordinate system.
\begin{lemma}\label{lm:reg}
In the $(u,V)$ coordinate system, $\phi$, $r$ and $\widetilde{\Om}$ satisfy the following estimates:
\begin{equation}\label{eq:reg.Om.rdVr}
\f{1}{|u|^2} \ls \widetilde{\Om}^{2}(u,V)\ls 1,\quad |\rd_V r(u,V)|\ls \mathcal{D}_{\rm o}+\mathcal D_{\rm i},
\end{equation}
\begin{equation}\label{eq:reg.phi}
\int_{V_0}^0 |\rd_V\phi|^2(u,V')\, dV' + \int_{-\infty}^{u_0} u^2|\rd_u\phi|^2(u',V)\, du' \ls \mathcal{D}_{\rm o}+\mathcal D_{\rm i},
\end{equation}
\begin{equation}\label{eq:reg.Om}
\int_{V_0}^0 |\rd_V\log\widetilde{\Om}|^2(u,V')\, dV' + \int_{-\infty}^{u_0} u^2|\rd_u\log\widetilde{\Om}|^2(u',V)\, du' \ls 1,
\end{equation}
\begin{equation}\label{eq:reg.r}
\int_{V_0}^0 |\rd_V r|^2(u,V')\, dV' + \int_{-\infty}^{u_0} u^2|\rd_ur|^2(u',V)\, du' \ls \mathcal{D}_{\rm o}+\mathcal D_{\rm i}.
\end{equation}
\end{lemma}
\begin{proof}
\textbf{Proof of estimates for $\widetilde{\Om}$ in \eqref{eq:reg.Om.rdVr}.} By \eqref{eq:LinftyOmegav2} and \eqref{Vvrelation},
\begin{equation*}
\begin{split}
\left|\widetilde{\Om}^2(u,V)-\frac{4M^2}{(v(V)+|u|)^2}\Omega^{-2}_0(-1,v(V))\right| = &\:\Omega^{-2}_0(-1,v(V)) \left| \Om^2(u,v(V))-\frac{4M^2}{(v+|u|)^2}\right|\\
\ls &\: |u|^{-\f 12}\bigg(\f{v(V)+1}{v(V)+|u|}\bigg)^2,
\end{split}
\end{equation*}
which implies both the upper and lower bounds for $\widetilde{\Om}$ in \eqref{eq:reg.Om.rdVr}.

\textbf{Proof of estimate for $\rd_V r$ in \eqref{eq:reg.Om.rdVr}.} The estimate for $\rd_V r$ in \eqref{eq:reg.Om.rdVr} follows from \eqref{dvr.est} and \eqref{Vvrelation}.

\textbf{Proof of \eqref{eq:reg.phi} and \eqref{eq:reg.Om}.} These follow from Corollary~\ref{cor:phi}, Proposition~\ref{prop:eestimateOmega}, \eqref{Vvrelation} and \eqref{eq:Om0.est}.

\textbf{Proof of \eqref{eq:reg.r}.} Finally, \eqref{eq:reg.r} can be obtained by directly integrating the pointwise estimates in \eqref{dur.est} (for $\rd_u r$) and \eqref{eq:reg.Om.rdVr} (for $\rd_V r$).
\end{proof}

We also have the following $W^{1,2}$ estimate for the charge $Q$.
\begin{lemma}\label{lm:Q}
In the $(u,V)$ coordinate system, $Q$ satisfies the following estimate:
$$\int_{V_0}^{0} |\rd_V Q|^2(u,V')\, dV' + \int_{-\infty}^{u_0} u^2|\rd_u Q|^2(u',V)\, du' \ls \mathcal{D}_{\rm o} + \mathcal D_{\rm i}.$$
\end{lemma}
\begin{proof}
\textbf{Estimate for $\rd_V Q$.} By \eqref{eq:Q2} (adapted to the $(u,V)$ coordinate system), $\rd_V Q = -2\pi i r^2\mfe (\phi \overline{D_V\phi} - \overline{\phi} D_V\phi).$ Therefore, using \eqref{eq:Linftyphi}, \eqref{r.est} and \eqref{eq:Linftyphi},
$$\int_{V_0}^{0} |\rd_V Q|^2(u,V')\, dV' \ls \int_{V_0}^{0} |D_V \phi|^2(u,V')\, dV' \ls \mathcal{D}_{\rm o} + \mathcal D_{\rm i}.$$

\textbf{Estimate for $\rd_u Q$.} By \eqref{eq:Q1}, we have $\partial_uQ = 2\pi i r^2\mathfrak{e} (\phi \overline{D_u\phi}-\overline{\phi}D_u\phi)$. The desired estimate hence follows similarly as above using \eqref{eq:Linftyphi}, \eqref{r.est} and \eqref{eq:reg.phi}. 
\end{proof}

In order to consider the extension, we will also need to choose a gauge for $A_\mu$. We will fix $A$ such that 
\begin{equation}\label{A.gauge.con}
\mbox{$A_u=0$ everywhere and $A_V=0$ on the null hypersurface $\{u=u_0\}$.}
\end{equation}
To see that this is an acceptable gauge choice, simply notice that given any $\tilde{A}_u$, $\tilde{A}_V$, we can define
$$\chi(u,V) = \int_{u_0}^u \tilde{A}_u(u,V) \, du'+ \int_{V_0}^V \tilde{A}_V(u_0,V')\, dV',$$
where $V_0=V(v_0)$. This implies
$$A_u(u,V) = \tilde{A}_u(u,V) - (\rd_u\chi)(u,V) = 0, \,\forall u,\forall V,\quad A_V(u_0,V)=\tilde{A}_V(u,V)-(\rd_v\chi)(u_0,V)=0,\,\forall V.$$
Now in the gauge \eqref{A.gauge.con}, we have the following estimates:
\begin{lemma}\label{lm:AV}
Suppose $A$ satisfies the gauge condition above. Then $A_V$, $\rd_u A_V$ and $\rd_V A_V$ obey the following estimates
$$\sup_{u\in (-\infty,u_0],\,V\in [V_0,0)} |A_V(u,V)|\ls (u_0-u),\quad \sup_{V\in [V_0,0)}\int_{u}^{u_0} |\rd_uA_V|^2(u',V)\, du'\ls (u_0-u),$$
$$\sup_{u\in (-\infty,u_0]}\int_{V_0}^0 |\rd_V A_V|^2(u,V')\, dV' \ls (\mathcal{D}_{\rm o} + \mathcal D_{\rm i})(u_0-u).$$
\end{lemma}
\begin{proof}
\textbf{Pointwise estimate for $A_V$.} By \eqref{QintermsofA} (adapted to the $(u,V)$ coordinate system),
\begin{equation}\label{du.av}
\rd_u A_V = \f 12 \f{\widetilde{\Om}^2 Q}{r^2}.
\end{equation}
Using \eqref{eq:LinftyQ}, \eqref{r.est} and \eqref{eq:reg.Om.rdVr}, and the fact that $A_V(u_0,V)=0$, we obtain that for any $u\leq u_0$,
$$|A_V(u,V)|\ls \int_{u}^{u_0} \, du' = (u_0-u).$$

\textbf{$L^2$ estimate for $\rd_u A_V$.} To obtain the desired $L^2_u$ estimate for $\rd_uA_V$, we simply use the fact that the RHS of \eqref{du.av} is bounded (as shown above using \eqref{eq:LinftyQ}, \eqref{r.est} and \eqref{eq:reg.Om.rdVr}) and integrate it up in $u$.

\textbf{$L^2$ estimate for $\rd_V A_V$.} To estimate $\rd_V A_V$, we differentiate \eqref{du.av} in $V$ to obtain
\begin{equation}\label{du.dv.av}
\rd_u \rd_V A_V = \f 12 \rd_V (\f{\widetilde{\Om}^2 Q}{r^2}).
\end{equation}
Using the pointwise bounds in \eqref{eq:LinftyQ}, \eqref{r.est} and \eqref{eq:reg.Om.rdVr}, and the $L^2_V$ estimates in \eqref{eq:reg.phi}, \eqref{eq:reg.r} and \eqref{eq:reg.Om}, we obtain
$$\sup_{u\in (-\infty,u_0]}\int_{V_0}^{0} \left|\rd_V \left(\f{\widetilde{\Om}^2 Q}{r^2}\right)\right|^2(u,V')\, dV' \ls \mathcal{D}_{\rm o} + \mathcal D_{\rm i}.$$
Now since $\rd_V A_V(u_0,V)=0$ for all $V$, we have, for any $u\leq u_0$,
$$\int_{V_0}^0 |\rd_V A_V|^2(u,V')\, dV' \leq \int_u^{u_0} |\rd_u\rd_V A_V|^2(u',V')\, dV'\, du'\ls (\mathcal{D}_{\rm o} + \mathcal D_{\rm i})(u_0-u).$$
\end{proof}

\begin{proposition}\label{prop:C0.ext}
Let $V$ be as in \eqref{Vvrelation} and $A$ satisfy the gauge condition \eqref{A.gauge.con}. Then in the $(u,V)$ coordinate system, 
\begin{itemize}
\item $\phi$, $r$, $\widetilde{\Om}$, $A_V$ and $Q$ (as functions of $(u,V)\in (-\infty,u_0]\times [V_0,0)$) can be continuously extended to the Cauchy horizon $\{V=0\}$.
\item The extensions of $\phi$, $r$, $\widetilde{\Om}$, $A_V$ and $Q$ (as functions of $(u,V)\in (-\infty,u_0]\times [V_0,0]$) are all in $C^{0,\f 12}\cap W^{1,2}_{loc}$.
\item The Hawking mass $m$ (as a function of $(u,V)\in (-\infty,u_0]\times [V_0,0)$) can be continuously extended to the Cauchy horizon $\{V=0\}$.
\end{itemize}
\end{proposition}
\begin{proof}
\textbf{Continuous extendibility and H\"older estimates.} Let us first consider in detail the estimates for $\phi$. As we will explain, the estimates for $r$, $\widetilde{\Om}$, $A_V$ and $Q$ are similar. Consider two points $(u',V')$ and $(u'',V'')$. Denote $v'=v(V')$, $v''=v(V'')$, where $v$ is the inverse function of $v\mapsto V$ above. Then we have, using the fundamental theorem of calculus and the Cauchy--Schwarz inequality,
\begin{equation}\label{phi.C12}
\begin{split}
&\: |\phi(u',V')-\phi(u'',V'')|\\
\leq & \: \left|\int_{u'}^{u''} |\rd_u\phi|(u''',V')\, du'''\right| + \left| \int_{V'}^{V''} |\rd_V\phi|(u'',V''')\, dV'''\right|\\
\leq & \: \left|\int_{u'}^{u''} |D_u\phi|(u''',V')\, du'''\right| + \left| \int_{V'}^{V''} |D_V\phi|(u'',V''')+ |A_V||\phi|(u'',V''')\, dV'''\right|\\
\ls & \:|u'-u''|^{\f 12} \left(\int_{u'}^{u''} |D_u\phi|^2(u''',v')\, du'''\right)^{\f 12} + |V'-V''|^{\f 12}\left( \int_{v'}^{v''} (v'''+1)^2|D_v\phi|^2(u'',v''')\, dv'''\right)^{\f 12}\\
&\: + (\mathcal{D}_{\rm o}^{\f 12}+ \mathcal D_{\rm i}^{\f 12})(u_0-u'')|V'-V''|\\
\ls & \: (\mathcal{D}_{\rm o} + \mathcal D_{\rm i}) (|u'-u''|^{\f 12}+ |V'-V''|^{\f 12}) + (\mathcal{D}_{\rm o}^{\f 12}+ \mathcal D_{\rm i}^{\f 12})|V'-V''|,
\end{split}
\end{equation}
where in the last two lines we have used \eqref{eq:Linftyphi}, \eqref{eq:reg.phi} and Lemma~\ref{lm:AV}.

In a similar manner, using Lemmas~\ref{lm:reg}, \ref{lm:Q} and \ref{lm:AV} instead, $r$, $\widetilde{\Om}$, $A_V$ and $Q$ can be estimated as follows\footnote{In fact, the estimates for $r$, $\widetilde{\Om}$ $A_V$ and $Q$ are simpler as we do not need to handle the difference between $\rd_V$ and $D_V$.}: (To simplify the exposition, we suppress the discussion on the explicit dependence of the constant on $\mathcal{D}_{\rm o}$ and $\mathcal D_{\rm i}$.)
\begin{equation}\label{other.C12}
\begin{split}
&\:|r(u',V')-r(u'',V'')|+|\widetilde{\Om}(u',V')-\widetilde{\Om}(u'',V'')|\\
&\:\quad+|A_V(u',V')-A_V(u'',V'')|+|Q(u',V')-Q(u'',V'')|
\ls_{\mathcal{D}_{\rm o},\mathcal D_{\rm i}} (|u'-u''|^{\f 12}+ |V'-V''|^{\f 12}).
\end{split}
\end{equation}
Define the extension of $(\phi,r,\widetilde{\Om},A_V,Q)$ by
$$\phi(u,V=0) := \lim_{V\to 0}\phi(u,V),\quad r(u,V=0) := \lim_{V\to 0}r(u,V),$$
$$\widetilde{\Om}(u,V=0) := \lim_{V\to 0}\widetilde{\Om}(u,V),\quad A_V(u,V=0) := \lim_{V\to 0}A_V(u,V),\quad Q(u,V=0) := \lim_{V\to 0}Q(u,V).$$
The estimates in \eqref{phi.C12} and \eqref{other.C12} above show that the extensions are well-defined and that the extensions of $(\phi,r,\widetilde{\Om},A_V,Q)$ is indeed $C^{0,\f 12}$.

\textbf{$W^{1,2}_{loc}$ estimates.} Now that we have constructed an extension of $(\phi,r,\widetilde{\Om},A_V,Q)$ to $D_{u_0,v_0}\cup \CH$, it follows immediately from Lemmas~\ref{lm:reg}, \ref{lm:AV} and \ref{lm:Q} that the extension is in $W^{1,2}_{loc}$.

\textbf{$C^0$ extendibility of the Hawking mass.} Finally, we prove the $C^0$ extendibility of the Hawking mass (whose definition we recall from \eqref{Hawking.mass}). By \eqref{eq:dum} and \eqref{eq:dvm} (appropriately adapted in the $(u,V)$ coordinate system), we have
\begin{equation}\label{eq:dum.1}
\rd_u m = -8\pi \f{r^2(\rd_V r)}{\widetilde{\Om}^2}|D_u\phi|^2+2(\rd_u r)\mfm^2 \pi r^2|\phi|^2+\f 12 \f{(\rd_u r)Q^2}{r^2},
\end{equation}
\begin{equation}\label{eq:dvm.1}
\rd_V m = -8\pi \f{r^2(\rd_u r)}{\widetilde{\Om}^2}|D_V\phi|^2+2(\rd_V r)\mfm^2 \pi r^2|\phi|^2+\f 12 \f{(\rd_V r)Q^2}{r^2}.
\end{equation}
It now follows from \eqref{eq:LinftyQ}, \eqref{r.est}, \eqref{dvr.est} and Lemma~\ref{lm:reg} that the RHS of \eqref{eq:dum.1} is bounded in $L^1_u$ and the RHS of \eqref{eq:dvm.1} is bounded in $L^1_V$. This implies the following $L^1$ estimates
\begin{equation}\label{m.cont.1}
\int_{-\infty}^{u_0} |\rd_u m|(u',V)\, du' + \int_{V_0}^{0} |\rd_V m|(u,V')\, dV'\ls_{\mathcal{D}_{\rm o},\,\mathcal D_{\rm i}} 1.
\end{equation}
On the other hand, by the fundamental theorem of calculus,
\begin{equation}\label{m.cont.2}
\begin{split}
|m(u',V')-m(u'',V'')|
\leq & \: \left|\int_{u'}^{u''} |\rd_um|(u''',V')\, du'''\right| + \left| \int_{V'}^{V''} |\rd_Vm|(u'',V''')\, dV'''\right|.
\end{split}
\end{equation}
Combining \eqref{m.cont.1} and \eqref{m.cont.2}, we see that 
\begin{enumerate}
\item $m$ can be extended to $\CH$ by
$$m(u,0)= \lim_{V\to 0} m(u,V),$$
and that
\item the extension is continuous up to $\CH$,
\end{enumerate}
which concludes the proof of the proposition. (Let us finally note that since we only have $L^1$ (as opposed to $L^2$) estimates for $\rd_u m$ and $\rd_V m$, we only show that $m$ is continuous, but do \emph{not} obtain any H\"older estimates.) \qedhere
\end{proof}

\begin{remark}[$C^{0,\f 12}\cap W^{1,2}_{loc}$ regularity in the $(3+1)$-dimensional spacetime]
In Proposition~\ref{prop:C0.ext}, we proved that the extensions of $\phi$, $r$, $\widetilde{\Om}$, $A_V$ and $Q$ are $C^{0,\f 12}\cap W^{1,2}_{loc}$ on the $(1+1)$-dimensional quotient manifold $\mathcal Q$ (cf.~notations in Section~\ref{doublenull}). It easily follows that these functions, when considered as functions on $\mathcal M = \mathcal Q\times \mathbb S^2$ are also in $C^{0,\f 12}\cap W^{1,2}_{loc}$. As a consequence, in the coordinate system $(u,v,\theta,\varphi)$, the spacetime metric, the scalar field and the electromagnetic potential all extend to the Cauchy horizon in a manner that is in the $(3+1)$-dimensional spacetime norm $C^{0,\f 12}\cap W^{1,2}_{loc}$.
\end{remark}

\section{Constructing extensions beyond the Cauchy horizon}\label{sec:extension}

In this section, we prove that the solution can be extended locally beyond the Cauchy horizon \textbf{as a spherically symmetric $W^{1,2}$ solution to \eqref{EMCSFS}} (in a non-unique manner). Together with Propositions~\ref{prop:final.est} and \ref{prop:C0.ext}, this completes the proof of Theorem~\ref{thm:main}.

The idea behind the construction of the extension is that the system \eqref{EMCSFS} is \emph{locally well-posed} in spherical symmetry for data such that $\rd_V\phi$, $\rd_V r$ and $\rd_V\log\widetilde{\Omega}$ are merely in $L^2$ (when $r$ and $\Om$ are bounded away from $0$). This follows from the well-known fact that $(1+1)$-dimensional wave equations are locally well-posed with $W^{1,2}$ data. Related results in the context of general relativity can be found throughout the literature; see for instance \cite{CGNS1, LR2, LeSt}. For completeness, we give a proof in our specific setting.

The section is organized as follows. We first discuss a general local well-posedness result on $(1+1)$-dimensional wave equation (cf.~Definition~\ref{def:sol} and Proposition~\ref{prop:gen.wave}). We then apply the wave equation result in our setting to construct extensions to our spacetime solutions by solving appropriate characteristic initial value problems. In particular, since we will be able to prescribe data for the construction of the extensions, there are (infinitely many) non-unique extensions.

We begin by considering a general class of $(1+1)$-dimensional wave equation and introduce the following notion of solution, which makes sense when the derivative of $\Psi$ is only in $L^2$ in one of the null directions.
\begin{definition}\label{def:sol}
Let $k\in \mathbb N$. Consider a wave equation\footnote{For $k>1$, this should be thought of as a system of wave equations.} for $\Psi:[0,\ep)\times [0,\ep) \to \mathcal V$ (where $\mathcal V\subset \mathbb R^k$ is an open subset) of the form
\begin{equation}\label{eq:gen.wave}
\rd_u \rd_v \Psi_{A} = f_A(\Psi)+ N_A^{BC}(\Psi)\rd_u\Psi_B \rd_v\Psi_C + K_A^{BC}(\Psi) \rd_u \Psi_B \rd_u \Psi_C + L_A^B(\Psi)\rd_u\Psi_B + R_A^B(\Psi)\rd_v\Psi_B,
\end{equation}
where $\Psi_A$ denotes the components of $\Psi$, $f_A, N_A^{BC}, K_A^{BC}, L_A^B, R_A^B:\mathcal V\to \mathbb R$ are smooth, and we sum over all repeated capital Latin indices.

We say that a continuous function $\Psi:[0,\ep)\times [0,\ep)\to \mathcal V$ satisfying $\rd_v\Psi \in L^2_v(C^0_u)$ and $\rd_u\Psi \in C^0_u C^0_v$ is a \textbf{solution in the integrated sense} if
$$(\rd_v \Psi_A)(u,v) = (\rd_v \Psi_A)(0,v) + \int_0^u \mbox{(RHS of \eqref{eq:gen.wave})}(u',v)\, du' ,\quad \mbox{for all $u\in [0,\ep)$ and for a.e.~$v\in [0,\ep)$}$$
and
$$(\rd_u \Psi_A)(u,v) = (\rd_u \Psi_A)(u,0) + \int_0^v \mbox{(RHS of \eqref{eq:gen.wave})}(u,v')\, dv' ,\quad \mbox{for all $v\in [0,\ep)$ and for a.e.~$u\in [0,\ep)$}.$$
\end{definition}

\begin{remark}
Given a solution $\Psi$ in the sense of Definition~\ref{def:sol}, it is also a weak solution in the following sense: for any $\chi\in C_c^\infty$, 
$$\iint (\rd_u \chi)(u,v)(\rd_v\Psi)(u,v) \, du dv = -\iint \chi(u,v)\mbox{(RHS of \eqref{eq:gen.wave})}(u,v)\, du dv$$
and
$$\iint (\rd_v \chi)(u,v)(\rd_u\Psi)(u,v) \, du dv = -\iint \chi(u,v)\mbox{(RHS of \eqref{eq:gen.wave})}(u,v)\, du dv.$$
\end{remark}

The following is a general local existence result for $(1+1)$-D wave equation where $\rd_v\Psi$ is initially only in $L^2_v$. We construct local solutions in the sense of Definition~\ref{def:sol}. (Let us note that the following wave equation result holds for rougher data where $\rd_v\Psi$ is only in $L^1_v$. This will however be irrelevant to our problem; see Remark~\ref{rmk:L1sol}.)
\begin{proposition}\label{prop:gen.wave}
Consider the setup in Definition~\ref{def:sol}. Let $\mathcal K\subset\mathcal V$ be a compact subset. Given initial data to the wave equation \eqref{eq:gen.wave} on two transversely intersecting characteristic curves $\{(u,v):u=0,v\in [0,v_*]\}\cup \{(u,v):v=0,u\in [0,u_*]\}$ such that 
\begin{itemize}
\item $\Psi$ takes value in $\mathcal K$; and
\item the following estimates hold for the derivatives of $\Psi$ for some $C_{wave}>0$:
$$\int_0^{v_*} |\rd_v\Psi|^2(0,v')\, dv'\leq C_{wave},\quad \sup_{u\in [0,u_*]}|\rd_u\Psi|^2(u',0)\leq C_{wave}.$$ 
\end{itemize}
Then, there exist $\ep_{wave}>0$ depending on $\mathcal K$ and $C_{wave}$ (and the equation) such that there exists a unique solution to \eqref{eq:gen.wave} in the sense of Definition~\ref{def:sol} in the region
$$(u,v)\in \{(u,v): u\in [0,\ep_{wave}),\, v\in [0,\ep_{wave})\}$$
which achieves the prescribed initial data.
\end{proposition}

\begin{proof}
We directly work with the formulation in Definition~\ref{def:sol} and prove the existence and uniqueness of integral solutions. This proposition can be proven via a standard iteration argument. In order to illustrate the main idea and the use of the structure of the nonlinearity, we will only discuss below the proof of \emph{a priori estimates}.

By a bootstrap argument, we assume that 
\begin{equation}\label{wave.BA}
\sup_{\substack{u'\in [0,\ep_{wave})\\v'\in [0,\ep_{wave})}}|\rd_u\Psi|(u',v')\leq 4C_{wave}.
\end{equation}
Let $\mathcal K'\subset \mathcal V$ be a fixed compact set such that $\mathcal K\subset \mathring{\mathcal{K}'}$. We estimate $\Psi$ using the fundamental theorem of calculus as follows:
\begin{equation*}
\begin{split}
\sup_{\substack{u'\in [0,\ep_{wave})\\v'\in [0,\ep_{wave})}}|\Psi(u',v')-\Psi(0,v')| \leq &\sup_{v'\in [0,\ep_{wave})}\int_0^{\ep_{wave}} |\rd_u\Psi|(u',v') \,du' \\
\leq &\ep_{wave} \sup_{\substack{u'\in [0,\ep_{wave})\\v'\in [0,\ep_{wave})}}|\rd_u\Psi|(u',v') \leq 4C_{wave}\ep_{wave} .
\end{split}
\end{equation*}
Using the compactness of $\mathcal K$, we can choose $\ep_{wave}$ sufficiently small so that $\Psi(u,v)\in \mathcal K'$ for all $u\in [0,\ep)$
Now that we have estimated $\Psi$, since $\mathcal K'$ is compact, it follows that $f_A(\Psi), N_A^{BC}(\Psi), K_A^{BC}(\Psi), L_A^B(\Psi), R_A^B(\Psi)$ are all bounded. From now on, we will use these bounds and write $C$ for constants that are allowed to depend on $\sup_{x\in \mathcal K'}f_A(x)$, etc.

We now turn to the estimates for the derivatives of $\Psi$. First, we bound $\rd_v\Psi$ using (the integral form of) \eqref{eq:gen.wave} and H\"older's inequality and Young's inequality:
\begin{equation*}
\begin{split}
& \int_0^{\ep_{wave}} \sup_{u'\in [0,\ep_{wave})}|\rd_v\Psi|^2(u',v')\, dv'\\
\leq & C_{wave} + C \int_0^{\ep_{wave}} \int_0^{\ep_{wave}} |\rd_v\Psi| (1+|\rd_v\Psi|+|\rd_u\Psi|+|\rd_v\Psi||\rd_u\Psi|+|\rd_u\Psi|^2)(u',v') \, du'dv' \\
\leq & C_{wave} + C \left(1+\int_0^{\ep_{wave}} \sup_{u'\in [0,\ep_{wave})}|\rd_v\Psi|^2(u',v')\, dv'\right)\left(\ep_{wave} + \ep_{wave} \sup_{\substack{u'\in [0,\ep_{wave})\\v'\in [0,\ep_{wave})}}|\rd_u\Psi|(u',v') \right)\\
&\: + C_{wave}\ep_{wave}^2\sup_{\substack{u'\in [0,\ep_{wave})\\v'\in [0,\ep_{wave})}}|\rd_u\Psi|(u',v').
\end{split}
\end{equation*}
For $\rd_u\Psi$, we again use (the integral form of) \eqref{eq:gen.wave} and H\"older's inequality and Young's inequality to get
\begin{equation*}
\begin{split}
 &\: \sup_{\substack{u'\in [0,\ep_{wave})\\v'\in [0,\ep_{wave})}}|\rd_u\Psi|(u',v')\\
\leq &\: C_{wave} + C \sup_{u'\in [0,\ep_{wave})} \int_0^\ep  (1+|\rd_v\Psi|+|\rd_u\Psi|+|\rd_v\Psi||\rd_u\Psi|+|\rd_u\Psi|^2)(u',v'))(u',v') \, dv' \\
\leq &\: C_{wave} + C \left(1+\int_0^\ep \sup_{u'\in [0,\ep_{wave}]}|\rd_v\Psi|^2(u',v')\, dv'\right)\left(\ep_{wave}^{\f 12} + \ep_{wave} \sup_{\substack{u'\in [0,\ep_{wave})\\v'\in [0,\ep_{wave})}}|\rd_u\Psi|^{2}(u',v')\right)\\
\leq &\: C_{wave} + C \left(1+\int_0^\ep \sup_{u'\in [0,\ep_{wave}]}|\rd_v\Psi|^2(u',v')\, dv'\right)\left(\ep_{wave}^{\f 12} + \ep_{wave} C_{wave} \sup_{\substack{u'\in [0,\ep_{wave})\\v'\in [0,\ep_{wave})}}|\rd_u\Psi|(u',v')\right).
\end{split}
\end{equation*}
Summing the above two estimates and choosing $\ep_{wave}$ sufficient small (depending on $C_{wave}$ and $\mathcal K'$), it follows that 
$$\int_0^\ep \sup_{u'\in [0,\ep_{wave})}|\rd_v\Psi|^2(u',v')\, dv'+\sup_{\substack{u'\in [0,\ep_{wave})\\v'\in [0,\ep_{wave})}}|\rd_u\Psi|(u',v')\leq 2C_{wave}.$$
This in particular improves the bootstrap assumption \eqref{wave.BA} so that we conclude the argument.\qedhere
\end{proof}

We now use Proposition~\ref{prop:gen.wave} to solve \eqref{EMCSFS}. In particular, this allows us to extend the solution in $D_{u_0,v_0}$ (in infinitely many ways!) beyond the Cauchy horizon \emph{as a spherically symmetric strong solution} to \eqref{EMCSFS}. Before we proceed, let us define a notion of spherically symmetric strong solutions to \eqref{EMCSFS} (using Definition~\ref{def:sol}) appropriate for our setting. For simplicity, in our notion of spherically symmetric strong solutions, we will already fix a gauge so that $A_u=0$.
\begin{definition}\label{def:strong.solutions}
Let $(\phi,\Om,r,A_v,Q)$ be continuous functions on $\{(u,v): u\in [u_0,u_0+\ep),\, v\in [v_0,v_0+\ep)\}$ for some $\ep>0$ with $\phi$ complex-valued, $(\Om,r,A_v,Q)$ real-valued and $\Om,\,r>0$. We say that $(\phi,\Om,r,A_v,Q)$ is a spherically symmetric strong solution to \eqref{EMCSFS} if the following hold\footnote{We remark that \eqref{eq:phi2} is not explicitly featured below. Note however that \eqref{eq:phi2} follows as an immediate consequence of \eqref{eq:lastdef.Av}.}:
\begin{itemize}
\item $(\phi,\Om,r,A_v,Q)$ are in the following regularity classes: 
$$\rd_v \phi,\,\rd_v\log\Om \in L^2_v(C^0_u),\quad \rd_u \phi,\,\rd_u\log\Om,\,\rd_u r,\,\rd_v r,\,\rd_u A_v \in C^0_uC^0_v.$$
\item \eqref{eq:transpeqr}, \eqref{eq:waveqOmega} and \eqref{eq:phi1} are satisfied as wave equations in the integrated sense as in Definition~\ref{def:sol} after replacing $D_v \mapsto \rd_v + i \mathfrak{e} A_v$, $D_u \mapsto \rd_u$.
\item \eqref{eq:Q1}, \eqref{eq:Q2}, \eqref{eq:raychu} and \eqref{eq:raychv} are all satisfied in the integrated sense as follows, again with the understanding that $D_v \mapsto \rd_v + i \mathfrak{e} A_v$, $D_u \mapsto \rd_u$:
\begin{align}
\label{eq:intconstraint1}
Q(u,v)=&\: Q(0,v)+\int_{u_0}^u \left[2\pi i r^2 \mfe(\phi \overline{D_u\phi}-\overline{\phi}D_u\phi)\right] (u',v)\,du',\\
\label{eq:intconstraint2}
Q(u,v)=&\: Q(u,0)-\int_{v_0}^v \left[2\pi i r^2\mathfrak{e} (\phi \overline{D_v\phi}-\overline{\phi}D_v\phi)\right] (u,v')\,dv',\\
\label{eq:intconstraint3}
r\partial_u r(u,v)=&\: r\partial_u r(0,v)+\int_{u_0}^u \left[2 r\partial_u r\partial_u\log \Omega+(\partial_ur)^2-4\pi r^2|D_u\phi|^2 \right](u',v)\,du',\\
\label{eq:intconstraint4}
r\partial_v r(u,v)=&\: r\partial_v r(u,0)+\int_{v_0}^v \left[2 r\partial_v r\partial_v\log \Omega+ (\partial_v r)^2-4\pi r^2|D_v\phi|^2 \right](u,v')\,dv',
\end{align}
for all $(u,v) \in \{(u,v): u\in [u_0,u_0+\ep),\, v\in [v_0,v_0+\ep)\}$.
\item \eqref{QintermsofA} is satisfied classically everywhere with $A_u = 0$, i.e.
\begin{equation}\label{eq:lastdef.Av}
\rd_u A_v = \f{Q\Om^2}{2r^2}.
\end{equation}
\end{itemize}

\end{definition}

We emphasize again that a spherically symmetric strong solution to \eqref{EMCSFS} in the sense of Definition~\ref{def:strong.solutions} is a fortiori a weak solution to \eqref{EMCSFS} in the sense of Remark~\ref{rmk:low.regularity}.

We now construct extensions to the solutions given by Proposition~\ref{prop:final.est} beyond the Cauchy horizon as spherically symmetric strong solutions to \eqref{EMCSFS}: 
\begin{proposition}\label{prop:extension}
For every $u_{ext}\in (-\infty,u_0)$, there exists $\ep_{ext}>0$ such that there are infinitely many inequivalent extensions $(\phi,\widetilde{\Om},r,A_V,Q)$ to the region $D_{u_0,v_0}\cup\CH\cup\{(u,V):u\in [u_{ext},u_{ext}+\ep_{ext}],\,V\in [0,\ep_{ext})\}$, each of which is a spherically symmetric strong solution to \eqref{EMCSFS} (cf.~Definition~\ref{def:strong.solutions}).
\end{proposition}
\begin{proof}
Let us focus the discussion on constructing \emph{one} such extension. It will be clear at the end that the argument indeed gives infinitely many inequivalent extensions.

\textbf{Setting up the initial data.}  Extend the constant $u$ curve $\{u=u_{ext}\}$ up to the Cauchy horizon. We will consider a sequence of characteristic initial problems with initial data given on $\{u=u_{ext}\}$ and $\{V=V_n\}$ where $V_n$ approaches the Cauchy horizon, i.e.~$V_n\to 0$. For a fixed $n\in \mathbb N$, the data on $\{V=V_n\}$ are simply induced by the solution that we have constructed in Proposition~\ref{prop:final.est}. On $\{u=u_{ext}\}$, the data when $V\in [V_n,0)$ are induced by the solution, but we prescribe data for $V\geq 0$ (i.e.~beyond the Cauchy horizon) by the following procedure:
\begin{itemize}
\item (\textbf{Data for $\widetilde{\Om}$.}) As we showed in \eqref{eq:reg.Om.rdVr}, \eqref{eq:reg.Om} and Proposition~\ref{prop:C0.ext}, for a fixed $u_{ext}$, $\widetilde{\Omg}(u_{ext},V)$ is continuous up to $\{V=0\}$, is bounded away from $0$, and $\rd_V\widetilde{\Omg}(u_{ext},V)\in L^2_V$. We can therefore extend $\Omg$ to $\{(u_{ext},V):V\geq 0\}$ so that it is continuous and bounded away from $0$ and that $\rd_V\widetilde{\Omg}(u_{ext},V) \in L^2_V$. 
\item (\textbf{Data for $\phi$.}) As we showed in \eqref{eq:reg.phi} and Proposition~\ref{prop:C0.ext}, $\phi(u_{ext},V)$ is continuous up to $\{V=0\}$ and $D_V\phi(u_{ext},V) \in L^2_V$. Since by Lemma~\ref{lm:AV}, $|A_V|(u_{ext},V)\ls (u_0-u_{ext})$ for $V\leq 0$, this also implies that $\rd_V\phi(u_{ext},V) \in L^2_V$. We can therefore extend $\phi$ to $\{(u_{ext},V):V\geq 0\}$ so that it is continuous and $\rd_V\phi(u_{ext},V) \in L^2_V$. 
\item (\textbf{Data for $A_V$.}) Next, by Lemma~\ref{lm:AV}, $A_V(u_{ext},V)$ is continuous up to $\{V=0\}$ and $\rd_V\phi(u_{ext},V) \in L^2_V$. Thus, just like $\widetilde{\Omg}$ and $\phi$, we can extend $A_V$ to $\{(u_{ext},V):V\geq 0\}$ so that it is continuous and $\rd_VA_V(u_{ext},V) \in L^2_V$. 
\item (\textbf{Data for $r$.}) Finally, we prescribe $r$. Note that this is the only piece of the initial data which is not free, but instead is required to satisfy constraints. First we note that by \eqref{r.BS.improve}, \eqref{eq:reg.r} and Proposition~\ref{prop:C0.ext}, for $V\leq 0$, $r(u_{ext},V)$ is continuous up to $\{V=0\}$, bounded away from $0$ and $(\rd_V r)(u_{ext},V) \in L^2_V$. Moreover, using \eqref{eq:raychv} (and the also estimates \eqref{eq:reg.Om.rdVr} and \eqref{eq:reg.phi}), it can be deduced that $(\rd_V r)(u_{ext},V)$ can be extended continuously up to $\{V=0\}$. Now we extend $r$ and $\rd_V r$ beyond the Cauchy horizon $\{V=0\}$ by solving the equation \eqref{eq:raychv}. Since $\rd_V\phi \in L^2_V$ and $\log \Om$ is bounded (by the choices above), provided that we only solve slightly beyond the Cauchy horizon (i.e.~for $V$ sufficiently small), both $r$ and $|\rd_V r|$ are continuous, bounded above, and $r$ is also bounded away from $0$.
\end{itemize}

\textbf{Formulating the problem as a system of wave equations.} Now apply Proposition~\ref{prop:gen.wave} to solve the following system of wave equations for $\Psi = (r,\log \widetilde{\Om}, Re(\phi), Im(\phi), A_V)$:
\begin{align}
\label{final.wave.r} r\partial_u\partial_V r=&\:-\frac{1}{4}\widetilde{\Omega}^2-\partial_ur \partial_Vr+\mathfrak{m}^2\pi r^2 \widetilde{\Omega}^2 |\phi|^2+\f{r^2}{\widetilde{\Omega}^2}(\rd_uA_V)^2,\\
\label{final.wave.Om} r^2\partial_u\partial_V\log \widetilde{\Omega}=&\:-2\pi r^2(\rd_u\phi \overline{(\rd_V+ i\mfe A_V)\phi}+\overline{\rd_u\phi}(\rd_V+ i\mfe A_V)\phi)\\
\nonumber &\:-2\f{r^2}{\widetilde{\Omega}^2}(\rd_uA_V)^2
+\frac{1}{4}\widetilde{\Omega}^2+\partial_ur\partial_Vr,\\
\label{final.wave.phi} \rd_u ((\rd_V+ i\mfe A_V)\phi)+(\rd_V+ i\mfe A_V)\rd_u\phi=&\:-\frac{1}{2}\mathfrak{m}^2\widetilde{\Omega}^2\phi-2r^{-1}(\partial_ur (\rd_V+ i\mfe A_V)\phi+\partial_V r \rd_u \phi),\\
\label{final.wave.AV} \partial_V\left(\f{r^2}{\widetilde{\Omega}^2}\rd_uA_V\right) = &\:-\pi i r^2\mathfrak{e} (\phi \overline{D_V\phi}-\overline{\phi}D_V\phi).
\end{align}
It is easy to check that this system of equations indeed has the structure as in \eqref{eq:gen.wave}.

\textbf{Solving the system of wave equations.} By Proposition~\ref{prop:gen.wave}, there exists $\ep_0>0$ (independent of $n$) such that for every $V_n$, a unique solution to the above system of equation exists for $(u,V)\in \{(u,V): u\in [u_{ext},u_{ext}+\ep_0),\,V\in [V_n,V_n+\ep_0)\}$. In particular, since $V_n\to 0$, we can choose $n\in \mathbb N$ sufficiently large so that $V_n+\ep_0>0$. Now fix such an $n$ and choose $\ep_{ext}>0$ sufficiently small such that $\ep_{ext}<V_n+\ep_0$. We have therefore constructed a solution $(r,\log \Om, Re(\phi), Im(\phi), A_V)$ to \eqref{final.wave.r}--\eqref{final.wave.AV} in $D_{u_0,v_0}\cup\CH\cup\{(u,V):u\in [u_{ext},u_{ext}+\ep_{ext}],\,V\in [0,\ep_{ext})\}$.

\textbf{Definition of $Q$ and equation \eqref{eq:intconstraint2}.} Define $Q = 2r^2\widetilde{\Om}^{-2}\rd_uA_V$. By definition $Q$ is continuous and \eqref{eq:lastdef.Av} is satisfied classically. Moreover, since \eqref{final.wave.AV} is satisfied in an integrated sense, it also follows that \eqref{eq:intconstraint2} is satisfied.

Plugging in the definition of $Q$ into \eqref{final.wave.r}--\eqref{final.wave.phi}, we also obtain that $r$, $\widetilde{\Om}$ and $\phi$ respectively satisfy \eqref{eq:transpeqr}, \eqref{eq:waveqOmega} and \eqref{eq:phi1} as wave equations in the integrated sense as in Definition~\ref{def:sol}.

\textbf{Propagation of constraints and equations \eqref{eq:intconstraint1}, \eqref{eq:intconstraint3} and \eqref{eq:intconstraint4}.} Next, we check that \eqref{eq:intconstraint1}, \eqref{eq:intconstraint3} and \eqref{eq:intconstraint4} are satisfied. This involves a propagation of constraints argument, which is standard except that we need to be slightly careful about regularity issues.

First, we note that since the equations are satisfied classically at $(u,V_n)$ for all $u\in [u_{ext},u_{ext}+\ep_0)$, \eqref{eq:intconstraint1} and \eqref{eq:intconstraint3} are satisfied on $\{V=V_n\}$. Moreover, by the construction of the data for $r$ above, \eqref{eq:intconstraint4} is also satisfied on $\{u=u_{ext}\}$.

Therefore, it follows that \eqref{eq:intconstraint1}, \eqref{eq:intconstraint3} and \eqref{eq:intconstraint4} are equivalent respectively to the following equations:
\begin{align}
\label{intconstraint1v2}
(Q(u,V)&-Q(u,V_n))-(Q(u_{ext},V)-Q(u_{ext},V_n))= \int_{u_{ext}}^u \left[2\pi i r^2 \mfe (\overline{D_u\phi}-\overline{\phi}D_u\phi)\right] (u',V)\,du'\\ \nonumber
&- \int_{u_{ext}}^u \left[2\pi i r^2 \mfe (\overline{D_u\phi}-\overline{\phi}D_u\phi)\right] (u',V_n)\,du',\\
\label{intconstraint2v2}
(r\partial_u r(u,V)&-r\partial_u r(u,V_n))-(r\partial_u r(u_{ext},V)-r\partial_u r(u_{ext},V_n))\\ \nonumber
=&\: \int_{u_{ext}}^u \left( \left[2 r\partial_u r\partial_u\log \widetilde{\Omega}+(\partial_ur)^2-4\pi r^2|D_u\phi|^2 \right](u',V) - [\cdots](u',V_n)\right)\,du',\\
\label{intconstraint3v2}
(r\partial_V r(u,V)&-r\partial_V r(u_{ext},V))-(r\partial_V r(u,V_n)-r\partial_V r(u_{ext},V_n))\\ \nonumber
=&\: \int_{V_n}^V \left(\left[2 r\partial_V r\partial_V\log \widetilde{\Omega}+(\partial_Vr)^2-4\pi r^2|D_V\phi|^2 \right](u,V')- [\cdots](u_{ext},V')\right)\,dV',
\end{align}
where $[\cdots]$ means that we take exactly the same expression as inside the previous pair of square brackets.

To proceed, observe now that we have the following integrated version of the Leibniz rule: let $f,\,g:[0,T]\to \mathbb R$, $f\in C^0$, $g\in C^1$. Assume that there exists an $F:[0,T] \to \mathbb R$ in $L^1$ such that $f(t)-f(0) = \int_0^t F(s)\,ds$ for all $t\in [0,T]$. Then by Fubini's theorem and the fundamental theorem of calculus,
\begin{equation}
\begin{split}
&\: \int_0^t F(s) g(s) \, ds \\
= &\: g(0) \int_0^t F(s) \,ds + \int_0^t \int_0^s F(s) g'(\tau) \,d\tau\,ds=  f(t) g(0) - f(0) g(0) + \int_0^t \int_\tau^t F(s) g'(\tau) \,ds\,d\tau\\
= &\: f(t) g(0) - f(0) g(0) + \int_0^t [f(t) g'(\tau) - f(\tau) g'(\tau)] \,d\tau=  f(t) g(t) - f(0) g(0) - \int_0^t f(s) g'(s) \,ds.
\end{split}
\end{equation}
In other words, supposed $\Psi_i$ satisfies $\rd_u\rd_v\Psi_i = F_i$ (for some $F_i \in L^1_v C^0_u$), the following integrated versions of the Leibniz rule hold:
\begin{align}
\label{Leibniz1}\partial_u\Psi_i(u,V) \Psi_j (u,V)=&\:\partial_u\Psi_i(u,V_n)\Psi_j (u,V_n)+ \int_{V_n}^V \left[\Psi_j  F_i+  \partial_v\Psi_j  \partial_u\Psi_i\right] (u,V')\,dV',\\
\label{Leibniz2}\partial_v\Psi_i (u,V)\Psi_j (u,V)=&\:\partial_u\Psi_i (u_{ext},V)\Psi_j (u_{ext},V)+ \int_{u_{ext}}^u \left[\Psi_j F_i+  \partial_u\Psi_j  \partial_v\Psi_i\right] (u',V)\,du'.
\end{align}

Let us now show that \eqref{eq:intconstraint1}, or equivalently \eqref{intconstraint1v2}, holds. Since we have already established that \eqref{eq:intconstraint2} holds, it follows that \eqref{intconstraint1v2} is equivalent to
\begin{equation}
\label{intconstraint1v2.5}
\begin{split}
&- \int_{V_n}^V \left([2\pi i r^2 \mfe (\phi\overline{D_v\phi} - \overline{\phi}D_v\phi)](u,V') - [\cdots](u_{ext},V') \right)\, dV' \\
= &\int_{u_{ext}}^u \left( \left[2\pi i r^2 \mfe (\phi\overline{D_u\phi}-\overline{\phi}D_u\phi)\right] (u',V)-[\cdots] (u', V_n)\right)\,du'. \end{split}
\end{equation}
By \eqref{Leibniz1} and \eqref{Leibniz2} above, it follows that we need to check
\begin{equation}\label{intconstraint1v3}
\int_{u_{ext}}^u \int_{V_n}^V \left(\rd_u\left(2\pi i r^2 \mfe (\phi\overline{D_v\phi} - \overline{\phi}D_v\phi)\right) + \partial_V\left(2\pi i r^2 \mfe (\phi\overline{D_u\phi}-\overline{\phi}D_u\phi)\right)\right) (u',V')\,du'dV' = 0,
\end{equation}
where expressions such as $\rd_u D_v\phi$ and $\rd_V D_u\phi$ are to be understood \emph{after plugging in the appropriate inhomogeneous terms arising from \eqref{final.wave.phi}.} On the other hand, after plugging in the appropriate expressions from \eqref{final.wave.phi}, it is easy to check that the integrand in \eqref{intconstraint1v3} vanishes almost everywhere. Therefore, \eqref{intconstraint1v3} indeed holds, which then implies that \eqref{eq:intconstraint1} holds.

Next, we consider \eqref{eq:intconstraint3}, or equivalently \eqref{intconstraint2v2}. Since we have already established \eqref{final.wave.r} in an integrated sense, using the definition of $Q$ above, it follows from \eqref{Leibniz1} that \eqref{intconstraint2v2} is equivalent to
\begin{equation}\label{intconstraint2v2.5}
\begin{split}
&\: \int_{V_n}^V \left(\left[-\frac{1}{4}\widetilde{\Omega}^2+\mathfrak{m}^2\pi r^2 \widetilde{\Omega}^2 |\phi|^2+\frac{1}{4}\f{\widetilde{\Om}^2}{r^2}Q^2\right](u,V') - [\cdots](u_{ext},V')\right)\,dV' \\
=&\: \int_{u_{ext}}^u \left( \left[2 r\partial_u r\partial_u\log \widetilde{\Omega}+(\partial_ur)^2-4\pi r^2|D_u\phi|^2 \right](u',V) - [\cdots](u',V_n)\right)\,du'
\end{split}
\end{equation}
Using again the integrated Leibniz's rule \eqref{Leibniz1} and \eqref{Leibniz2}, it then follows that \eqref{intconstraint2v2.5} is equivalent to
\begin{equation}\label{intconstraint2v3}
\begin{split}
&\: \int_{u_{ext}}^u \int_{V_n}^V \rd_u \left(\left[-\frac{1}{4}\widetilde{\Omega}^2+\mathfrak{m}^2\pi r^2 \widetilde{\Omega}^2 |\phi|^2+\frac{1}{4}\f{\widetilde{\Omega}^2}{r^2}Q^2\right](u',V') \right)\,dV'\,du' \\
-&\: \int_{V_n}^{V}\int_{u_{ext}}^u \rd_V\left( \left[2 r\partial_u r\partial_u\log \widetilde{\Omega}+(\partial_ur)^2-4\pi r^2|D_u\phi|^2 \right](u',V')\right)\,du'\,dV' = 0, 
\end{split}
\end{equation}
where (in a similar manner as \eqref{intconstraint1v3}) expressions $\rd_V\rd_u r$, $\rd_V\rd_u\log\widetilde{\Omega}$ and $\rd_VD_u\phi$ are to be understood \emph{after plugging in the appropriate inhomogeneous terms arising from \eqref{final.wave.r}, \eqref{final.wave.Om} and \eqref{final.wave.phi} respectively}, and $\rd_uQ$ is to be understood as $\rd_uQ = 2\pi i r^2 \mfe(\phi \overline{D_u\phi}-\overline{\phi}D_u\phi)$ (cf.~\eqref{eq:intconstraint1}). Direct algebraic manipulations (using in particular $Q=2r^2 \widetilde{\Omega}^{-2} \partial_uA_V$) then show that the integrand in \eqref{intconstraint2v3} vanishes almost everywhere. This verifies \eqref{eq:intconstraint3}.

Finally, we need to check \eqref{eq:intconstraint4}, or equivalently \eqref{intconstraint3v2}. This can be argued in a very similar manner as \eqref{eq:intconstraint3}; we omit the details.

\textbf{Checking the regularity of the functions.} We have now checked that all the equations are appropriately satisfied. To conclude that we have a solution in the sense of \eqref{def:strong.solutions}, it remains to check that $\rd_V r$ is continuous. (A priori, using Proposition~\ref{prop:gen.wave}, we only know that $\rd_V r \in L^2_V(C^0_u)$.) That $\rd_{V} r$ is continuous is an immediate consequence of \eqref{intconstraint3v2}, the fact that the data for $\rd_V r$ are continuous on $\{u=u_{ext}\}$, and the regularity properties of all the other functions.

We have thus shown how to construct one extension of the solution (as a spherically symmetric strong solution in the sense of Definition~\ref{def:strong.solutions}). Since the procedure involves prescribing arbitrary data, one concludes that in fact there are infinitely many inequivalent extensions. \qedhere

\end{proof}

\begin{remark}\label{rmk:L1sol}
Notice that in spherical symmetry, one can solve the wave equations with data such that one only requires $\rd_V\phi,\,\rd_V r,\,\rd_v\log\widetilde{\Om} \in L^1_V$. However, if $\rd_V\phi\notin L^2_V$, we have $\rd_V r\to -\infty$ along a constant $u$ hypersurface, and one cannot make sense of \eqref{eq:raychv} beyond the singularity. In other words, if $\rd_V\phi\notin L^2_V$, we cannot find appropriate data to the system of the wave equations so as to guarantee that the solution indeed corresponds to a solution to \eqref{EMCSFS}.
\end{remark}

\section{Improved estimates for massless and chargeless scalar field}\label{sec:imp.est}

\begin{proof}[Proof of Theorem~\ref{thm:Lipschitz}]
We will prove that
$$\sup_{u\in (-\infty,u_0],\, v\in [v_0,\infty)}(|u|^2|\rd_u\phi|(u,v)+v^2|\rd_v\phi|(u,v))<\infty.$$
Recalling the relation between $v$ and the regular coordinate $V$ in \eqref{Vvrelation}, this then implies the desired conclusion.

We prove the above bounds with a bootstrap argument. Assume that
\begin{equation}\label{imp.BA}
\sup_{u\in (-\infty,u_0],\, v\in [v_0,\infty)}|u|^2|\rd_u\phi|(u,v)\leq \mathcal A_{imp}.
\end{equation}
In the following argument, we will allow the implicit constant in $\ls$ to depend on all the constants in the previous sections, as well as the size of the LHS of \eqref{eq:extra.assump}. $\mathcal A_{imp}$ will then be thought of as larger than all these constants. We will show that for appropriate $|u_0|$, the estimate in \eqref{imp.BA} can be improved.

To proceed, note that when $\mfm = \mfe = 0$, \eqref{eq:phi1} can be written as
\begin{equation}\label{eq:wave.rewrite.1}
\rd_u(r\rd_v\phi) = -(\rd_vr)(\rd_u\phi) 
\end{equation}
and
\begin{equation}\label{eq:wave.rewrite.2}
\rd_v(r\rd_u\phi) = -(\rd_ur)(\rd_v\phi).
\end{equation}
Using \eqref{eq:wave.rewrite.1}, we estimate 
\begin{equation}\label{rdvphi.imp.est}
v^2|\rd_v\phi|(u,v)\ls 1+ \mathcal A_{imp} \int_{-\infty}^u v^2|u'|^{-2}(v+|u'|)^{-2}\, du'\ls 1+\mathcal A_{imp} |u_0|^{-1}.
\end{equation}
Using \eqref{eq:wave.rewrite.2} and the estimate \eqref{rdvphi.imp.est} that we just established, we have
\begin{equation}\label{rduphi.imp.est}
|u|^2|\rd_u\phi|(u,v)\ls 1+ (1+\mathcal A_{imp}|u_0|^{-1}) \int_{v_0}^{\infty} |u|^2|v'|^{-2}(v'+|u|)^{-2}\, dv'\ls 1+(1+\mathcal A_{imp}|u_0|^{-1}) v_0^{-1}.
\end{equation}
Choosing $\mathcal A_{imp}$ sufficiently large and $u_0$ sufficiently negative (in that order), we have improved the bootstrap assumption \eqref{imp.BA}. Then by \eqref{rdvphi.imp.est} and \eqref{rduphi.imp.est}, 
$$\sup_{u,V}(|\rd_V\phi|(u,V) + |u|^2|\rd_u\phi|(u,V)) \ls \sup_{u,v}(v^2|\rd_v\phi|(u,v) + |u|^2|\rd_u\phi|(u,v))<\infty,$$
from which the conclusion follows.
\end{proof}

\bibliographystyle{hplain}
\bibliography{Extremal}
\end{document}